\pdfoutput=1 
\documentclass[a4, 10pt]{article}

\usepackage{Commands}

\title{Trajectories and Platoon-forming Algorithm\\for Intersections with Heterogeneous Autonomous Traffic}

\author{P.C.~Joshi, M.A.A.~Boon and~S.C.~Borst}

\date{}

\begin{document}
\maketitle
\blfootnote{\includegraphics[height=0.8\baselineskip]{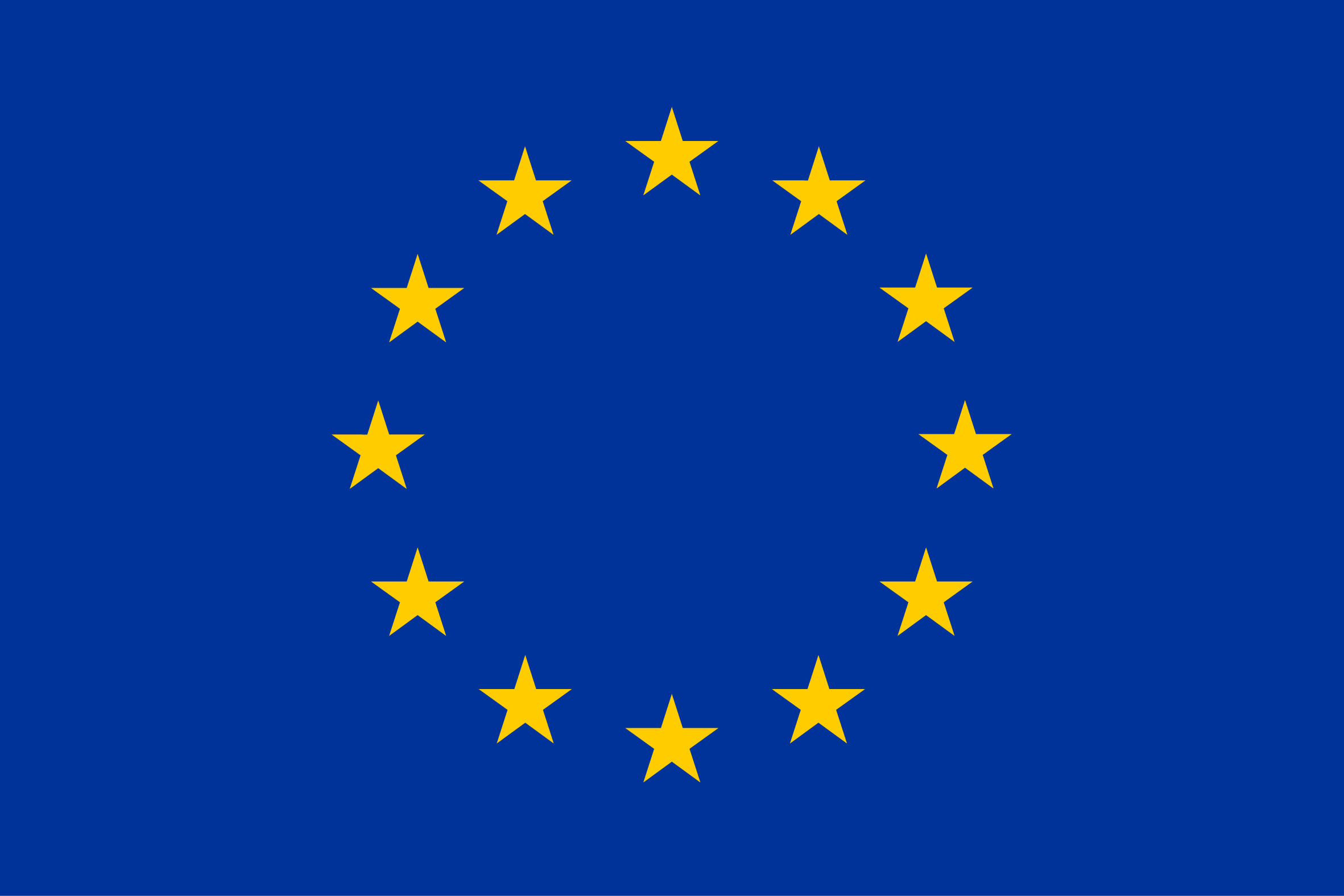} : This research was supported by the European Union’s Horizon 2020 research and innovation programme under the Marie Skłodowska-Curie grant agreement no.~945045, and by the NWO Gravitation project NETWORKS under grant no.~024.002.003.}
\blfootnote{The authors are with the Department of Mathematics \& Computer Science, Eindhoven University of Technology, P.O. Box 513, 5600 MB Eindhoven, The Netherlands (Email: \href{mailto:p.c.joshi@tue.nl}{p.c.joshi@tue.nl}, \href{mailto:m.a.a.boon@tue.nl}{m.a.a.boon@tue.nl}, \href{mailto:s.c.borst@tue.nl}{s.c.borst@tue.nl}).}
\vspace*{-1cm}

\begin{abstract}
    The anticipated launch of fully autonomous vehicles presents an opportunity to develop and implement novel traffic management systems. Intersections are one of the bottlenecks for urban traffic, and thus offer tremendous potential for performance improvements of traffic flow if managed efficiently. Platoon-forming algorithms, in which vehicles are grouped together with short inter-vehicular distances just before arriving at an intersection at high speed, seem particularly promising in this aspect. In this work, we present an intersection access control system based on platoon-forming for heterogeneous autonomous traffic. The heterogeneity of traffic arises from vehicles with different acceleration capabilities and safety constraints. We focus on obtaining computationally fast and interpretable closed-form expressions for safe and efficient vehicle trajectories that lead to platoon formation, and show that these trajectories are solutions to certain classes of optimisation problems. Additionally, we conduct a numerical study to obtain approximations for intersection capacity as a result of such platoon formation. 
\end{abstract}

\section{Introduction}
Fully autonomous vehicles are touted to revolutionise future road traffic operations, and provide an opportunity to improve the efficiency and performance of current traffic flow management systems. Some first working prototypes for autonomous trucking have already been launched \cite{Waabi2022IntroducingDriver}, signalling that this future might be closer than previously thought. Increasing flow capacity and reducing traffic congestion is then especially important to minimise travel time delays in the vicinity of natural bottlenecks such as urban intersections \cite{Lioris2016DoublingPlatooning}.

The traffic management ecosystem at intersections is currently largely comprised of traffic lights. They are used to control the flow of traffic in an efficient manner, under the restrictions of the present technology. With advances in automotive technology, however, new strategies for managing traffic flows can in fact prove to be more efficient. One such class of strategies is platoon-forming algorithms (or platooning). At its core, a platoon-forming algorithm relies on vehicular communication techniques to arrange vehicles in platoons with extremely short inter-vehicular distances. A platoon then travels forth as a single entity, and increases the road capacity since more vehicles can now fit in the same amount of space, while maintaining high speed. Autonomous vehicles are natural candidates to implement platooning because of their communication capabilities.

Platooning has been widely studied in the context of highways. For a survey of the literature on highway platooning, we refer to \textcite{Kavathekar2011VehicleCategorization}. On highways, vehicles travel at higher speeds, and if platooned, will typically experience reduced air resistance. This effect is even more pronounced with heavy-duty vehicles such as trucks, leading to several pilot projects on highway truck platooning \cite{Mcauliffe2017Fuel-EconomySystem, Deutschle2010UseMotorways, Tsugawa2013AnProject, Daems2022PlatooningAnalysis}. Results reveal significant energy savings \cite{Tsugawa2016ASavings, Mcauliffe2017Fuel-EconomySystem, Tsugawa2013AnProject, Daems2022PlatooningAnalysis}, reduction in emissions \cite{Tsugawa2013AnProject}, increase in road capacity \cite{Deutschle2010UseMotorways, Daems2022PlatooningAnalysis} and avoidance of dangerous situations \cite{Daems2022PlatooningAnalysis}. 

More recently, platoon-forming algorithms for fully autonomous vehicles have also been shown to be useful in managing intersections, providing benefits such as an increase in capacity \cite{Lioris2016DoublingPlatooning}, reducing travel times \cite{Bashiri2017AVehicles, Bashiri2018, Jin2013} as well as fuel emissions \cite{Jin2013, Kumaravel2022}. The conflicting paths of the traffic flows meeting at an intersection, however, means that process of platoon formation and crossing of vehicles needs to be scheduled carefully to maintain safety while simultaneously improving efficiency. Several approaches such as multi-agent systems \cite{Jin2013}, reinforcement learning \cite{Prathiba2021}, virtual platooning \cite{Medina2018}, model predictive control \cite{Wang2020ACAVPlatoon} and mixed-integer programming \cite{Lu2022OptimizationBased} have been proposed in the literature for tackling this challenge. Such approaches discuss and combine two aspects -- the scheduling of safe crossing times for vehicles, and trajectory planning to enforce these schedules and satisfy safety constraints. The authors in \cite{Bashiri2017AVehicles, Bashiri2018, Tachet2016RevisitingSystems} focus exclusively on approaches to tackle the scheduling problem, whereas \textcite{Teng2023MotionPlanning} survey motion planning techniques for autonomous vehicles. 

\textcite{Tachet2016RevisitingSystems} propose algorithms that schedule intersection access times for vehicles by forming platoons with a maximum batch size, while also focusing on safety of vehicles. \textcite{Miculescu2020Polling-Systems-BasedSignals} use polling models from queueing theory to schedule safe arrival times of vehicles at an intersection, in platoons. Additionally, the authors determine safe vehicle trajectories in order to facilitate formation of platoons, by solving linear optimisation problems. \textcite{Timmerman2021PlatoonIntersections} show that the linear optimisation problems used to generate vehicle trajectories have closed-form solutions. They present these closed-form trajectory expressions, which reduce computation time and provide explicit insight into the nature of the trajectories. Further, they propose a new optimisation objective for vehicle trajectories focusing on passenger comfort. \textcite{Kumaravel2022} develop a two-level optimal framework consisting of jobshop scheduling to schedule crossing times for platoons (first level) and optimisation problems for vehicle trajectories (second level). Analytical solutions for the trajectory optimisation problems are presented as well. 

Although the literature on platoon-forming algorithms has been rapidly growing, there remain significant challenges still. These include making the algorithms more flexible and suitable for realistic scenarios. In this paper we address one of the major limitations which the current algorithms do not account for - traffic heterogeneity, while maintaining efficiency and safety. The different characteristics of vehicles play a more significant role in urban areas as compared to highways, for instance. This is due to the frequency of accelerating and decelerating manoeuvres necessary in urban areas, especially at intersections.

With heterogeneity in traffic, ensuring safety of vehicles also becomes a more complex issue. While this could still be handled by for e.g.\ suitably extending the optimisation procedure proposed in \cite{Miculescu2020Polling-Systems-BasedSignals}, it is important to obtain solutions that can be computed fast, as in \cite{Timmerman2021PlatoonIntersections}, which is now more challenging due to the differing characteristics of vehicles. Thus, our goal is to come up with an efficient approach for controlling access at an unsignalised intersection, taking the heterogeneity of autonomous vehicular traffic into account, subject to safety constraints.

In this work, we use a platooning framework for the intersection access control system \footnote{3D animations of our control system are available as ancillary files, along with supporting text in Appendix \ref{app:animation}.}, consisting of two parts - a platoon-forming algorithm and a speed-profiling algorithm, as proposed by \cite{Miculescu2020Polling-Systems-BasedSignals, Timmerman2021PlatoonIntersections}, which we extend to account for heterogeneity in traffic. The platoon-forming algorithm uses the earliest arrival times of incoming vehicles to the intersection to compute the earliest possible safe crossing times, through a suitably adapted polling model from queueing theory. Using these safe vehicle crossing times as input, the speed-profiling algorithm generates vehicle trajectories that enable platoon formation by the time the vehicles reach the intersection, while maintaining safety and efficiency.

For the speed-profiling algorithm we choose an optimality criterion that helps to achieve the potential increase in capacity offered by the platoon-forming algorithm. Vehicle trajectories are then obtained as the solution to certain classes of optimisation problems where trajectories are optimised simultaneously. We demonstrate that it is possible to obtain these expressions for the trajectories in closed-form, and in fact show that these trajectories are optimal in a stronger sense as well. This is followed by a numerical study to approximate the maximum capacity of the intersection via the access control system, revealing an interplay between queueing primitives and the fundamental parameters of the speed-profiling algorithm which affects the stability of the entire system. 

With our novel analysis for management of an unsignalised intersection with heterogeneous autonomous traffic, the contribution to the literature is four-fold. Firstly, we generalise the platoon-forming algorithms from \cite{Miculescu2020Polling-Systems-BasedSignals, Timmerman2021PlatoonIntersections} to accommodate the more complex safety requirements that arise with heterogeneous traffic. Secondly, we provide an alternative formulation of linear optimisation problems which can be used to simultaneously determine the optimal trajectories for all vehicles in a platoon, and show that solutions to such problems are also optimal for the multi-step process of obtaining the trajectory of each vehicle individually, as proposed in \cite{Miculescu2020Polling-Systems-BasedSignals}. Thirdly, we obtain closed-form trajectory solutions to the previously mentioned optimisation problems, and analytically show that the trajectories are actually optimal in a much stronger sense. These closed-form expressions can be computed very efficiently, and they help obtain better insights into the nature and feasibility of these optimal trajectories. Lastly, we present a numerical method to analyse the capacity of an intersection controlled by such a mechanism, and uncover how the parameters of the speed-profiling algorithm also affect the overall stability of the system.

The rest of the paper is organised as follows: in Section~\ref{sec:framework}, we outline the framework for the access control system and explain the inter-connection between the platoon-forming algorithm and the speed-profiling algorithm. Section~\ref{sec:arrival_scheduler} deals with the platoon-forming algorithm, while the speed-profiling algorithm is covered in Section~\ref{sec:SPA}. The discussion on the integration of the two algorithms is presented in Section~\ref{sec:analysis}, along with the numerical study on the capacity of the intersection under this access control system.

\section{Intersection Access Control System}\label{sec:framework}
\noindent As stated before, in this paper, we develop a novel framework for an intersection access control system with multiple vehicle types. In our framework, building upon the framework introduced in \cite{Timmerman2021PlatoonIntersections}, we assume that the vehicular traffic in an urban area consists of two kinds of autonomous vehicles, say, cars and trucks. We assume that both cars and trucks can attain the same maximum allowed speed $v_{\textit{max}}$, because the restriction on speeds of a vehicle is not due to its capabilities but rather due to the existing legal speed limits. On the other hand, maximum acceleration and/or deceleration is significantly different depending on the type of the vehicles - cars have a significantly lower weight-to-power ratio than trucks and therefore, can accelerate and decelerate at greater magnitude than trucks, which is reflected in the assumption below: 
\begin{assumption}\label{assump:acc_relation}
The maximum acceleration of trucks is lower than that of cars, that is:
\begin{equation}\label{eq:acc_relation_car_truck}
    0 < a_{\textit{max},tr} < a_{\textit{max},c},
\end{equation}
where $a_{\textit{max},tr}$ and $a_{\textit{max},c}$ are the maximum feasible accelerations for a truck and car, respectively. The maximum feasible deceleration for cars and trucks is assumed to be $-a_{\textit{max},c}$ and $-a_{\textit{max},tr}$ respectively.
\end{assumption}

\begin{remark}
The assumption that the maximum feasible decelerations are equal in magnitude to the maximum feasible accelerations is made for convenience and simplification of notation - the framework can be suitably extended to accommodate different magnitudes for the maximum acceleration and deceleration. In fact, the ideas in this paper can easily be extended to the case with more than two vehicle types for the platoon-forming algorithm (Section~\ref{sec:arrival_scheduler}) and the joint optimisation procedure (Section~\ref{sec:SPA_optimisation}).
\end{remark}

The goal is to develop an intersection access mechanism which can dictate how vehicles from various traffic flows access the intersection. As we observed in the introduction, the intersection is a resource shared between conflicting traffic flows, and hence it is important to ensure that safety is maintained while optimising efficiency. 
\begin{figure}[!t]
    \centering
    \includegraphics[width = 0.75\linewidth]{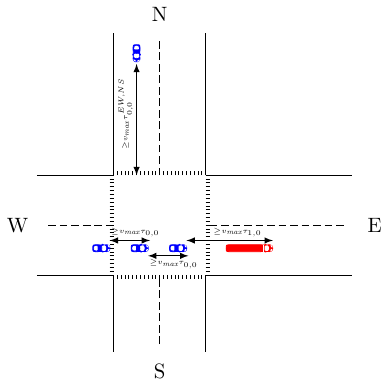}
    \caption{Capturing safety conditions for a two-lane intersection. Type $0$ vehicles are cars and type $1$ vehicles are trucks.}
    \label{fig:intersection_schematic}
\end{figure}

Safety constraints are two-fold, to account for dynamics within each lane and also across lanes. For two vehicles in a lane, the corresponding time separation is comprised of two components - the length of the first vehicle, and the actual time headway, which is calculated considering the rear end of the first vehicle and the front tip of the second. It is also important to ensure safe conditions between certain vehicles from different lanes. Here, the time separation depends on the geometry of the intersection as well as the characteristics of the vehicles involved. Appendix~\ref{app:time_separations} contains expressions for the time separation $\tau^{\ell_i, \ell_j}_{i,j}$ between two vehicles: $i$ travelling in lane $\ell_i$ and $j$ travelling in lane $\ell_j$, which can be different from lane $\ell_i$.  Then, safety can be defined as:

\begin{definition}[Safety]\label{def:safety}
A pair of vehicles -- vehicle A of type $i$ travelling in lane $\ell_i$ followed by vehicle B of type $j$ travelling in the same lane -- is safe at time $s$, if the distance between the front tips of both vehicles is at least $v_{\textit{max}} \tau^{\ell_i,\ell_i}_{i,j}$, where $\tau^{\ell_i,\ell_i}_{i,j}$ is the corresponding minimum time separation. That is, if the position functions (w.r.t.\ time) of the bumpers of two consecutive vehicles are denoted by $x_A$ and $x_B$ respectively, then the vehicles are safe at time $s$ if:
\begin{equation*}
    |x_A(s) - x_B(s)| \geq v_{\textit{max}} \tau^{\ell_i,\ell_j}_{i,j}.
\end{equation*}
If the vehicles are travelling in separate lanes, then the distance between the front tip of vehicle B (travelling in lane $\ell_j$) and the beginning of the intersection should be at least $v_{\textit{max}} \tau^{\ell_i,\ell_j}_{i,j}$, when vehicle A is entering the intersection area, as demonstrated in Figure \ref{fig:intersection_schematic} via the lone car on the NS lane.
\end{definition}
A visual interpretation of the safety constraints is provided in Figure \ref{fig:intersection_schematic}. From now on, we will omit the word `bumper', and assume that the position function of a vehicle refers to the position function of the tip of its bumper. Within a lane, the time separation depends on the characteristics of the corresponding vehicles. Thus, for ease of notation, we will drop the superscripts when referring to a pair of vehicles in the same lane, and denote the time separation by $\tau_{i,j}$.

In this paper, we will be looking at trajectories of vehicles, which represent paths taken by vehicles throughout given time intervals. Or, more accurately, a trajectory represents the function of position of a vehicle w.r.t.\ time for a certain time interval. We will discuss trajectories in greater detail in Section~\ref{sec:SPA}.

\paragraph{Possible implementation considerations}\leavevmode

Before we proceed, we briefly indicate how our proposed mechanism might be implemented in practice. It should be noted that the following discussion is on a conceptual basis -- a detailed analysis of the implementation is beyond the scope of this paper. We consider a setting where all of the vehicles are autonomous, and these would be Level 4 or 5 autonomous vehicles, based on the 6-level classification by \cite{2021SurfaceVehicles}. It is assumed that a central controller receives information about the imminent arrivals of these vehicles at some time before they actually reach the intersection, by utilising vehicle communication techniques. The functions of this central controller could possibly be served by Unmanned Aerial Vehicles (UAVs), as proposed by \textcite{Bouassida2023}. Vehicles have no cause for slowing down when they are still far away from the intersection area, and hence the vehicles are travelling at speed $v_{\textit{max}}$. Thus, the central controller has a dynamic list consisting of the expected time of arrival for each vehicle, along with its type, that is, the time at which each vehicle would reach the intersection area if it continued at speed $v_{\textit{max}}$.

With this dynamic list, a scheduling algorithm employed by the central controller (platoon-forming algorithm) computes a list of safe crossing times for each vehicle, that is, times at which it is safe for the vehicle to enter the intersection area without conflict. We will see in Section~\ref{sec:arrival_scheduler} that this can be achieved with the help of a polling model from queueing theory. With this list of safe crossing times, a speed-profiling algorithm determines trajectories for all the vehicles, so that they can be present at the intersection at their scheduled crossing time, at maximum speed $v_{\textit{max}}$, as constrained. The specific trajectory can then be communicated by the central controller to each vehicle. Note that with the constraint on vehicle trajectories, there is now no room for error - the trajectories have to be determined perfectly in order to maintain safety.

\section{Platoon-forming Algorithm}\label{sec:arrival_scheduler}
\noindent In this section we discuss the mechanism behind the platoon-forming algorithm. In principle, the platoon-forming algorithm takes as input expected arrival times of vehicles at an intersection, and generates a safe schedule of crossing times, i.e., a list of earliest possible times for each vehicle at which it is safe to enter the intersection area. This is achieved by describing the queueing dynamics near an intersection in terms of a polling model.

Polling models are broadly used to analyse situations in which a single server is responsible for serving several queues. Queues are formed by customers/jobs arriving at random times, all requiring service from the server. As soon as its service is completed, a job leaves the system and the server starts serving the next job. The server can switch from serving jobs from one queue to another based on its service discipline; commonly analysed service disciplines include exhaustive, gated, $k$-limited, Global First-Come First-Served and random. However, during the switching process, an additional amount of time, known as the switchover time, might be necessary before the server can serve the next queue. In a standard polling model, jobs arrive to queues according to Poisson processes, and inter-arrival and services times are i.i.d random variables within each of the queues. A comprehensive overview of polling models and their analyses is provided by \textcite{Takagi2000AnalysisModels} and \textcite{Borst2018Polling:Perspective}.

At an intersection, the fact that several road traffic flows have to share the same space (and also that conflicting flows cannot access the intersection space at the same time) makes polling systems a natural way to model this situation, as explained by \textcite{Boon2011ApplicationsSystems}. The server of the polling model represents access to the intersection in our framework, and each lane $\ell \in \{1, \cdots, n\}$ is represented by a queue of its own. Each vehicle forms a `job' or a `customer' in a queue. When the server in the polling model starts serving a particular customer, the corresponding vehicle starts crossing the intersection area. The arrival of a vehicle in a lane in the polling model is its time of arrival at the intersection, had it been the only vehicle in the system, that is, if it experienced no delay due to other vehicles accessing the intersection area. The service time of a vehicle is modelled to be exactly the time needed to pass before the next vehicle can safely enter the intersection area. In other words, a vehicle enters the intersection area at the beginning of its service period, and at the end of the service period, the next vehicle from the same lane enters the intersection area, if the service discipline allows it.

The capacity of the intersection is inversely proportional to the mean service time, hence it is desirable to have quick service. Of course, our safety constraints dictate that the minimum distance between two consecutive vehicles of types $i$ and $j$ in the same lane must be at least $v_{\textit{max}} \tau_{i,j}$, where $\tau_{i,j}$ is the minimum time separation between them. From the time it enters the intersection area onward, every vehicle is assumed to possess speed $v_{\textit{max}}$. Hence the service time of a vehicle of type $i$ will then be $\tau_{i,j}$, where $j$ is the type of the following vehicle. Within each queue, vehicles are served on a First-Come, First-Served basis. The waiting time of a vehicle in the queueing model is equal to the delay incurred by the vehicle. The mean delay incurred by a vehicle in a platoon-forming algorithm has been analysed in \cite{Timmerman2021PlatoonIntersections}. The delay or the waiting time is an important quantity while determining the nature of trajectories, as we will see later in Section~\ref{sec:SPA_closed_form}. Finally, clearance times, which are the times when the access to the intersection is being switched from one lane to another, are represented by switchover times in the polling model. Table~\ref{tab:polling_model} summarises the discussion above, to provide a link between the polling model and our traffic system. 

\begin{table}[!t]
\centering
\begin{tabular}{|c|c|}
    \hline
    \textbf{Polling Model}  &  \textbf{Traffic System} \\
    \hline
        Queues & Lanes \\
        Server & Intersection controller \\
        Jobs/Customers & Vehicles \\
        Arrival time & Free-flow arrival time at intersection\\
        Service time & Time headway\\
        & (depends on current and following vehicle)\\
         Waiting time & Delay\\
         Switchover times & Clearance times\\
         \hline
    \end{tabular}
    \caption{Correspondence between the polling model and our traffic system.}
    \label{tab:polling_model}
\end{table}

Although we have specified how elements from the polling model relate to our system, we have not yet discussed the service discipline, which dictates how the server switches between queues in a polling model. In the traffic system, the service discipline will determine when access to the intersection is cut off for one flow and opened up to the next flow. On the one hand, it is globally fair to let vehicles access the intersection in order of their arrival across all lanes. However, if lanes are switched frequently, the clearance/switchover times add up and the rate of vehicles crossing the intersection per unit time decreases, thus compromising on efficiency. This was also observed by \textcite{Tachet2016RevisitingSystems}, when comparing the mean delay between the global First-Come First-Served service discipline (FAIR) with a $k$-limited service discipline (BATCH). On the other hand, in an exhaustive service discipline, the server continues serving a queue until it empties, and only then moves on to the next queue. \textcite{Timmerman2021PlatoonIntersections} performed a three-way comparison of the mean delay experienced by vehicles by focusing on the BATCH strategy, gated and exhaustive service discipline, and the exhaustive service discipline was observed to have the lowest mean delay. Additionally, as we will see in Section~\ref{sec:SPA}, the exhaustive service discipline leads to simpler trajectories with at most one stop, thus possibly reducing computational complexity, leading us to focus on exhaustive service. 

An interesting point of difference between a traditional polling model and the polling model used here arises during the switchover process. In a traditional polling model, the server visits all the queues in a cyclic order, even if they are empty. In this version of the polling model, however, the server skips the next queue in the cyclic order during the switchover process if it is empty, and visits the next non-empty queue in order. This saves on switchover times and thus reduces delays.

\begin{algorithm}[!t]
\caption{Pseudocode for the main function in the platoon-forming algorithm.}\label{alg:polling_sim}

\hspace*{\algorithmicindent} \textbf{Input:}
\begin{itemize}
    \item For each lane $\ell$:
\begin{itemize}
    \item Free-flow vehicle arrival times $t_{a}[\ell] = [\cdots, t_a[\ell][i] \cdots] $
    \item Types of arriving vehicles $\Psi[\ell] = [\cdots, \Psi[\ell][i], \cdots]$
\end{itemize}
    \item Service times between vehicles in the same lane $\tau_{i,j}$ (same lane) and $\tau^{\ell_i, \ell_j}_{i,j}$ (when changing lanes)
\end{itemize}

\hspace*{\algorithmicindent} \textbf{Output:} Intersection crossing times $t_{f}[\ell] = [\cdots,t_{f}[\ell][i], \cdots]$ of vehicles arriving in lane $\ell$
\begin{algorithmic}[1]
    \Function{Simulate}{$T$} \Comment{$T$ is the time horizon}
        \State $FutureEventSet \gets \{Event(\mathcal{A}, \Psi[\ell][i], \ell, t_{a}[\ell][i]) \ \forall i, \ell \} $ \Comment{Upcoming events}
        \State $Queues \gets [\ \underbrace{[\ ], \cdots, [\ ]}_{n \text{ many}}\ ] $ \Comment{Vehicles awaiting service}
        \State $t \gets 0$ \Comment{Current time}
        \State $t_f \gets [\ \underbrace{[\ ], \cdots, [\ ]}_{n \text{ many}}\ ]$ \Comment{Crossing times of vehicles}
        \State $count \gets [\ \underbrace{0, \cdots, 0}_{n \text{ many}}\ ]$ \Comment{Number of vehicles served}
        \Statex    
        \While{$t < T$}
            \State $e \gets FutureEventSet.next()$ \Comment{Earliest upcoming event}
            \State $t,\ \psi,\ j \gets e.time,\ e.vehicleType,\ e.lane$
            \If{ $e.kind$ is $\mathcal{A}$}
                \State $v \gets Vehicle(t, \psi)$
                \State $Queues[j].\text{append}(v)$ \Comment{Add vehicle to queue}
            \Else \Comment{$e$ is a Beginning-of-Service event}
                \State $count[j] \gets count[j] + 1$
                \State $v \gets Queues[j].\text{pop}()$ \Comment{First vehicle in queue $j$ gets service}
                \State $t_{f}[j].\text{append}(t)$ \Comment{Crossing time of vehicle $v$}
                \If{$\textsc{Length}(Queues[j]) \neq 0$ } \Comment{ Can add another vehicle to the platoon}
                    \State $v_{new} \gets Queues[j].\text{pop}()$
                    \State $FutureEventSet.\text{append}(Event(\mathcal{BS},v_{new}.type, j, t + \tau_{\psi, v_{new}.type}))$
                    \State $Queues[j].\text{appendleft}(v_{new})$ \Comment{Add back to the head of the queue}
                \ElsIf{$t_a[j][count[j]]
                \leq t + \tau_{\psi, \Psi[j][count[j]]}$} \Comment{Can extend platoon}
                    \State $FutureEventSet.\text{append}(Event(\mathcal{BS},\Psi[j][count[j]], j,t + \tau_{\psi, \Psi[j][count[j]]}))$   
                \Else 
                    \If{$\textsc{Length}(Queues[i]) = 0 \ \forall i$} \Comment{No vehicles awaiting service in system}
                        \State $ nextQueue \gets \argmin\limits_{\ell } \left(\max(t_a[\ell][count[\ell]], t + \tau^{j, \ell}_{\psi, \Psi[\ell][count[\ell]]}) \right) $
                        \State $\psi_{new} \gets \Psi[nextQueue][count[nextQueue]]$
                        \State $t_1 \gets \max(t_a[nextQueue][count[nextQueue]], t + \tau^{j, nextQueue}_{\psi, \psi_{new}})$
                        \State $FutureEventSet.\text{append}(Event(\mathcal{BS},\psi_{new}, nextQueue, t_1))$
                    \Else \Comment{A queue exists with vehicles waiting for service}
                        \State $k \gets \textsc{Next\_NonEmpty\_Queue}(Queues, j, t, \Psi, t_a)$
                        \State $v_{new} \gets Queues[k].\text{pop}()$
                        \State $FutureEventSet.\text{append}(Event(\mathcal{BS},v_{new}.type, k, t + \tau^{j, k}_{\psi, v_{new}.type}))$
                        \State $Queues[k].\text{appendleft}(v_{new})$
                    \EndIf
                \EndIf
            \EndIf
        \EndWhile
        \State \Return $t_f$
    \EndFunction
\end{algorithmic}
\end{algorithm} 

\paragraph{Implementation}\leavevmode

The platoon-forming algorithm is implemented by means of two categories of objects -- Vehicles and Events. Associated with each vehicle are two Events -- an Arrival ($\mathcal{A}$) Event, and a Beginning-of-Service ($\mathcal{BS}$) Event. A Vehicle object also possesses characteristics such as time of arrival to the intersection in free flow and type (car/truck), whereas an Event is identified by the time of its occurrence, its kind ($\mathcal{A}$/$\mathcal{BS}$), along with the vehicle type (car/truck) and the lane. The function \textsc{Simulate} described in Algorithm~\ref{alg:polling_sim} yields the beginning-of-service times for each vehicle, which are the earliest possible times at which a vehicle can safely enter the intersection area, as discussed before. Since the beginning-of-service times for two consecutive vehicles are separated by at least the appropriate time headway, the schedule is safe for vehicles from the moment they enter the intersection area onward.  

\section{Speed-profiling algorithm} \label{sec:SPA}

In this section, we turn our attention to the other half of the framework. Specifically, we now have the task of ensuring that the crossing times for each lane $\ell$ as received from the platoon-forming algorithm have feasible trajectories associated with them, leading up to the intersection. Such algorithms are named as speed-profiling algorithms in \cite{Timmerman2021PlatoonIntersections}. It is worth observing that the nature of the arrival process and service requirements, as well the exhaustive service discipline, all together imply that the waiting times of vehicles from the head to the tail in a platoon are non-increasing - for a pair of vehicles in a platoon, the difference between their crossing times is exactly the corresponding time separation, but the associated inter-arrival time is equal to at least the same time separation service time. If not handled carefully, this could lead to inefficient trajectories at best and unsafe situations at worst, thus motivating the need for a speed-profiling algorithm. Such a speed-profiling algorithm would take in as input the arrival and crossing times from the platoon-forming algorithm, and then produces trajectories satisfying the following conditions:
\begin{itemize}\label{list:trajectory_conditions}
    \item The maximum allowed speed of all vehicles is $v_{\textit{max}}$.
    \item Each vehicle enters the intersection area at its crossing time, and at speed $v_{\textit{max}}$.
    \item Cars can change their speed at rates belonging to the interval $[-a_{\textit{max},c}, a_{\textit{max},c}]$.
    \item Trucks can change their speed at rates belonging to the interval $[-a_{\textit{max},tr}, a_{\textit{max},tr}]$.
\end{itemize}
To proceed further, we make the following assumption:
\begin{assumption}
There exists a `control region' of length $x_0(\ell)$ in lane $\ell$ just before the intersection, in which vehicle speeds can be dictated by the central traffic controller, so as to ensure that the vehicles follow pre-determined trajectories.
\end{assumption}
The notion of a control region has been introduced by \textcite{Miculescu2020Polling-Systems-BasedSignals}, and its existence ensures that the trajectory of each vehicle is strictly controlled when near the intersection, so that efficiency and safety can be guaranteed. The platoon-forming algorithm ensures that vehicles have feasible trajectories that are safe with respect to the other vehicles, given that the control region is long enough. This is because when vehicles are grouped into platoons, the delays they incur are non-increasing (from the head to the tail of the platoon) as we have observed previously. However, it is not enough to know about the existence of feasible trajectories, given arrival and crossing times; it is also important to determine these trajectories, which is what the speed-profiling algorithm addresses. 

Since a speed-profiling algorithm concerns only the dynamics internal to each lane, we drop all lane-related notation for ease, from now on. If the arrival time of vehicle $i$ (denoted as $t_{a,i}$) in the polling model is known, then the time $t_{0,i}$ at which the vehicle enters the control region can be derived - it is exactly $t_{a,i} - \frac{x_0}{v_{\textit{max}}}$, because the arrival time of a vehicle in the polling model is the time at which that vehicle would have arrived at the intersection if it travelled at maximum speed $v_{\textit{max}}$ all the way to the intersection, i.e., if it incurred no delay. For convenience, we assume that the control region is long enough that when a vehicle enters the control region, it is far away from the intersection area and hence possesses speed $v_{\textit{max}}$ and does not have to change its trajectory to account for the upcoming intersection yet. 

Let us model each lane forming the intersection in its own one-dimensional space, such that the origin represents the intersection and the control region occupies the interval $[-x_0, 0)$, thus simplifying notation. We then arrive at the following set of initial and terminal conditions:
\begin{align}
x_i(t_{0,i}) &= -x_0,  &\dot{x_i}(t_{0,i})= v_{\textit{max}},\ i \in \{1, \cdots, m \}, \label{eq:initial_cond}\\    
x_i(t_{f,i}) &= 0,  &\dot{x_i}(t_{f,i})= v_{\textit{max}}, \ i \in \{1, \cdots, m \}, \label{eq:terminal_cond}
\end{align}
where $t_{f,i}$ is the crossing time of vehicle $i$ (as determined by the platoon-forming algorithm). The position function $x_i$ for vehicle $i$ is interpreted as the distance to the intersection area. In our model, the position functions of vehicles while in the control region have the range $[-x_0, 0]$, to be consistent with \cite{Miculescu2020Polling-Systems-BasedSignals} and \cite{Timmerman2021PlatoonIntersections}.

The list of conditions above, along with the conditions in Equations~\eqref{eq:initial_cond}-\eqref{eq:terminal_cond}, and the relation between accelerations (\thref{assump:acc_relation}) and safety (\thref{def:safety}), do not fully determine the trajectories in general, and still leave some room for these to be constructed so as to optimise some performance criterion. By specifying the criterion, it is possible to set up optimisation problems to obtain the optimal trajectories for each vehicle. Let us first introduce a type function $\psi$ for each lane, that identifies the type of the vehicle:
\begin{equation}\label{eq:type_function}
    \begin{aligned}
    &\psi: \{1, \cdots, m \} \to \{ 0, 1\}\\
    &\psi(i) = 
    \begin{dcases}
        0, & \text{ if vehicle } i \text{ is a car} \\
        1, & \text{ if vehicle } i \text{ is a truck.}
    \end{dcases}
\end{aligned}
\end{equation}

We also define classes of compatible trajectories based on their type:
\begin{equation}\label{eq:pi_c}
    \Pi_j = \Set*{x_i \given
    \begin{aligned}
    &x_i \in C^1([t_{0,i} , t_{f,i}] \to \R), \\ &\dot{x}_i \textrm{ is differentiable a.e.}  \\
    & 0 \leq \dot{x}_i(s) \leq v_{\textit{max}} \textrm{ for } s \in [t_{0,i} , t_{f,i}]\\
    & |\ddot{x}_i(s)| \leq a_{\textit{max},j} \quad \textrm{ for } s \in [t_{0,i} , t_{f,i}] \\
    & x_i(t_{0,i}) = -x_0, \quad \dot{x}_i(t_{0,i}) = v_{\textit{max}} \\
    & x_i(t_{f,i}) = 0, \quad \dot{x}_i(t_{f,i}) = v_{\textit{max}} 
    \end{aligned}
    },
\end{equation}
where $ j \in \{c, tr\}$. These classes $\Pi_c$ and $\Pi_{\textit{tr}}$ contain exactly those trajectories that satisfy all assumptions and constraints. In particular, it is important for trajectories to be continuously differentiable, so that there are no jumps in the corresponding velocity function. Velocity and acceleration constraints on the trajectories are also included, along with the initial and terminal conditions from Equations~\eqref{eq:initial_cond}-\eqref{eq:terminal_cond}.

\subsection{Joint Optimisation}\label{sec:SPA_optimisation}
We will examine a general case: suppose that the entering and crossing times, along with the types are given by the platoon-forming algorithm for a platoon of $m$ vehicles in a lane. For each vehicle $i$, $t_{0,i}$, and $\psi(i)$ are already known, and $t_{f,i}$ can be obtained as output from the platoon-forming algorithm. We use this to set up the following optimisation problem with a general objective function:
\begin{equation}
\begin{aligned}\label{eq:opt_problem_general_form}
\argmin_{\substack{x_i:[t_{0,i} , t_{f,i}] \to \R\\  i = 1, \cdots, m}} & f(x_1, \cdots, x_m) \\
\textrm{subject to: } & x_i \in \begin{cases}
    \Pi_c, & \text{if } \psi(i) = 0,\\
    \Pi_{\textit{tr}}, & \text{if } \psi(i) = 1,
\end{cases}\\
 & x_i(s) - x_{i+1}(s) \geq v_{\textit{max}} \tau_{\psi(i),\psi(i+1)} \\
 & \:\textrm{for } s \in [t_{0,i+1} , t_{f,i-1}],\\
 &\: i = 1, \cdots, m-1.
\end{aligned}
\end{equation}
The objective function is represented here by a general function $f$ of trajectories of the vehicles in the platoon. The first set of constraints ensures that the vehicle trajectories are from the appropriate class of compatible trajectories (either $\Pi_c$ or $\Pi_{\textit{tr}}$), and then all pertinent constraints are imposed on the trajectories by being a member of that class. The second set of constraints dictates that the distance between two consecutive vehicles should always be at least the safety distance (\thref{def:safety}).

The optimality criterion that we will examine in this paper minimises the sum of areas under the position curve for each vehicle:
\begin{equation}\label{eq:optimization_criterion_capacity}
    f(x_1, \cdots, x_m) = \sum_{i=1}^{m}\int_{t_{0,i}}^{t_{f,i}} |x_i(s)|\ ds.
\end{equation}
\begin{figure}[t!]
    \centering
    \includegraphics[width = 0.8\linewidth]{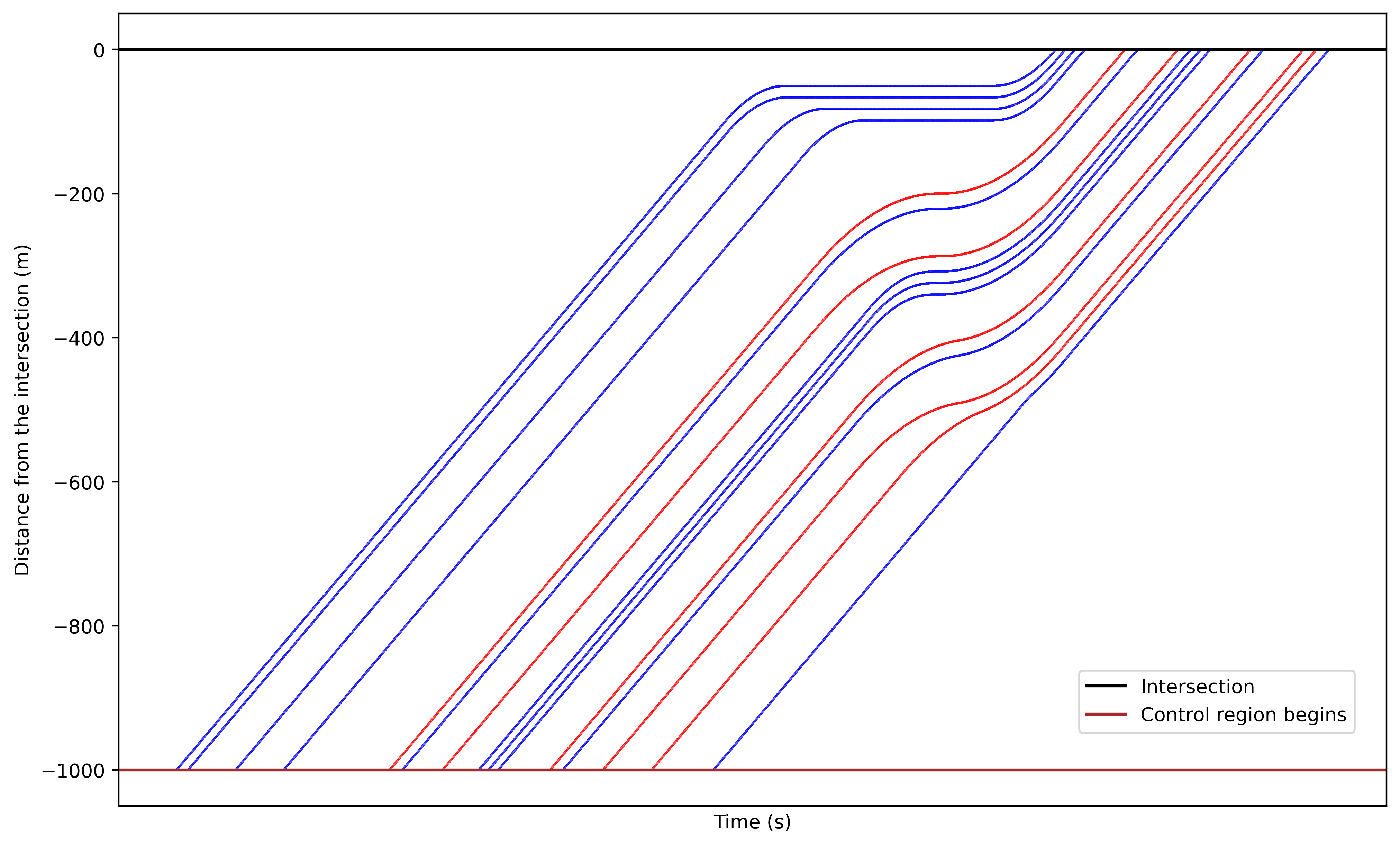}
    \caption{Trajectories obtained by solving the joint optimisation problem for a platoon. Blue and red curves represent cars and trucks, respectively.}
    \label{fig:joint_opt_capacity}
\end{figure}
This optimality criterion jointly ensures that all vehicles in the platoon are as close to the intersection as possible during their time in the control region, which in turn frees up the rest of the control region for other (following) vehicles. This is important, since the trajectories of vehicles can only be controlled in the control region. Thus, it is possible that the above optimisation problem along with the objective function \eqref{eq:optimization_criterion_capacity} might not have a feasible solution for a given length of the control region. However, there always exists a control region length (say $x_0^*$) (depending on the platoon) for which the problem \eqref{eq:opt_problem_general_form} with the objective function \eqref{eq:optimization_criterion_capacity} has a feasible solution. Furthermore, it is easy to see that for all control region lengths larger than $x^*_0$, the associated optimisation problem is always feasible.

This criterion is inspired by the MotionSynthesize procedure proposed in \cite{Miculescu2020Polling-Systems-BasedSignals}, which computes the trajectory for each vehicle individually to be as close to the intersection as possible, while being safe w.r.t.\ the preceding vehicle. In our \emph{joint} optimisation approach, however, we compute the trajectories of all vehicles in a platoon at the same time, resulting in the equivalent trajectories, as shown in \thref{prop:jointly_optimal}.

The trajectories obtained by solving the resulting optimisation problem for a sample platoon are illustrated in Figure~\ref{fig:joint_opt_capacity}. We can make a few interesting observations about the resulting optimal trajectories:
\begin{enumerate}
    \item Firstly, the trajectories are such that the vehicles are travelling at maximum speed $v_{\textit{max}}$ for as long as possible after entering the control region, due to the choice of the optimality criterion, which leads to vehicles being as close to the intersection as possible while traversing the control region, as also observed in \cite{Miculescu2020Polling-Systems-BasedSignals} and \cite{Timmerman2021PlatoonIntersections}. 
    \item There are various possible trajectory shapes depending on the delay of the vehicle. If the vehicle faces no delay, i.e., if $t_{f,i} = t_{a,i}$, then the required average speed is $v_{\textit{max}}$. However, if a vehicle is delayed, then its average speed has to be less than $v_{\textit{max}}$, meaning that it has to decelerate at some point in the control region. It can decelerate to a full stop in case of a `long' delay, and if the delay is `short', then the vehicle does not decelerate to a full stop, but to some strictly positive minimum speed. The required average velocity for each vehicle can also be computed from the entering and crossing times - it is exactly $\frac{x_0}{t_{f,i} - t_{0,i}}$ for vehicle~$i$. Or, in terms of arrival time, the required average speed for vehicle $i$ is $\frac{x_0}{t_{f,i} - t_{a,i} + \frac{x_0}{v_{\textit{max}}}}$.
    \item If a vehicle needs to decelerate and accelerate, i.e., if it is delayed, then the manoeuvre takes place at maximum deceleration/acceleration, unless a car is preceded by a truck (see point~\ref{item:behaviour_preceded_truck}).
    \item Furthermore, if the head of a platoon is delayed, then any change in its speed is delayed as long as possible. The distance between the intersection and the position at which the head of a platoon starts decelerating is just long enough for the deceleration and acceleration manoeuvres. 
    \item All vehicles in a platoon that need to decelerate and accelerate reach maximum speed $v_{\textit{max}}$ again at the same time, i.e., when the head of the platoon is exiting the control region.
    \item \label{item:behaviour_preceded_truck} If a platoon of vehicles contains a truck, then the behaviour of the vehicles following the truck is influenced by the truck, because the maximum acceleration and deceleration for the truck is less than that of a car. The optimal trajectory for a car that enters the control region with the minimum time separation from the truck is to follow the truck exactly, and as the inter-arrival time between the truck and the following car increases, the car has more flexibility to revert back to its original behaviour (i.e., decelerating and accelerating maximally), as can be observed from Figure~\ref{fig:joint_opt_capacity}.
\end{enumerate}

\subsection{Closed-form Expressions}\label{sec:SPA_closed_form}

By solving linear optimisation problems, we can obtain optimal trajectories that ensure safety and efficiency. However, solving these problems is cost-intensive, in terms of time as well as computing resources, even with a state of the art optimisation solver \cite{GurobiOptimizationLLC2022GurobiManual}. Moreover, these solutions have limited interpretability; it is difficult to obtain more insight into the shape and nature of trajectories or predict accurately how a trajectory might change with a slightly different arrival or crossing time or even know if the optimal solution is unique. For these reasons, \textcite{Timmerman2021PlatoonIntersections} solved the optimisation problem corresponding to the MotionSynthesize procedure for two choices of the optimality criterion and computed closed-form expressions for trajectories for homogeneous traffic. In this section, we will motivate the closed-form expressions for vehicle trajectories by inspecting the nature and characteristics of the trajectories obtained from the joint optimisation algorithm, and then prove analytically that such closed-form expressions indeed give rise to safe and optimal trajectories in Appendix~\ref{app:spa_optimality}. These closed-form expressions allow us to better understand the nature of the optimal trajectories and allow for much faster implementation (which we refer to again at the end of this section).

As we noted before in Section~\ref{sec:SPA_optimisation}, the trajectory of a vehicle is not only dependent on its arrival and crossing time, but also on the preceding vehicle(s), which is why it is important to know certain information about the preceding vehicle(s) in the platoon as well. This gives rise to three main cases:
\begin{enumerate}
    \item Car preceded only by cars in its platoon
    \item Car preceded by a truck in its platoon, and
    \item Truck. 
\end{enumerate}
The trajectory of a truck is governed by its entering and crossing times as well as the crossing time of the head of the platoon. The type of vehicles preceding a truck does not affect the trajectory of that truck, which we also remarked upon in Section~\ref{sec:SPA_optimisation}. Only the crossing time of a truck is determined by taking into account the type of its preceding vehicle. If a truck is not delayed, then the optimal trajectory is to drive at speed $v_{\textit{max}}$ throughout the control region, and if it is delayed, then the optimal trajectory is to delay deceleration as long as possible and then change velocity at the maximum rate possible, which is lower than that of a car. Neither behaviour is affected by the type of the preceding vehicle; however, if the current vehicle is a car, then we no longer can make such a simple distinction.

For the above reasons, the trajectories in cases 1 and 3 above can be obtained using the expressions that have already been derived in \cite{Timmerman2021PlatoonIntersections}, by inserting the appropriate acceleration and deceleration parameters. Thus, we focus on case 2, i.e., on obtaining the trajectory for a car given that it is preceded by a truck at some position in the platoon. We emphasise here that this truck can occupy any position from the head of the platoon up to the immediately preceding vehicle, but this truck needs to be the last truck that precedes the car. It is trivial to examine the case where the closest preceding truck is not delayed, because that means that the current car is not delayed either, and hence the only trajectory is to travel at speed $v_{\textit{max}}$ throughout the control region. However, a delayed truck can have two kinds of trajectories: one where it decelerates to a full stop, and the other one where the delay is low, and so the truck decelerates to a certain non-zero speed, and later accelerates to attain speed $v_{\textit{max}}$ again. This indicates the need to consider two sub-cases, and determine the trajectory for the car in each sub-case separately.

\paragraph{Preceding truck comes to a stop}

We first look at the case when the closest preceding truck w.r.t.\ the current car decelerates to a full stop. A necessary (but not sufficient) condition for the car to be able to decelerate to a full stop is for the closest preceding truck to also stop. Suppose that the closest preceding truck for the current car occupies the $j^{th}$ position, and the car occupies the $i^{th}$ position in the platoon. There are $i-j$ $(\geq 0)$ cars between the truck and the current car. Based on the scheduled delay for the car, we can distinguish between the following four situations. The trajectories obtained in these four situations are illustrated in Figure~\ref{fig:SPA_CS_Traj}. Note that the delay for the car can be at most equal to the delay incurred by the truck, since the waiting time or the delay of a vehicle is non-increasing in terms of its position in the platoon as observed before. 

\begin{figure}[!t]
    \centering
    \includegraphics[width = \linewidth]{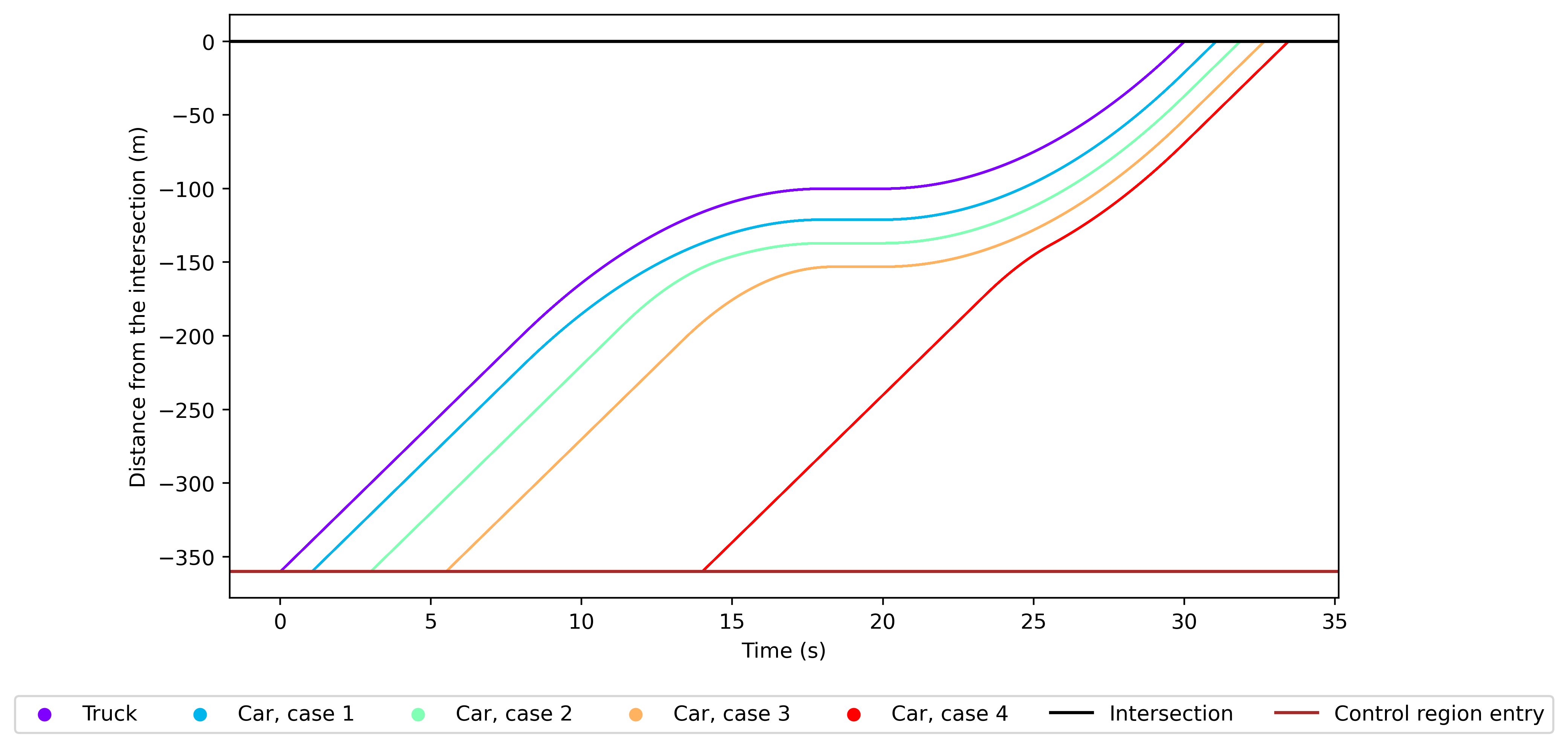}
    \caption{Possible trajectories of a car when preceded by a truck.}
    \label{fig:SPA_CS_Traj}
\end{figure}

\textbf{1) Car follows truck exactly} 
If the delay of the car is exactly equal to the delay of the truck, then the car and the truck have very similar trajectories. This occurs when the vehicles following the truck (up to the current car) all arrive separated exactly by the corresponding minimum inter-arrival times. One possible reason for this phenomenon could be that they were part of a platoon formed at an upstream intersection. Since the delay of the truck is large enough so as to decelerate to a complete stop, the car is also forced to emulate this behaviour. In this case, the acceleration, velocity and position functions for the truck are exactly like the corresponding equations in the case of a car stopping with no trucks preceding, but with $a_{\textit{max},c}$ replaced by $a_{\textit{max},tr}$. Further, the time at which the truck comes to a full stop is given by (see \cite{Timmerman2021PlatoonIntersections}):
\begin{equation*}
    t_{\textit{stop},j} = \frac{x_0}{v_{\textit{max}}} - (t_{f,j} - t_{f,1}) + t_{0,j},
\end{equation*}
where $t_{f,1}$ is the crossing time of the head of the platoon. The position of the truck at that instant is:
\begin{equation*}
    x_j(t_{\textit{stop},j}) = -v_{\textit{max}}(t_{f,j} - t_{f,1}) - \frac{v_{\textit{max}}^2}{2a_{\textit{max},tr}}.
\end{equation*}
Since the delay of the car is equal to that of its closest preceding truck, the acceleration functions of the two vehicles are observed to coincide in the interval $[t_{0,i}, t_{f,j}]$, leading to:
\begin{equation*}
    t_{*,i} = t_{*,j}, \textrm{ for } * \in \{dec, stop, acc \},
\end{equation*}
where $t_{\textit{dec},i}$, $t_{\textit{stop},i}$ and $t_{\textit{acc},i}$ denote the instants at which vehicle $i$ starts its deceleration phase, comes to a full stop and starts its acceleration phase, respectively.

The acceleration profile of the car (below) is motivated by the shape of the first velocity plot in Figure~\ref{fig:SPA_CS_velo}, and the velocity function and the trajectory are obtained by integrating $\ddot{x}_i$ once and twice respectively, with the appropriate limits:\\
\begin{minipage}[t!]{0.45\linewidth}
\begin{equation*}
    \ddot{x}_i (s) = 
    \begin{dcases}
    \ddot{x}_j (s) & s \in [t_{0,i}, t_{f,j}),\\
    0 & s \in [t_{f,j}, t_{f,i}),
    \end{dcases}
\end{equation*}
\end{minipage}
\begin{minipage}[!t]{0.45\linewidth}
\vspace*{-\baselineskip}
\begin{equation*}
    \dot{x}_i (s) = 
    \begin{dcases}
    \dot{x}_j (s) & s \in [t_{0,i}, t_{f,j}],\\
    v_{\textit{max}} & s \in [t_{f,j}, t_{f,i}],
    \end{dcases}
\end{equation*}
\end{minipage}

\begin{equation}\label{eq:spa_cs_truck_stop_1}
    x_i(s) = 
    \begin{dcases}
    x_j (s) -  v_{\textit{max}} \sum_{k=j}^{i-1} \tau_{
    \psi(k), \psi(k+1)
    } & s \in [t_{0,i}, t_{f,j}],\\
    -v_{\textit{max}} \sum_{k=j}^{i-1} \tau_{\psi(k), \psi(k+1)} + v_{\textit{max}} (t - t_{f,j}) & s \in [t_{f,j}, t_{f,i}].
    \end{dcases}
\end{equation}
In other words, the trajectory of the car is just a shifted version of that of the closest preceding truck. It is shifted by a distance of $v_{\textit{max}} \sum_{k=j}^{i-1} \tau_{\psi(k), \psi(k+1)}$ throughout the control region, which is exactly equal to the required minimum safety distance between them, based on their relative positions. This case is represented by the blue curve in Figure \ref{fig:SPA_CS_Traj}.

We also introduce the notion of feasibility of a trajectory, which outlines the condition(s) necessary for a vehicle trajectory to be \emph{compatible} (see Equation~\eqref{eq:pi_c}) and satisfy all constraints. In this case, the trajectory is feasible as long as the deceleration time-point of the car is after it enters the control region, because we require that vehicles enter the control region with speed $v_{\textit{max}}$. Thus, $t_{\textit{dec},i} \geq t_{0,i}$ must hold, and hence:
\begin{equation*}
     \frac{x_0}{v_{\textit{max}}} - (t_{f,j} - t_{f,1}) + t_{0,j} - \frac{v_{\textit{max}}}{a_{\textit{max},tr}} - t_{0,i} = \frac{x_0}{v_{\textit{max}}} - \sum_{k=0}^{j-1} \tau_{\psi(k), \psi(k+1)} + t_{0,j} - \frac{v_{\textit{max}}}{a_{\textit{max},tr}} - t_{0,i}  \geq 0.
\end{equation*}
Recall that the delay of the last preceding truck and the current car are equal in this situation, and hence the above statement can be simplified to:
\begin{equation*}
     \boxed{\frac{x_0}{v_{\textit{max}}} - \frac{v_{\textit{max}}}{a_{\textit{max},tr}} \geq t_{f,i} - t_{f,1} = \sum_{k=0}^{i -1} \tau_{\psi(k), \psi(k+1)} .}
\end{equation*} 
Finally, let us define $T^*_i$ in the following way:
\begin{equation}\label{eq:t_star}
    T^*_i \coloneqq \inf_{s \in [t_{0,i},\ t_{f,i}] } \left\{ \ddot{x}_i(s) = \ddot{x}_j(s) \textrm{ and } \dot{x}_i(s) = \dot{x}_j(s) \textrm{ and } x_i(s) = x_j(s) - v_{\textit{max}} \sum_{k=j}^{i-1} \tau_{\psi(k), \psi(k+1)} \right\},
\end{equation}
which is the first instant when a car catches up with the closest preceding truck, i.e.\ the distance between them is exactly equal to the required minimum distance for safety. Then in case 1, $T^*_i = t_{0,i}$, and the speed and the acceleration profile of the car and the truck coincide after $T^*_i$.

\begin{figure}[!t]
    \centering
    \includegraphics[width = \linewidth]{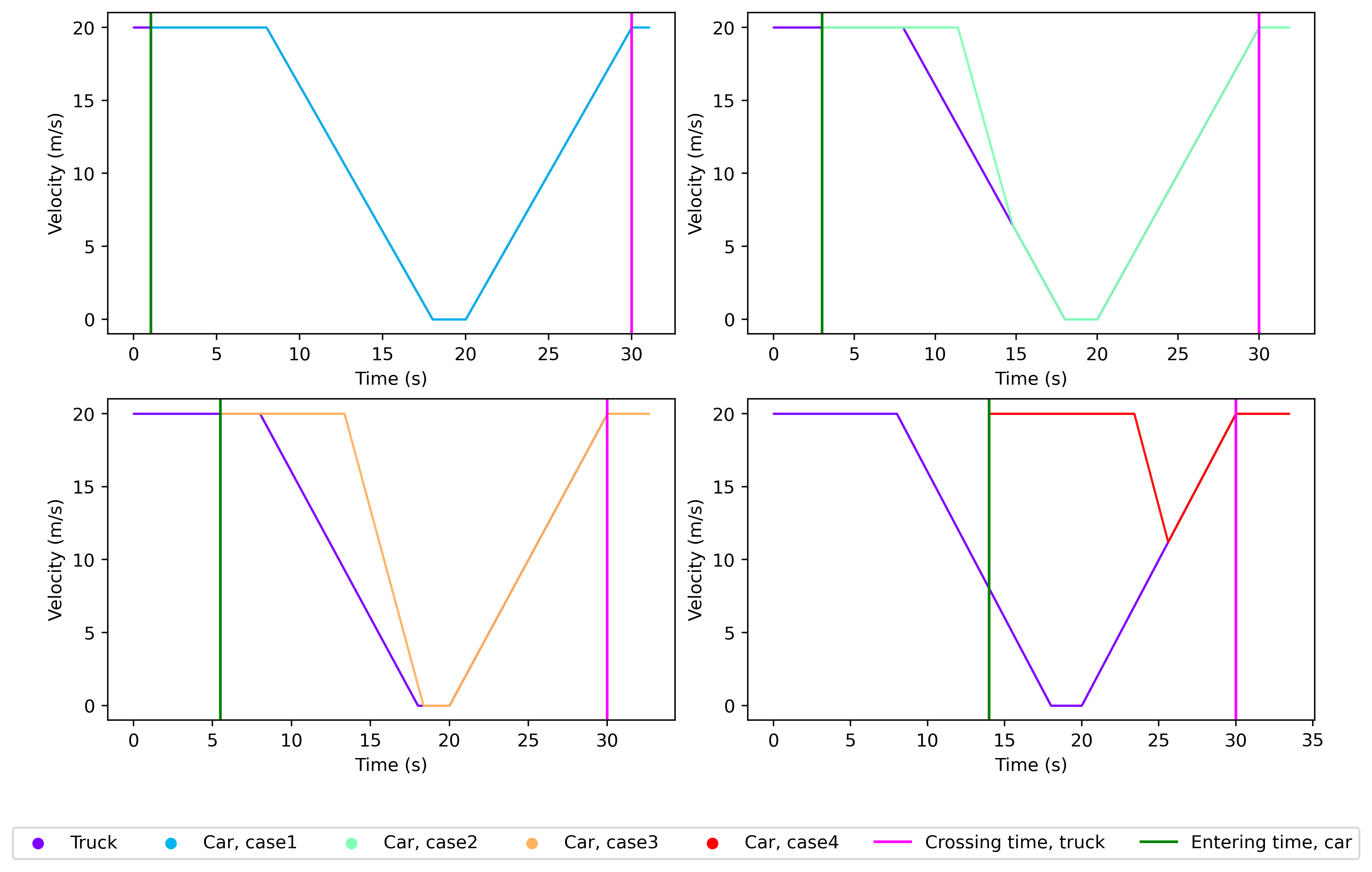}
    \caption{Velocity profile of the car in these four cases.}
    \label{fig:SPA_CS_velo}
\end{figure}

\textbf{2) Car switches decelerations}
If the delay of the car is slightly lower than the delay of the truck, i.e., if the arrival time of the car is slightly later than the required minimum, then the observed optimal trajectory for the car involves two deceleration phases. In the first phase, the car decelerates at $-a_{\textit{max},c}$, and in the second phase the deceleration is equal to that of the truck ($-a_{\textit{max},tr}$). This can be explained intuitively as follows: since the entering time of the car is later than the first case, the car is trailing behind and needs to catch up with the truck in order to exit the control region at its scheduled time. With higher deceleration, less time is required to achieve lower speeds, which means that the first part of the trajectory (travelling at speed $v_{\textit{max}}$) can continue for longer, satisfying the criterion of minimising the distance to the intersection. Alternatively, as the car enters later, there is more flexibility for the car to follow the trajectory which is optimal when there are no trucks, since a truck forces non-optimal behaviour on the cars in the trailing platoon. At the switching time $t_{\textit{sw},i}$ the car has caught up with the all the vehicles in between it and the closest preceding truck (i.e.\ the distance between them is exactly the minimum required safety distance), and from that time onward, it follows the trajectory of its closest preceding truck. Again, we write down an expression for the acceleration profile of the car by inspecting the plot for the corresponding case in Figure~\ref{fig:SPA_CS_velo}, and obtain the velocity and position functions from the acceleration function as before:\\
\begin{minipage}[t!]{0.48\linewidth}
\begin{equation*}
    \ddot{x}_i(s) = 
    \begin{dcases}
    0 & s \in [t_{0,i}, t_{\textit{dec},i}),\\
    -a_{\textit{max},c} & s \in [t_{\textit{dec},i}, t_{\textit{sw},i}),\\
    -a_{\textit{max},tr} & s \in [t_{\textit{sw},i}, t_{\textit{stop},i}),\\
    0 & s \in [t_{\textit{stop},i}, t_{\textit{acc},i}),\\
    a_{\textit{max},tr} & s \in [t_{\textit{acc},i}, t_{f,1}),\\
    0 & s \in [t_{f,1}, t_{f,i}],
    \end{dcases}
\end{equation*}
\end{minipage}
\begin{minipage}[!t]{0.48\linewidth}
\vspace*{-\baselineskip}
\begin{equation*}
    \dot{x}_i(s) = 
    \begin{dcases}
    v_{\textit{max}} & s \in [t_{0,i}, t_{\textit{dec},i}),\\
    v_{\textit{max}} - a_{\textit{max},c}(t - t_{\textit{dec},i}) & s \in [t_{\textit{dec},i}, t_{\textit{sw},i}),\\
    u - a_{\textit{max},tr}(t - t_{\textit{sw},i}) & s \in [t_{\textit{sw},i}, t_{\textit{stop},i}),\\
    0 & s \in [t_{\textit{stop},i}, t_{\textit{acc},i}),\\
    a_{\textit{max},tr} & s \in [t_{\textit{acc},i}, t_{f,1}),\\
    v_{\textit{max}} & s \in [t_{f,1}, t_{f,i}],
    \end{dcases}
\end{equation*}
\end{minipage}

\begin{equation*}
    x_i(s) = 
    \begin{dcases}
    -x_0 + v_{\textit{max}} (t - t_{0,i}) & s \in [t_{0,i}, t_{\textit{dec},i}],\\
    -x_0 + v_{\textit{max}} (t - t_{0,i}) - \frac{1}{2} a_{\textit{max},c}(t - t_{\textit{dec},i})^2 & s \in [t_{\textit{dec},i}, t_{\textit{sw},i}],\\
    -x_0 + v_{\textit{max}}(t_{\textit{dec},i} - t_{0,i}) + \frac{v_{\textit{max}}^2 - u^2}{2a_{\textit{max},c}} + u(t - t_{\textit{sw},i}) - \frac{1}{2} a_{\textit{max},tr}(t - t_{\textit{sw},i})^2 & s \in [t_{\textit{sw},i}, t_{\textit{stop},i}],\\
    -x_0 + v_{\textit{max}}(t_{\textit{dec},i} - t_{0,i}) + \frac{v_{\textit{max}}^2 - u^2}{2a_{\textit{max},c}} + \frac{u^2}{2a_{\textit{max},tr}} & s \in [t_{\textit{stop},i}, t_{\textit{acc},i}],\\
    -x_0 + v_{\textit{max}}(t_{\textit{dec},i} - t_{0,i}) + \frac{v_{\textit{max}}^2 - u^2}{2a_{\textit{max},c}} + \frac{u^2}{2a_{\textit{max},tr}} + \frac{1}{2} a_{\textit{max},tr}(t - t_{\textit{acc},i})^2 & s \in [t_{\textit{acc},i}, t_{f,1}],\\
    -x_0 + v_{\textit{max}}(t_{\textit{dec},i} - t_{0,i}) + \frac{v_{\textit{max}}^2 - u^2}{2a_{\textit{max},c}} + \frac{u^2 + v_{\textit{max}}^2}{2a_{\textit{max},tr}} + v_{\textit{max}}(t - t_{f,1}) & s \in [t_{f,1}, t_{f,i}].
    \end{dcases}
\end{equation*}
Here $u$ is the velocity of the car at time $t_{\textit{sw},i}$, i.e., when it switches decelerations. Using the relations $t_{\textit{dec},i} = t_{\textit{stop},i} - \frac{u}{a_{\textit{max},tr}} - \frac{v_{\textit{max}} - u}{a_{\textit{max},c}}$ and $t_{\textit{stop},i} = t_{\textit{stop},j}$ along with the positional terminal condition of the car, we get:
\begin{equation*}
    u = v_{\textit{max}} - \sqrt{ \frac{2 a_{\textit{max},c} a_{\textit{max},tr} (v_{\textit{max}} [ (t_{f,j} - t_{0,j}) - (t_{f,i} - t_{0,i})  ])}{a_{\textit{max},c} - a_{\textit{max},tr}} },
\end{equation*}
with which we can entirely determine the trajectory of the car. The condition required for this situation to occur can be obtain by realising that $v_{\textit{max}} > u > 0$, leading to:
\begin{equation*}
    \boxed{t_{f,j} - t_{0,j} - \frac{v_{\textit{max}}}{2} \left( \frac{1}{a_{\textit{max},tr}} - \frac{1}{a_{\textit{max},c}} \right) < t_{f,i} - t_{0,i} < t_{f,j} - t_{0,j}.}
\end{equation*}
Further, in case 2, $T^*_i = t_{\textit{sw},i}$.

\textbf{3) Car catches up with preceding truck while truck is at rest}\label{sec:mul_car_preceding_truck_rest}
By now, we observe a pattern - as the difference between the entering and crossing times of the car keeps getting smaller and smaller, the car enters relatively later than the truck, but there is some catching up to be done in order to satisfy its scheduled crossing time, leading to more `flexibility' for the car to decelerate maximally. In order to be as close to the intersection as possible at all times in the control region, it is ideal for the car to delay the deceleration and subsequent acceleration as long as possible to get closer to the intersection. In this case, the car is able to decelerate maximally, so that the start of its deceleration is delayed as much as possible, thus enabling the car to catch up with its closest preceding truck (and all of the cars in between) as soon as the car comes to rest, as seen in Figure~\ref{fig:SPA_CS_velo}. As a result, the trajectory of the car can be constructed as follows:\\
\begin{minipage}[t!]{0.48\linewidth}
\begin{equation*}
    \ddot{x}_i(s) = 
    \begin{dcases}
    0 & s \in [t_{0,i}, t_{\textit{dec},i}),\\
    -a_{\textit{max},c} & s \in [t_{\textit{dec},i}, t_{\textit{stop},i}),\\
    0 & s \in [t_{\textit{stop},i}, t_{\textit{acc},i}),\\
    a_{\textit{max},tr} & s \in [t_{\textit{acc},i}, t_{f,1}),\\
    0 & s \in [t_{f,1}, t_{f,i}],
    \end{dcases}
\end{equation*}
\end{minipage}
\begin{minipage}[!t]{0.48\linewidth}
\vspace*{-\baselineskip}
\begin{equation*}
    \dot{x}_i(s) = 
    \begin{dcases}
    v_{\textit{max}} & s \in [t_{0,i}, t_{\textit{dec},i}],\\
    v_{\textit{max}} - a_{\textit{max},c}(t - t_{\textit{dec},i}) & s \in [t_{\textit{dec},i}, t_{\textit{stop},i}],\\
    0 & s \in [t_{\textit{stop},i}, t_{\textit{acc},i}],\\
    a_{\textit{max},tr}(t - t_{\textit{acc},i}) & s \in [t_{\textit{acc},i}, t_{f,1}],\\
    v_{\textit{max}} & s \in [t_{f,1}, t_{f,i}],
    \end{dcases}
\end{equation*}
\end{minipage}

\begin{equation*}
    x_i(s) = 
    \begin{dcases}
    -x_0 + v_{\textit{max}} (t - t_{0,i}) & s \in [t_{0,i}, t_{\textit{dec},i}],\\
    -x_0 + v_{\textit{max}} (t - t_{0,i}) - \frac{1}{2} a_{\textit{max},c}(t - t_{\textit{dec},i})^2 & s \in [t_{\textit{dec},i}, t_{\textit{stop},i}],\\
    -x_0 + v_{\textit{max}} (t_{\textit{dec},i} - t_{0,i}) + \frac{v_{\textit{max}}^2}{2a_{\textit{max},c}} & s \in [t_{\textit{stop},i}, t_{\textit{acc},i}],\\
    -x_0 + v_{\textit{max}} (t_{\textit{dec},i} - t_{0,i}) + \frac{v_{\textit{max}}^2}{2a_{\textit{max},c}} + \frac{1}{2} a_{\textit{max},tr}(t - t_{\textit{acc},i})^2 & s \in [t_{\textit{acc},i}, t_{f,1}],\\
    -x_0 + v_{\textit{max}} (t_{\textit{dec},i} - t_{0,i}) + \frac{v_{\textit{max}}^2}{2a_{\textit{max},c}} + \frac{v_{\textit{max}}^2}{2a_{\textit{max},tr}} + v_{\textit{max}}(t - t_{f,1}) & s \in [t_{f,1}, t_{f,i}],
    \end{dcases}
\end{equation*}
with the relation $T^*_i = t_{\textit{stop},i}$. The terminal condition for the position of the car at time $t_{f,i}$ leads to:
\begin{equation*}
    t_{\textit{dec},i} = \frac{x_0}{v_{\textit{max}}} - \frac{v_{\textit{max}}}{2} \left( \frac{1}{a_{\textit{max},c}} + \frac{1}{a_{\textit{max},tr}} \right) - (t_{f,i} - t_{0,i}) + t_{f,1}.
\end{equation*}
In this situation, it must hold that $t_{\textit{stop},i} \leq t_{\textit{acc},i}$ and also $t_{\textit{dec},i} \geq t_{\textit{dec},j}$, yielding:
\begin{equation*}
    \boxed{ \frac{x_0}{v_{\textit{max}}} + \frac{v_{\textit{max}}}{2}\left( \frac{1}{a_{\textit{max},tr}} + \frac{1}{a_{\textit{max},c}} \right) \leq  t_{f,i} - t_{0,i} \leq t_{f,j} - t_{0,j} - \frac{v_{\textit{max}}}{2} \left( \frac{1}{a_{\textit{max},tr}} - \frac{1}{a_{\textit{max},c}} \right) .}
\end{equation*}

\textbf{4) Car catches up with truck during the acceleration phase}
The last possible case in this situation is when the car does not come to a full stop, even though its nearest preceding truck does. The delay for the car is not enough to decelerate to a full stop, and hence, we can write the acceleration function $\ddot{x}_i$ as below, and obtain the velocity and position functions by integrating $\ddot{x}_i$ once and twice, respectively:\\
\begin{minipage}[t!]{0.45\linewidth}
\begin{equation*}
    \ddot{x}_i(s) = 
    \begin{dcases}
    0 & s \in [t_{0,i}, t_{\textit{dec},i}),\\
    -a_{\textit{max},c} & s \in [t_{\textit{dec},i}, t_{\textit{acc},i}),\\
    a_{\textit{max},tr} & s \in [t_{\textit{acc},i}, t_{f,1}),\\
    0 & s \in [t_{f,1}, t_{f,i}],
    \end{dcases}
\end{equation*}
\end{minipage}
\begin{minipage}[!t]{0.45\linewidth}
\vspace*{-\baselineskip}
\begin{equation*}
    \dot{x}_i(s) = 
    \begin{dcases}
    v_{\textit{max}} & s \in [t_{0,i}, t_{\textit{dec},i}],\\
    v_{\textit{max}} - a_{\textit{max},c}(t - t_{\textit{dec},i}) & s \in [t_{\textit{dec},i}, t_{\textit{acc},i}],\\
    u + a_{\textit{max},tr}(t - t_{\textit{acc},i}) & s \in [t_{\textit{acc},i}, t_{f,1}],\\
    v_{\textit{max}} & s \in [t_{f,1}, t_{f,i}],
    \end{dcases}
\end{equation*}
\end{minipage}

\begin{equation*}
    x_i(s) = 
    \begin{dcases}
    -x_0 + v_{\textit{max}} (t - t_{0,i}) & s \in [t_{0,i}, t_{\textit{dec},i}],\\
    -x_0 + v_{\textit{max}} (t - t_{0,i}) - \frac{1}{2} a_{\textit{max},c}(t - t_{\textit{dec},i})^2 & s \in [t_{\textit{dec},i}, t_{\textit{acc},i}],\\
    -x_0 + v_{\textit{max}} (t_{\textit{dec},i} - t_{0,i}) + \frac{v_{\textit{max}}^2 - u^2}{2 a_{\textit{max},c}} + u(t - t_{\textit{acc},i}) + \frac{1}{2} a_{\textit{max},tr}(t - t_{\textit{acc},i})^2 & s \in [t_{\textit{acc},i}, t_{f,1}],\\
    -x_0 + v_{\textit{max}} (t_{\textit{dec},i} - t_{0,i}) +  \frac{v_{\textit{max}}^2 - u^2}{2} \left( \frac{1}{a_{\textit{max},tr}} + \frac{1}{a_{\textit{max},c}} \right) + v_{\textit{max}}(t - t_{f,1}) & s \in [t_{f,1}, t_{f,i}],
    \end{dcases}
\end{equation*}
where $u$ is the lowest speed attained by the car while in the control region. Note that in this case the car under consideration has caught up with the closest preceding truck and all the vehicles in between at the time $t_{\textit{acc},i}$. This is also confirmed by the closed-form expressions, which mimic the trajectory of a truck from then onward, thus leading to $T^*_i = t_{\textit{acc},i}$.

By substituting $t_{\textit{dec},i} = t_{f,1} - \frac{v_{\textit{max}} - u}{a_{\textit{max},tr}} - \frac{v_{\textit{max}} - u}{a_{\textit{max},c}}$, and utilising $x_i(t_{f,i}) = 0$, we get:
\begin{equation*}
    u = v_{\textit{max}} - \sqrt{ \frac{2 a_{\textit{max},tr} a_{\textit{max},c} (v_{\textit{max}} (t_{f,i} - t_{0,i}) - x_0) }{a_{\textit{max},tr} + a_{\textit{max},c}} }.
\end{equation*}
Furthermore, this situation is possible if $v_{\textit{max}} \geq u > 0$:
\begin{equation*}
    \boxed{\frac{x_0}{v_{\textit{max}}} \leq t_{f,i} - t_{0,i} < \frac{x_0}{v_{\textit{max}}} + \frac{v_{\textit{max}}}{2}\left( \frac{1}{a_{\textit{max},tr}} + \frac{1}{a_{\textit{max},c}} \right).}
\end{equation*}

The analysis above covers all possible situations that can arise when a car is preceded somewhere in the platoon by a truck that decelerates to a full stop. We notice that the delay of the car determines its trajectory. Alternatively, we can define $\Delta_i \coloneqq t_{f,i} - t_{0,i}$ to be the amount of time that the $i^{th}$ vehicle in the platoon is scheduled to spend in the control region. Then, if we look at all the four cases together, we can see the quantity $\Delta_i$ of the car holding the $i^{th}$ position can take values in 4 consecutive sub-intervals, and this can be used to apply the appropriate case to compute the closed-form expressions. This is illustrated in Figure \ref{fig:delay_cpbt}. If the $i^{th}$ vehicle in a platoon is a car, and if there is at least one truck in the first $i-1$ vehicles, then the quantity $\Delta_i$ can only lie in the interval $[\frac{x_0}{v_{\textit{max}}}, \Delta_j]$, where $j$ is the index of the closest preceding truck. More precisely, it lies in the interval $[\frac{x_0}{v_{\textit{max}}}, \min\{\Delta_{i-1}, \Delta_j\}]$, since delays in a platoon are non-increasing from the head to the tail. Obviously, $\Delta_i \geq \frac{x_0}{v_{\textit{max}}}$, because this is the minimum time required by any vehicle to traverse the control region.

\begin{figure}[t!]
    \centering
    \includegraphics[width = 0.9\linewidth]{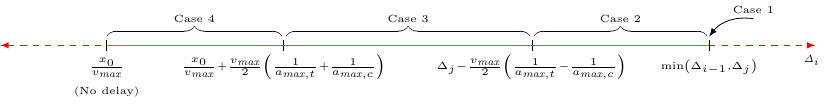}
    \caption{Case distinction for a car preceded by a truck in its platoon, given that the truck comes to a stop in the control region. The distinction depends on the $\Delta$ quantities for the car and the truck (at $i^{th}$ and $j^{th}$ positions respectively), where $\Delta_i = t_{f,i} - t_{0,i}$.}
    \label{fig:delay_cpbt}
\end{figure}

For the analysis of the trajectories of the car that arise when preceded by a truck that does not decelerate to a full stop, we refer the reader to Appendix~\ref{app:SPA_closed_form}.

\paragraph{Implementation}\leavevmode

We now have the closed-form expressions for a general car preceded by a truck, in addition to the previously obtained expressions for the trajectory of an arbitrary truck as well as a car (preceded only by cars in a platoon). A natural question that might arise is how the trajectories for an entire platoon can be obtained from the output of the platoon-forming algorithm. Some trajectory expressions can be expressed directly in terms of quantities such as the delays, acceleration parameters and so on. These include the trajectories of all trucks, and cars that are preceded only by cars. Thus for all such vehicles in a platoon, it is possible to compute their trajectories in a decentralised manner, requiring no input from other trajectories. The only required information is the crossing time $t_{f,1}$ of the head of the platoon, which is obtained as part of the output of the platoon-forming algorithm. We will discuss this more in Appendix \ref{app:spa_optimality}, where we also present the closed-form expressions for the trajectories of such vehicles.

The process of decentralisation would make the trajectory computations even faster. We will see in Appendix \ref{app:SPA_closed_form} that to obtain the trajectories for cars following a truck that does not make a stop in the control region, it is necessary to have information about the lowest speed attained by that truck in the control region. Thus, we propose a trajectory computing scheme as follows: 
\begin{enumerate}[label = (Step \arabic*)]
    \item Compute the trajectories of:
    \begin{enumerate}
        \item all cars until the first truck in the platoon, and
        \item all trucks.
    \end{enumerate}
    These are the vehicles whose trajectories can be computed without requiring any input from other trajectories (except for entering and crossing times of the head of the platoon, which are obtained from the platoon-forming algorithm).
    \item Using the information from the trajectory of each truck (lowest speed attained in the control -- can also be zero), generate the trajectories of all the cars between this truck and the next following truck, if any. The trajectories of these cars can be generated independently of each other.
\end{enumerate}

\section{Analysis}\label{sec:analysis}
In this section, we first combine the platoon-forming algorithm with a speed-profiling algorithm, leading to the complete intersection access control system. Interconnecting the platoon-forming algorithm with a speed-profiling algorithm naturally raises questions about capacity, i.e.\ the maximum traffic intensity that can be sustained by the intersection, and thus the rest of the section is devoted to understanding the capacity of a lane (and the intersection).
\subsection{Integrated mechanism}
\begin{table}[t!]
\centering
    \begin{tabular}{|c|c|c|}
    \hline
       Parameter & Symbol & Value\\
       \hline
       Length of control region & $x_0$ & 600\\
       Max. speed for both vehicle types & $v_{\textit{max}}$ & 20 \\
       Max. feasible acceleration, car & $a_{\textit{max},c}$ & 4 \\
       Max. feasible acceleration, truck & $a_{\textit{max},tr}$ & 2 \\
       Long-term fraction of trucks & $f_{\textit{tr}} $ & 0.4 \\
       Arrival rate parameter & $\lambda_0$, $\lambda_1$ & 0.35 \\
       \hline
    \end{tabular}
    \caption{Parameter values for working example in Section~\ref{sec:analysis}.}
    \label{tab:parameters}
\end{table}
Choosing one of the proposed speed-profiling algorithms along with the platoon-forming algorithm forms our complete access control system that manages traffic flow near an intersection. The working example for this paper is an intersection consisting of two lanes, as shown in Figure~\ref{fig:intersection_schematic}, but we remind the reader that the entire framework can be easily extended to accommodate an arbitrary number of lanes. Vehicles cross the intersection in straight paths, and no turns are allowed.

In (futuristic) principle, the arrival times should be ideally obtained from real-time, on-road vehicle data. Since that is not currently possible, for the purposes of this paper, we simulate arrival times as input to the platoon-forming algorithm. Since the arrival times represent the free-flow arrival times of vehicles to the intersection, the inter-arrival time cannot be arbitrarily small, even in probability, since that would mean that two vehicles from the same lane are arbitrarily close to one another, which is not safe. By the same arguments used to determine the service times, it follows that the inter-arrival time depends on the time separation. As opposed to the service time however, the inter-arrival time of a vehicle is not dependent on the \emph{following} vehicle type, but that of the immediately \emph{preceding} vehicle. Therefore, we set arrival rate parameters $\lambda_\ell$ for each lane, $\ell \in \{1, \cdots, n\}$, and assume the inter-arrival time between vehicle $i$ preceded in its lane $\ell$ by vehicle $i-1$ to be distributed as:
\begin{equation}\label{eq:inter-arrival_dist}
    A_\ell(i) =  \max\left(\tau_{\psi(i-1), \psi(i)} , \text{Exp}(\lambda_{\ell}) \right),
\end{equation}
where $\text{Exp}(\lambda_{\ell})$ refers to an exponential random variable with rate $\lambda_{\ell}$, and $\tau_{\psi(i-1),\psi(i)}$ refers to the minimum time separation between vehicles $i-1$ and $i$ (in a platoon). We see that with the choice of the arrival process, the inter-arrival and service times are dependent on (ordered) pairs of vehicles, and thus a single vehicle will influence two arrival and service times, thereby introducing a weak dependency.

The vehicle types in each lane are independently drawn from a Bernoulli distribution, with the long term probability of a vehicle being a truck denoted by $f_{\textit{tr}}$. 

In the implementation of the platoon-forming algorithm (see Algorithm \ref{alg:polling_sim}), we make a few changes for efficiency and convenience. Firstly, the arrival times of vehicles to the intersection (in free flow) are generated on a sequential basis, so that they can be combined with the platoon-forming algorithm. Secondly, the clearance/switchover times in our implementation are included in the service times for changing lanes.

The resulting trajectories for a certain instance can be found in Figure~\ref{fig:integrated_plot_spa}, with the parameters as given in Table \ref{tab:parameters} above, and Tables \ref{tab:time_separation_same_lane}, \ref{tab:time_separation_diff_lane} in Appendix \ref{app:time_separations}. We also have included 3D video animations of the access control system as ancillary files. A preview image from these animations is included for reference (Figure \ref{fig:animation_snapshot}).
\begin{figure}[!t]
\centering
\begin{tikzpicture}
    \node[anchor=south west,inner sep=0] (image) at (0,0) {\includegraphics[width=0.8\textwidth]{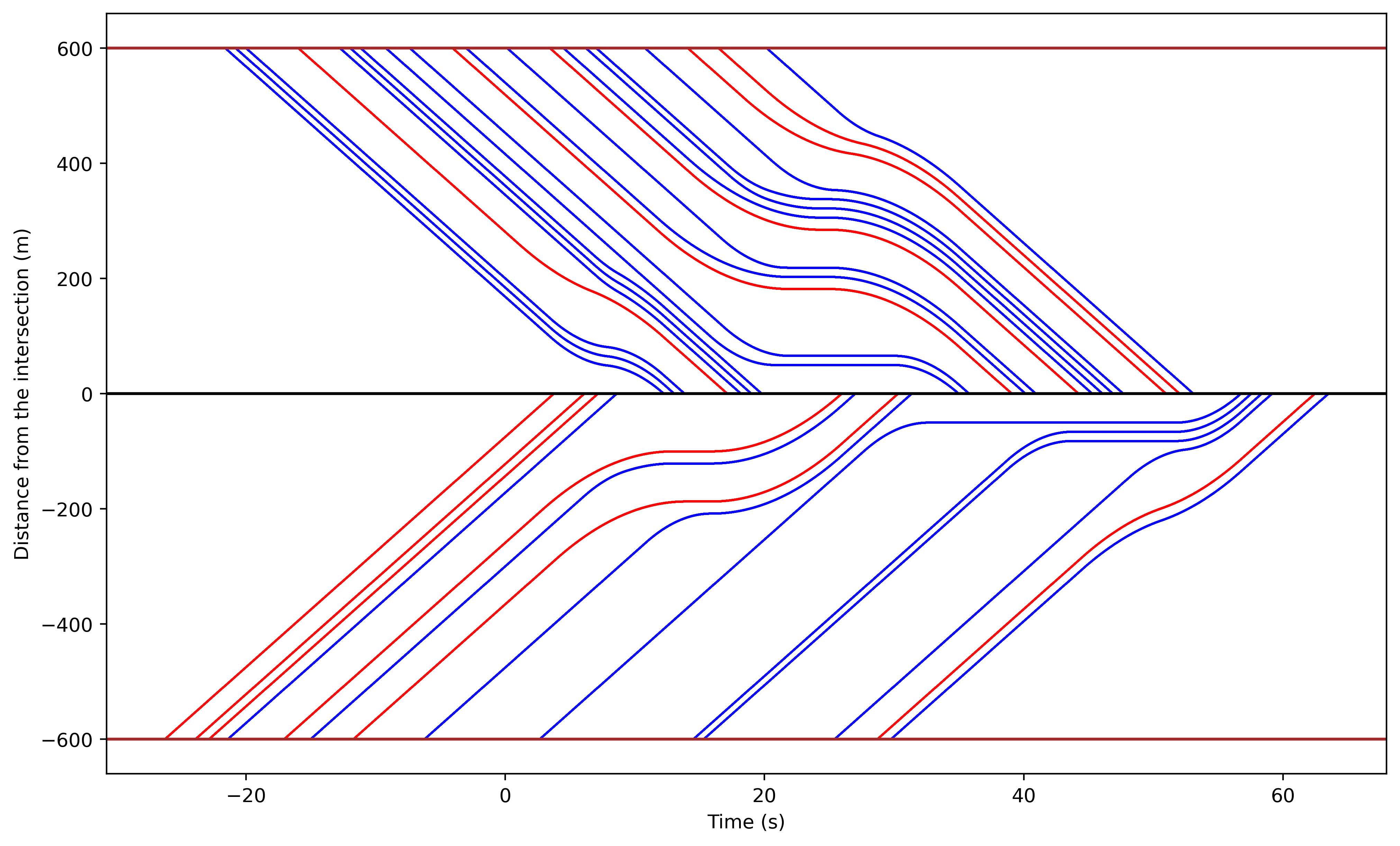}};
    \begin{scope}[x={(image.south east)},y={(image.north west)}]
        \draw[to reversed-to reversed] (0.385,0.547)--(0.4,0.547);
        \draw[to reversed-to reversed] (0.408,0.547)--(0.43,0.547);
        \draw[to reversed-to reversed] (0.435,0.547)--(0.445,0.547);
        \draw[to reversed-to reversed] (0.465,0.522)--(0.553,0.522);
        \draw[to reversed-to reversed] (0.59,0.547)--(0.658,0.547);
        \draw[to reversed-to reversed] (0.67,0.522)--(0.86,0.522);
        \draw[to reversed-to reversed] (0.873,0.547)--(0.953,0.547);
        \draw[to reversed-to reversed] (0.15,-0.05)--(0.2,-0.05) node[right] {Platoon};
        \draw[red, thick] (0.475,-0.05)--(0.525,-0.05) node[right]{\textcolor{black}{Truck}};
        \draw[blue, thick] (0.8,-0.05)--(0.85,-0.05) node[right]{\textcolor{black}{Car}};
    \end{scope}
\end{tikzpicture}
\caption{Visualisation of the optimal trajectories obtained by minimising distance to the intersection.}
\label{fig:integrated_plot_spa}
\end{figure}
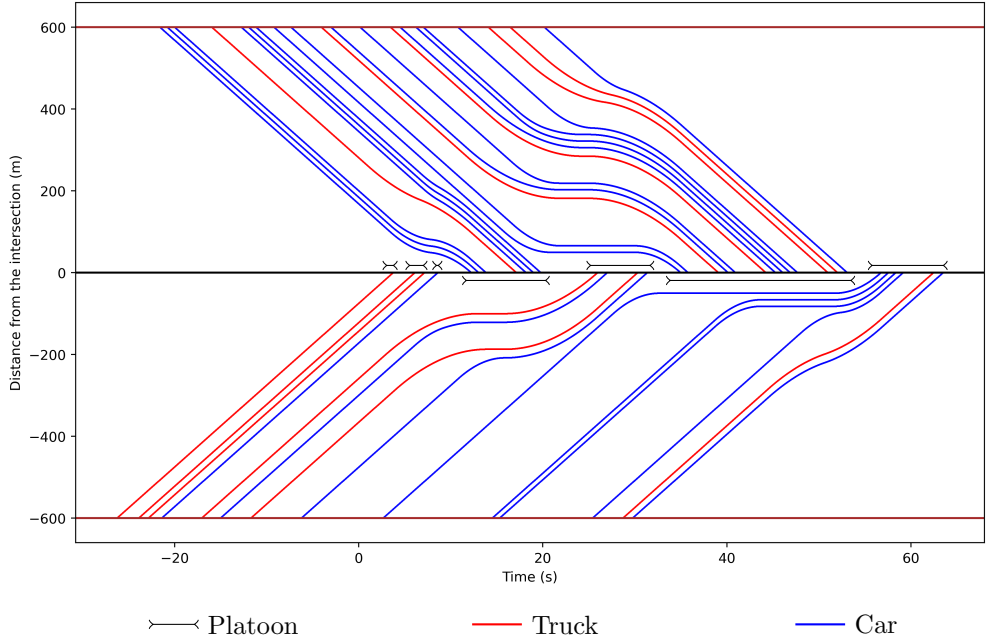
\begin{figure}[!t]
    \centering
    \includegraphics[width = 0.8\linewidth]{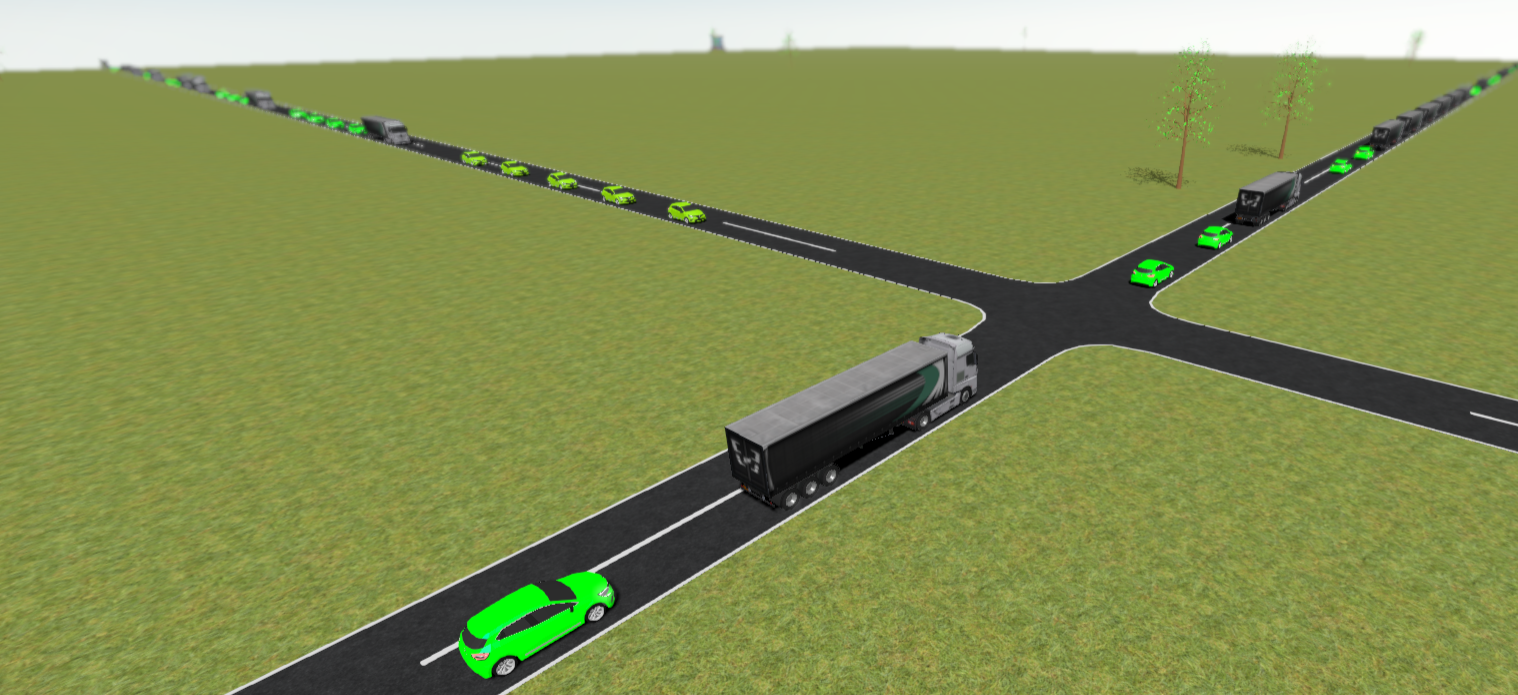}
    \caption{Snapshot from the 3D Animation}
    \label{fig:animation_snapshot}
\end{figure}

Platoons are formed in one of two ways -- vehicles arrive in a platoon structure (with minimum time separation in between them), or vehicles are delayed long enough to come within the minimum safety distance of the following vehicles. In other words, an incoming vehicle can join a platoon as long as its time of arrival is before the end of service of its predecessor. As we noted before in Section~\ref{sec:arrival_scheduler}, by choosing the service discipline to be exhaustive, vehicles are guaranteed to make at most one stop in the control region. This is because an incoming vehicle will either join the current platoon (where some of the member vehicles are already crossing the intersection) or will wait and start a new platoon from the same lane. This makes the trajectories of vehicles independent of other platoons, and hence we can compute trajectories on a platoon-to-platoon basis. This is not possible, for example, with $k$-limited or gated or time-limited service disciplines -- the trajectory of a vehicle can be influenced by its preceding platoon(s).

The speed-profiling algorithm needs input from the platoon-forming algorithm -- by the time a vehicle enters the control region, its trajectory in the control region should ideally be determined completely. In practice, this could be achieved by collecting information about vehicle locations and speeds by means of a second region prior to the control region, as in \cite{Timmerman2021PlatoonIntersections}, and ensuring that the platoons are formed and the trajectories determined before vehicles enter the control region. In-depth analysis of the exact mechanisms involved is beyond the scope of this paper, and for more information we refer to \cite{Timmerman2021PlatoonIntersections}.

\subsection{Capacity Analysis}

\subsubsection{Critical traffic load}
An important benefit of using a polling model (and queueing theory) in our framework is that it provides a natural basis for capacity and performance analysis. A key quantity for the capacity and performance analysis is the (traffic) load per lane. Here, $B_\ell(i)$ denotes the service time of vehicle $i$ in a platoon in lane $\ell$, and the inter-arrival time $A_\ell(i)$ between the arrivals of vehicles $i$ and $i+1$ of a platoon in lane $\ell$ is as defined in Equation ~\eqref{eq:inter-arrival_dist}. With our structure of the arrival process and service times, the joint stochastic process $\{(B_\ell(i), A_\ell(i+1)\}_i$ has a unique stationary distribution for every lane. More specifically, the process $\{(B_\ell(i), A_\ell(i+1)\}_i$ is strictly stationary and metrically transitive, as defined in the framework of \textcite{Loynes1962}. Following \cite{Loynes1962}, we define the process $\{U_\ell(i)\}_i$ for each lane $\ell$:
\begin{equation*}
    U_\ell(i) = B_\ell(i) - A_\ell(i+1),\quad i \in \mathbb{N}_{\geq 0}.
\end{equation*}
Under this setup, Loynes has shown in \cite{Loynes1962} that $E[U_\ell] < 0$ is a sufficient condition for the queue associated with lane $\ell$ to be stable, where $U_\ell$ is an arbitrarily chosen random variable from the sequence $\{U_\ell(i)\}_i$. The condition $E[U_\ell] < 0$ may be rewritten in the form $\frac{E[B_\ell]}{E[A_\ell]} < 1$, where again the random variables $B_\ell$ and $A_\ell$ refer to an arbitrarily chosen service time and inter-arrival time respectively. This makes it natural to use $\tilde{\rho}_\ell = \frac{E[B_\ell]}{E[A_\ell]}$ as the definition of the load (for lane $\ell$) in our setup, and to define the aggregate load (over all lanes) as $\tilde{\rho} = \sum_{\ell = 1}^n \tilde{\rho}_\ell$. Indeed, for polling systems with Poisson arrivals and a wide range of service disciplines, the stability condition is known to be $\tilde{\rho} < 1$ (see \cite{Altman1992, Fricker1994}).

In our specific context, we can write, using the law of total expectation:
\begin{subequations}
\begin{align}
    E[B_\ell] =&\ \sum_{\eta, \gamma} E[B_\ell(i)|\psi(i) = \eta, \psi(i+1) = \gamma] P(\psi(i) = \eta, \psi(i+1) = \gamma), \label{eq:rho_lane_e_b}\\
    E[A_\ell] =&\ \sum_{\kappa, \eta} E[A_\ell(i)|\psi(i-1) = \kappa, \psi(i) = \eta] P(\psi(i-1) = \kappa, \psi(i) = \eta), \label{eq:rho_lane_e_a}
\end{align}
\end{subequations}
where $\kappa, \eta, \gamma \in \{ 0,1\}$. The terms $P(\psi(i-1) = \kappa, \psi(i) = \eta)$ and $P(\psi(i) = \eta, \psi(i+1) = \gamma)$ denote the probability of the occurrence of such sequences of vehicles. It is possible to decompose these probabilities into products of the probabilities of individual vehicles being of a specific type, since vehicle types are independently drawn from a Bernoulli distribution for every vehicle. Recall that in our model, the service and inter-arrival times depend on pairs of consecutive vehicles and not just the current vehicle, and hence additional information is needed to compute the expected values, which is why we condition on the type of the preceding or succeeding vehicle.

The quantity $E[A_\ell]$ can be readily computed (see the definition in \eqref{eq:inter-arrival_dist}), using the conditional expectation below: 
\begin{equation*}
        E[A_\ell(i)|\psi(i-1) = \kappa, \psi(i) = \eta] = \tau_{\kappa,\eta} + \frac{e^{-\lambda_\ell \tau_{\kappa, \eta}}}{\lambda_\ell}, \quad \kappa, \eta, \gamma \in \{ 0,1\}.
\end{equation*}
By \emph{un}conditioning the above expression, we arrive at an expression for the mean inter-arrival time $E[A_\ell]$.

The mean service time $E[B_\ell]$ is obtained very easily via Equation~\eqref{eq:rho_lane_e_b}, since we know that:
\begin{equation*}
    E[B_\ell(i)|\psi(i) = \eta, \psi(i+1) = \gamma] = \tau_{\eta,\gamma}.
\end{equation*}

Finally, as mentioned before, it is possible to decompose terms of the form $P(\psi(i) = \eta, \psi(i+1) = \gamma)$ into probabilities for individual vehicles:
\begin{align*}
    P(\psi(i) = \eta, \psi(i+1) = \gamma) &=\ 
    P(\psi(i) = \eta) \cdot P(\psi(i+1) = \gamma), \text{ and}\\
    P(\psi(i) = \eta) &=\ \begin{dcases}
        f_{\textit{tr}} & \text{if vehicle } i \text{ is a truck, i.e.\ if }\eta = 1\\
        1 - f_{\textit{tr}} & \text{otherwise.}
    \end{dcases}
\end{align*}
Substituting in the definition of load $\tilde{\rho}_\ell$ for lane $\ell$, with our two vehicle types, we get:
\begin{align*}
    \tilde{\rho}_\ell =&\ \frac{\tau_{0,0}(1-f_{\textit{tr}})^2 + (\tau_{0,1}+\tau_{1,0})f_{\textit{tr}}(1-f_{\textit{tr}}) + \tau_{1,1}f_{\textit{tr}}^2}{\left(\tau_{0,0} + \frac{e^{-\lambda_\ell \tau_{0,0}}}{\lambda_\ell}\right)(1-f_{\textit{tr}})^2 + \left(\tau_{0,1} + \frac{e^{-\lambda_\ell \tau_{0,1}}}{\lambda_\ell} + \tau_{1,0} + \frac{e^{-\lambda_\ell \tau_{1,0}}}{\lambda_\ell}\right)f_{\textit{tr}}(1-f_{\textit{tr}}) + \left( \tau_{1,1} + \frac{e^{-\lambda_\ell \tau_{1,1}}}{\lambda_\ell}\right)f_{\textit{tr}}^2} .
\end{align*}
In the next subsection, we will use the quantity $\tilde{\rho}$ for setting the parameters $\lambda_\ell$ so as to achieve high load values, and look at how the length of the control region affects the entire mechanism for such high load conditions.

\subsubsection{Effect of the control region on the capacity}

Any deviation in speed required by the speed-profiling algorithm (from the initial speed $v_{\textit{max}}$) must occur within the control region where trajectories can be manipulated. As we can observe in Figures~\ref{fig:joint_opt_capacity} and~\ref{fig:integrated_plot_spa}, the location/instant at which a vehicle decelerates (i.e., changes speed for the first time) depends, among other factors, on its position in its platoon. Generally, the location at which deceleration begins for a vehicle is propagated from the head to the tail of a platoon such that it is closer and closer to the beginning of the control region. Thus, in case of high load, it could be possible that the proposed optimal trajectory for an incoming vehicle requires it to decelerate before the control region even begins, which violates the assumptions of our framework. This naturally leads to bounds on the capacity of each lane, and thereby also of the intersection for a given length of the control region. Recall that a trajectory is deemed to be feasible if any changes in the speed of the vehicle occur within the control region.

We conduct a numerical study to explore the relationship between the length of the control region and its effect on the feasibility of vehicle trajectories, in high-load situations. Furthermore, we analyse the underlying polling model, but now interpreted in a different manner. In this new polling model, arrivals are marked by vehicles entering the control region, and service times and departures have the same interpretation as before (see Section~\ref{sec:arrival_scheduler}). Thus, the new polling model is intrinsically linked to the length of the control region, which was not the case with the polling model introduced in Section \ref{sec:arrival_scheduler}. Further, at any given point in time, the number of waiting customers in this model (denoted by $Q$) corresponds to the number of vehicles in the control region.

In this numerical study, we simulate arrivals to an intersection for a very long period of time. Next, we use the speed-profiling algorithm to compute the positions at which each vehicle decelerates. The length of the control region is suitably chosen to be large enough so that all vehicles arriving to the intersections have feasible trajectories. Next, we (artificially) reduce the length of the control region in steps, and compute the proportion of trajectories (`unsuitable') for which the position at the start-of-deceleration instant is now beyond the (beginning of the) control region. Recall that the optimal trajectories obtained by the speed-profiling algorithm do not vary with the control region length, so we can continue to use the same optimal trajectories even if the length of the control region is reduced.

Every time we reduce the length of the control region, we also compute the proportion of delayed vehicles (as a proxy for $Q$) being greater than a certain, critical value (see the definition of $N_1(\cdot)$ in \eqref{eq:n_1} below), with help from the polling model. This acts as a standard of comparison for the proportion of unsuitable trajectories.

Finally, we repeat the entire process for the case when all vehicles have infinite acceleration and deceleration -- this changes the optimal trajectories because the head of the platoon now no longer stops at a distance away from the intersection, since it can accelerate to speed $v_{\textit{max}}$ from rest instantaneously.

The results of these experiments are plotted in Figures~\ref{fig:capacity_analysis_symm} and ~\ref{fig:capacity_analysis_asymm}, where the quantities $N_1(\cdot)$ and $N_2(\cdot)$ are the expected number of vehicles that would fit in the control region in the realistic and infinite acceleration regime respectively. These are given by:
\begin{align}
    N_1(x) &= \max \left\{0, \left\lfloor \frac{x - \frac{v_{\textit{max}}^2}{2}\left(\frac{f_{\textit{tr}}}{a_{\textit{max},tr}} + \frac{1 - f_{\textit{tr}}}{a_{\textit{max},c}} \right)}{D} \right\rfloor \right\}, \label{eq:n_1}\\
    N_2(x) &= \left\lfloor \frac{x}{D} \right\rfloor, \label{eq:n_2}\\
    \text{with } D &= v_{\textit{max}}\left(f_{\textit{tr}}^2\tau_{1,1} + f_{\textit{tr}}(1-f_{\textit{tr}})( \tau_{1,0} + \tau_{0,1}) + (1-f_{\textit{tr}})^2\tau_{0,0}\right) = v_{\textit{max}}E[B_\ell],
\end{align}
where $x$ is the length of the control region, and $\left\lfloor \cdot \right\rfloor$ denotes the floor function. The denominator $D$ of the term on the right-hand-side of Equation~\eqref{eq:n_2} also appears in Equation~\eqref{eq:n_1}, and it represents the minimum required separation between an arbitrary pair of (consecutive) vehicles on average. The term $\frac{v_{\textit{max}}^2}{2}\left(\frac{f_{\textit{tr}}}{a_{\textit{max},tr}} + \frac{1 - f_{\textit{tr}}}{a_{\textit{max},c}} \right)$ can be interpreted as the average distance from the intersection at which the head of an arbitrary platoon comes to a stop, if faced with a long delay, in the realistic acceleration-deceleration regime. 

\begin{figure}[!t]
    \begin{minipage}{0.49\linewidth}
        \centering
        \includegraphics[width=0.9\linewidth]{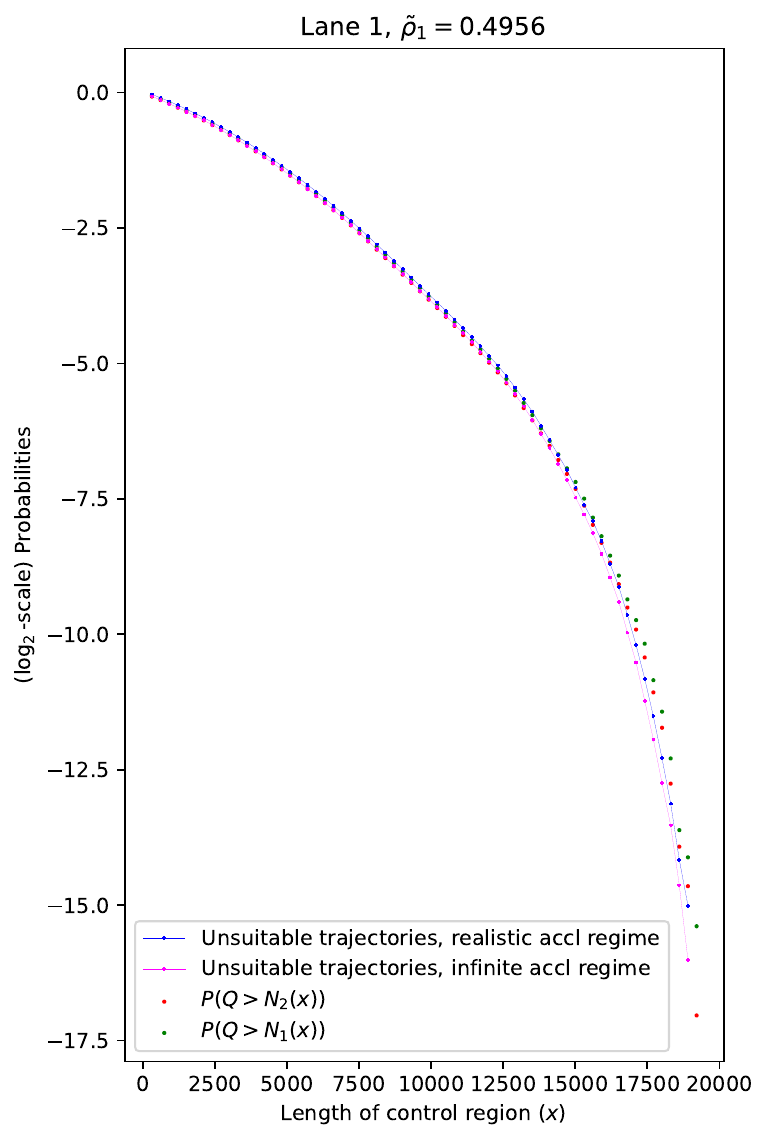}
    \end{minipage}%
    \hfill
    \begin{minipage}{0.49\linewidth}
        \centering
    \includegraphics[width=0.9\linewidth]{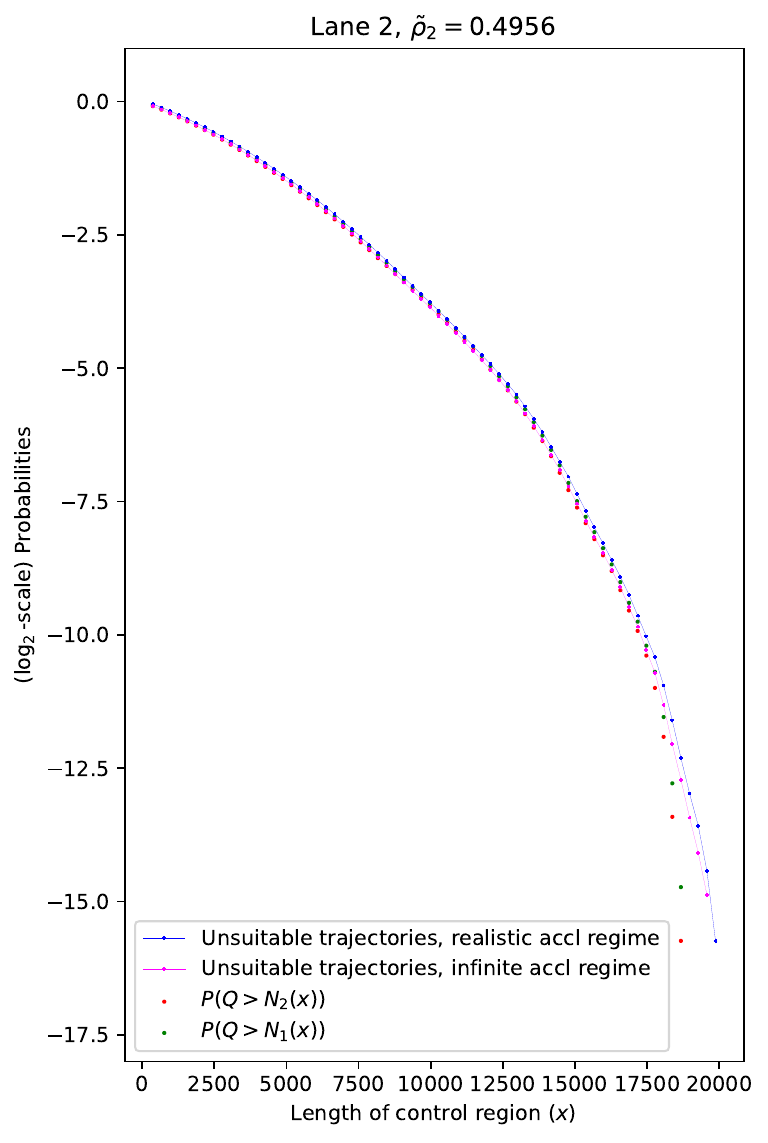}
    \end{minipage}
    \caption{Proportion of unsuitable trajectories compared to queue length probabilities, symmetric high-load case.}
    \label{fig:capacity_analysis_symm}
\end{figure}

We look at two situations -- a symmetric high-load scenario (where both lanes possess the same load), as well as an asymmetric high-load scenario, where we have one lane experiencing significantly more traffic than the other, reminiscent of a major lane--minor lane situation. The key parameters used for this analysis, along with their values, are presented in Table \ref{tab:capacity_analysis}. The control region was gradually reduced in steps of $10$ units. The relationship between the length of the control region and the proportion of unsuitable trajectories is examined for very large, as well as realistic accelerations, and is visualised in Figure~\ref{fig:capacity_analysis_symm} and~\ref{fig:capacity_analysis_asymm} for the symmetric and asymmetric high-load scenarios respectively.

\begin{table}[!t]
\centering
\begin{tabular}{|c|*{4}{p{11mm}}|}
\hline
\multirow{3}{*}{Parameters (Units)} & \multicolumn{4}{c|}{Values}                                                                                    \\ \cline{2-5} 
                            & \multicolumn{2}{c|}{Symmetric high-load scenario}         & \multicolumn{2}{c|}{Asymmetric high-load scenario} \\ \cline{2-5} 
                            & \multicolumn{1}{p{22mm}|}{Lane 1} & \multicolumn{1}{p{11mm}|}{Lane 2} & \multicolumn{1}{p{22mm}|}{Lane 1}        & Lane 2        \\ \hline
\multicolumn{1}{|c|}{$\lambda_\ell$ (arrivals/s)}      & \multicolumn{1}{p{11mm}|}{0.39}       & \multicolumn{1}{p{11mm}|}{0.39}       & \multicolumn{1}{p{11mm}|}{1.34}              &          0.06     \\ \hline
\multicolumn{1}{|c|}{$\tilde{\rho}_\ell$}      & \multicolumn{1}{p{11mm}|}{0.4956}       & \multicolumn{1}{p{11mm}|}{0.4956}       & \multicolumn{1}{p{11mm}|}{0.8997}              &      0.0895         \\ \hline
\multicolumn{1}{|c|}{Average number of delayed vehicles}      & \multicolumn{1}{p{11mm}|}{73.03}       & \multicolumn{1}{p{11mm}|}{72.88}       & \multicolumn{1}{p{11mm}|}{21.21}              &    19.31           \\ \hline
\multicolumn{1}{|c|}{Average delay (s)}      & \multicolumn{1}{p{11mm}|}{221.01}       & \multicolumn{1}{p{11mm}|}{220.89}       & \multicolumn{1}{p{11mm}|}{35.37}              &       325.69        \\ \hline
\end{tabular}
\caption{Parameters and key quantities for capacity analysis scenarios.}
\label{tab:capacity_analysis}
\end{table}

We observe the following:
\begin{enumerate}
    \item The probability estimates for $P(Q > N_1(x))$ and $P(Q > N_2(x))$ closely approximate the proportion of unsuitable trajectories in the realistic and infinite acceleration regime respectively, for high-load situations. As we discussed before, the quantities $N_1(\cdot)$ and $N_2(\cdot)$ approximate the average number of vehicles that can fit into a control region of given length, taking into account the required separation between pairs of vehicles and the frequency of vehicle types. Thus, in high-load situations, where the vehicle-to-capacity ratio is high, platoon sizes will also grow to be large due to an exhaustive service discipline. If an incoming vehicle has an optimal trajectory such that deceleration starts beyond the scope of the control region, then this is because it is not possible to start deceleration at a later point in time. This in turn means that the control region is (completely) occupied by other vehicles with speeds less than $v_{\textit{max}}$ when the vehicle in question enters the control region, and the situation becomes immediately unsafe, and the trajectory, unsuitable.
    \item Even though it is intuitive to see that the queue length probabilities approximate the probability of an unsuitable trajectory well, it is still remarkable because $P(Q > \alpha)$ for some $\alpha$ can be interpreted as the long-run fraction of time that the queue length is greater than $\alpha$. On the other hand, the probability estimate of an unsuitable trajectory for a given length of the control region was computed by counting the proportion of trajectories that were rendered unsuitable. We are essentially comparing time averages with proportions, and both follow the same trend as the length of the control region varies.
    \item For control region lengths $x$ such that $N_1(x) > 0$, i.e.\ for $x > \frac{v_{\textit{max}}^2}{2}\left( \frac{f_{\textit{tr}}}{a_{\textit{max},\textit{tr}}} + \frac{1 - f_{\textit{tr}}}{a_{\textit{max},c}}\right)$, we have:
    \begin{align*}
        N_2(x) - N_1(x) =&\ \frac{v_{\textit{max}}^2}{2D}\left( \frac{f_{\textit{tr}}}{a_{\textit{max},\textit{tr}}} + \frac{1 - f_{\textit{tr}}}{a_{\textit{max},c}}\right)\\
        =&\ \frac{v_{\textit{max}}}{2E[B_\ell]}\left( \frac{f_{\textit{tr}}}{a_{\textit{max},\textit{tr}}} + \frac{1 - f_{\textit{tr}}}{a_{\textit{max},c}}\right).
    \end{align*}
    This difference is exactly the gap in potential lane capacity that is `lost' due to vehicles accelerating to maximum speed $v_{\textit{max}}$ just before entering the intersection. Thus it might seem at first glance that lane capacity reduces when vehicles enter intersection at maximum speed $v_{\textit{max}}$. However, one must remember that if vehicle start from rest when entering the intersection, their travel time while crossing the intersection is much larger, and thus we would need to increase service times to reflect this change, which finally leads to reduced capacity. Thus from an efficiency perspective, it is prudent to let vehicles attain the maximum possible speed just before entering the intersection.
    \item If we have estimates for the (maximum) arrival rates of lanes, then running such an experiment could help guide to select control region lengths appropriately for each lane.
    \item The capacity of lane $\ell$ thus has an upper bound, which is equal to $N_1(x_\ell)$, where $x$ is the length of the control region in lane. The capacity refers to the number of vehicles that can cross the intersection from that lane per unit of time. From this, we can also obtain an estimate for the maximum capacity of the intersection by weighing the capacity of each lane with the fraction of time spent by the server in that lane ($\tilde{\rho}_\ell$):
    \begin{equation*}
        \textrm{Maximum Intersection Capacity} \approx \sum_{\ell = 1}^n \tilde{\rho}_\ell \cdot N_1(x_\ell)
    \end{equation*}
    The above expression is an approximation because we do not consider switchover times in our analysis.
\end{enumerate}

\begin{figure}[!t]
    \begin{minipage}{0.49\linewidth}
        \centering
        \includegraphics[width=0.9\linewidth]{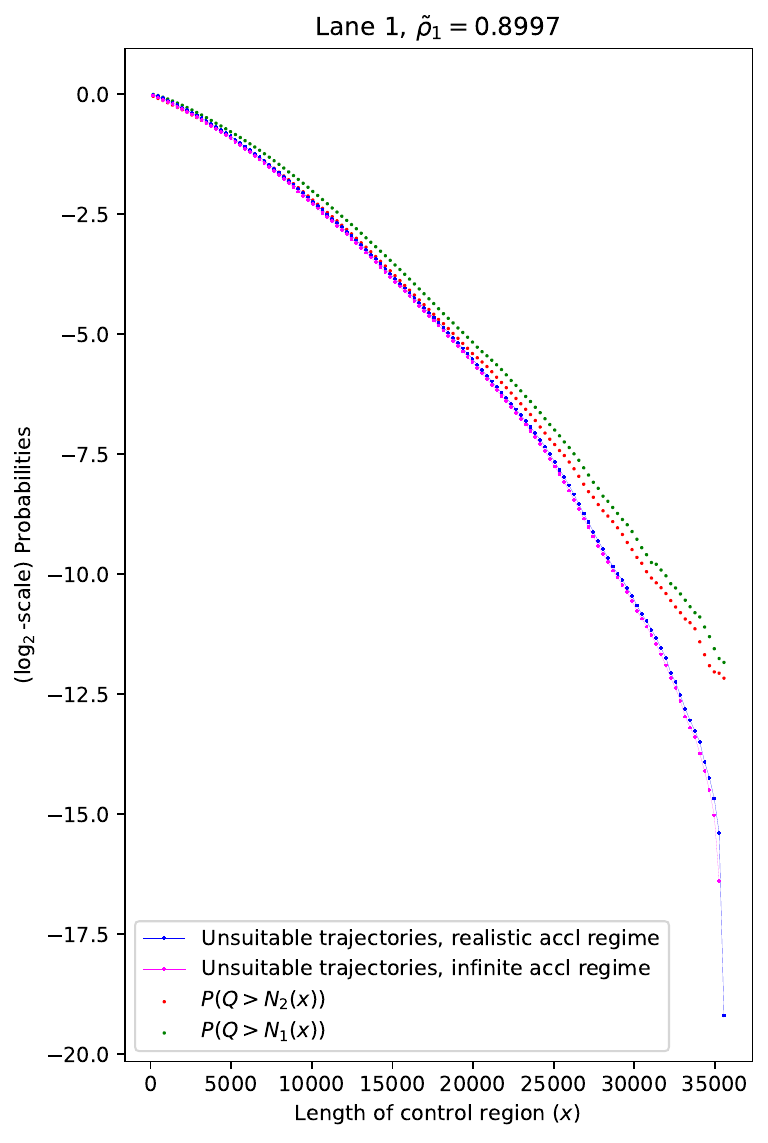}
    \end{minipage}%
    \hfill
    \begin{minipage}{0.49\linewidth}
        \centering
    \includegraphics[width=0.9\linewidth]{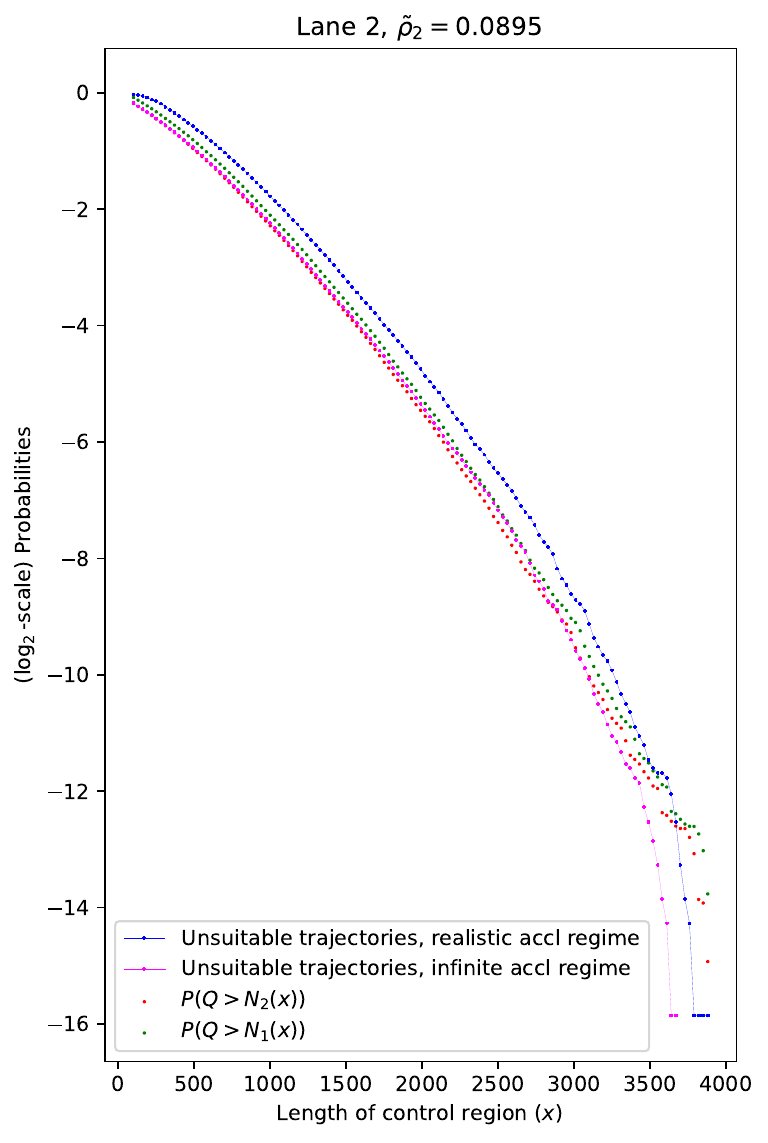}
    \end{minipage}
    \caption{Proportion of unsuitable trajectories compared to queue length probabilities, asymmetric high-load case.}
    \label{fig:capacity_analysis_asymm}
\end{figure}
\section{Conclusion}
In this paper, we have developed and analysed a comprehensive platooning-focused access system for an urban intersection with heterogeneous autonomous traffic. We further show that the capacity of a lane for a given length of the control region is approximately equal to a (modified) queue length tail probability in the polling model, thus providing a guide to select the appropriate control region lengths as well.

There are several interesting avenues for further research. We believe that it is possible to extend this analysis to accommodate trajectories involving turning at intersections, for e.g.\ by considering a modified intersection design as proposed by \textcite{Chen2023AVehicles}. 

The polling model and by extension the platoon-forming algorithm, currently involves an exhaustive service discipline, which minimises the average delay under certain assumptions, as shown in \cite{Timmerman2021PlatoonIntersections}. More work is needed to investigate the framework for service disciplines which might be perceived as more fair, such as $k$-limited or gated disciplines. However, studies have shown that the capacity significantly decreases when the service discipline is not exhaustive (see \cite{Tachet2016RevisitingSystems, Timmerman2021PlatoonIntersections}), making this a less attractive alternative.

It is also interesting to consider how our framework might be extended to network scenarios, to capture and efficiently manage the dynamics for an entire city, for instance. Furthermore, an urban area inevitably involves pedestrian movement. More work needs to be done to integrate pedestrian traffic flow in our framework, as approached for instance in \cite{Zhu2023} using learning methods.

Finally, it would be very interesting to study in more detail whether this framework could also be applied to a mixed-traffic setting, where human-driven vehicles share road space along with autonomous vehicles, as investigated in \cite{Gunther2016PlatooningLights, Moradi2021, Wu2022InfluenceTraffic, Wang2023SafetyBalanced, Yang2023, Zhang2022PlatoonCentered, Wang2022DigitalTwin}. Guaranteeing the safety in such a system might be significantly more challenging, most likely leading to larger safety margins and additional constraints.

\section*{Acknowledgements} The authors wish to thank Jeroen van Riel for his work on the animations, and to Sanne van Kempen for her helpful comments.

\printbibliography[heading=bibintoc, title = {References}]

\appendix

\section{3D Animation}\label{app:animation}
To effectively demonstrate the intersection access control system, we have made 3D animations of the integrated mechanism available as supplementary material. We present a scenario with low to medium traffic intensity as well as one with high traffic intensity. To fully appreciate the extended framework necessary to deal with traffic heterogeneity, we also have added scenarios with homogeneous traffic (i.e.\ only cars). The animations involve cars and trucks, and the cars are colour-coded to aid viewers to detect their speed: if a car is green, then it is travelling at speed $v_{\textit{max}}$, and if it is at rest (i.e.\ possessing zero speed), it takes on a red appearance, with in-between speeds colour-coded accordingly. 

The time separations used to make these animations are:
\begin{table}[!h]
\centering
\begin{tblr}{
  width = 0.6\linewidth,
  colspec = {Q[377,r,m]Q[200,c,m]Q[162,c,m]Q[162,c,m]},
  column{1} = {r},
  column{2} = {c},
  cell{1}{3} = {c=2}{0.1944\linewidth},
  cell{3}{1} = {r=2}{},
  vline{3-5} = {2,4}{},
  vline{2-5} = {3,4}{},
  hline{2} = {3-4}{},
  hline{3-5} = {2-4}{},
}
                  &       & Preceding Vehicle &       \\
                  &       & Car               & Truck \\
Following Vehicle & Car   & 0.65 / 0.8              & 0.8 / 1.5 \\
                  & Truck & 1.5 / 2.5              & 0.9 / 2.5
\end{tblr}
\caption{Time separations for the animation in various cases, (same/different) lane(s) (in seconds).}
\label{tab:animation_time_sep}
\end{table}

The rest of the vehicle parameters ($v_{\textit{max}},\ a_{\textit{max},c},\ a_{\textit{max},tr}$) are the same as in Table \ref{tab:parameters}. The scenario parameters are presented in Table \ref{tab:animation_parameters}.

We consider four instances -- two with medium traffic loads and two with heavy traffic loads. Out of the two instances for each traffic load regime, one instance is with only cars, as a standard of comparison. To help visualise these scenarios further, we also have plots with the specific trajectories seen in the corresponding animations in Figures \ref{fig:animation_cars_trucks} and \ref{fig:animation_cars}.

\begin{table}[!h]
\centering
\begin{tblr}{
  cells = {c},
  cell{6}{2} = {c=2}{},
  cell{6}{4} = {c=2}{},
  cell{7}{2} = {c=4}{},
  cell{8}{2} = {c=4}{},
  hlines,
  vlines,
}
Parameter (Unit) & {Medium load,\\trucks and cars} & {High load,\\trucks and cars} & {Medium load,\\only cars} & {High load,\\only cars}\\
{Arrival rate parameter\\lane 1 (s$^{-1}$)} & $0.41$ & $0.59$ & $0.59$ & $0.86$\\
Traffic load, lane 1 & $0.3562$ & $0.4826$ & $0.2581$ & $0.4943$ \\
{Arrival rate parameter\\lane 2 (s$^{-1}$)} & $0.42$ & $0.62$ & $0.62$ & $0.86$\\
Traffic load, lane 2 & $0.3638$ & $0.5016$ & $0.2640$ & $0.4943$ \\
Fraction of trucks & $0.4$ &  & $0$ & \\
Control region length (m) & $700$ &  &  &
\end{tblr}
\caption{Parameters for 3D animations with cars and trucks.}
\label{tab:animation_parameters}
\end{table}

\begin{figure}[!h]
    \centering
    \begin{subfigure}{0.47\linewidth}
        \centering
        \includegraphics[width = \linewidth]{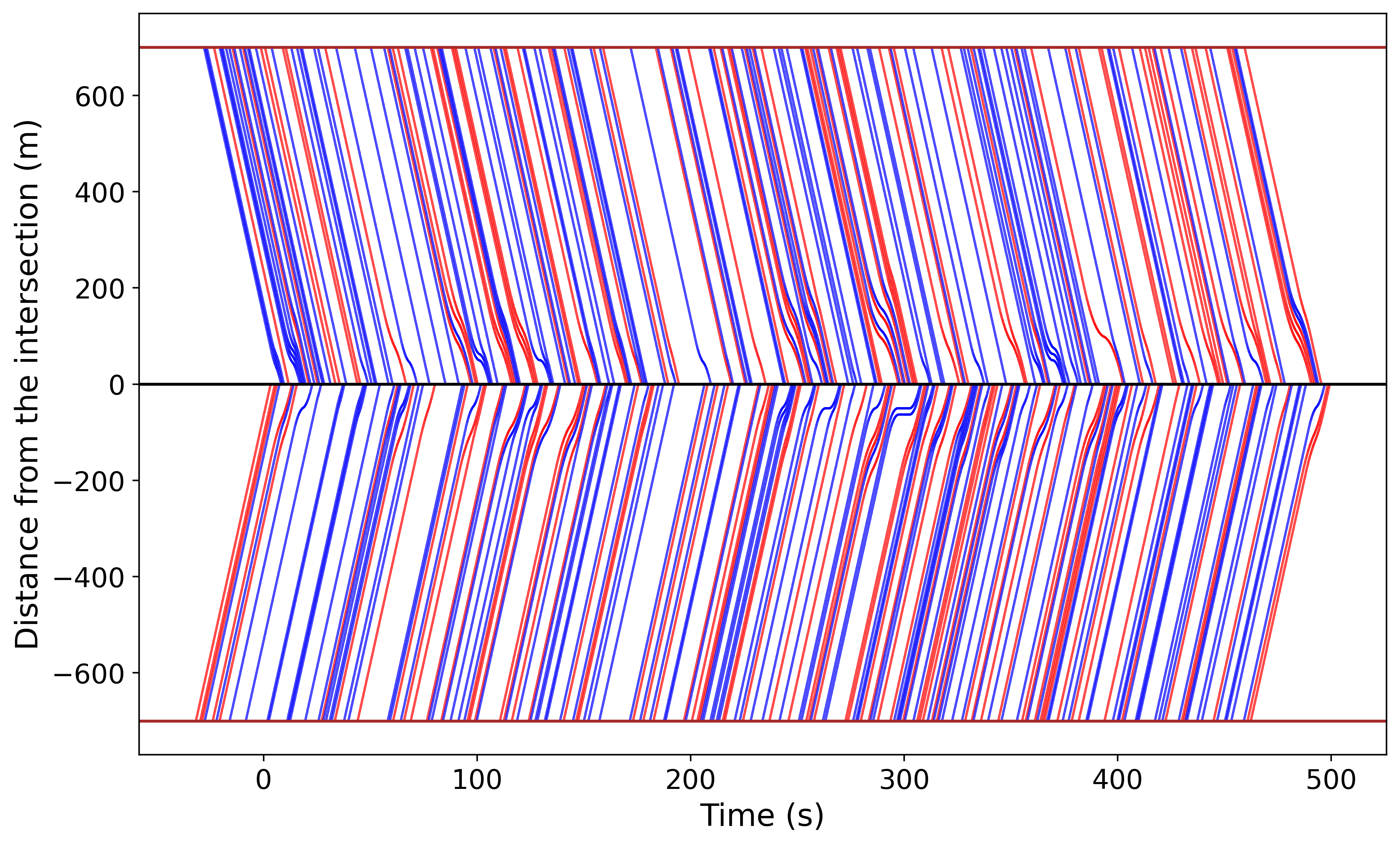}
        \caption{Medium traffic load situation}
    \end{subfigure}
    \begin{subfigure}{0.47\linewidth}
        \centering
        \includegraphics[width=\linewidth]{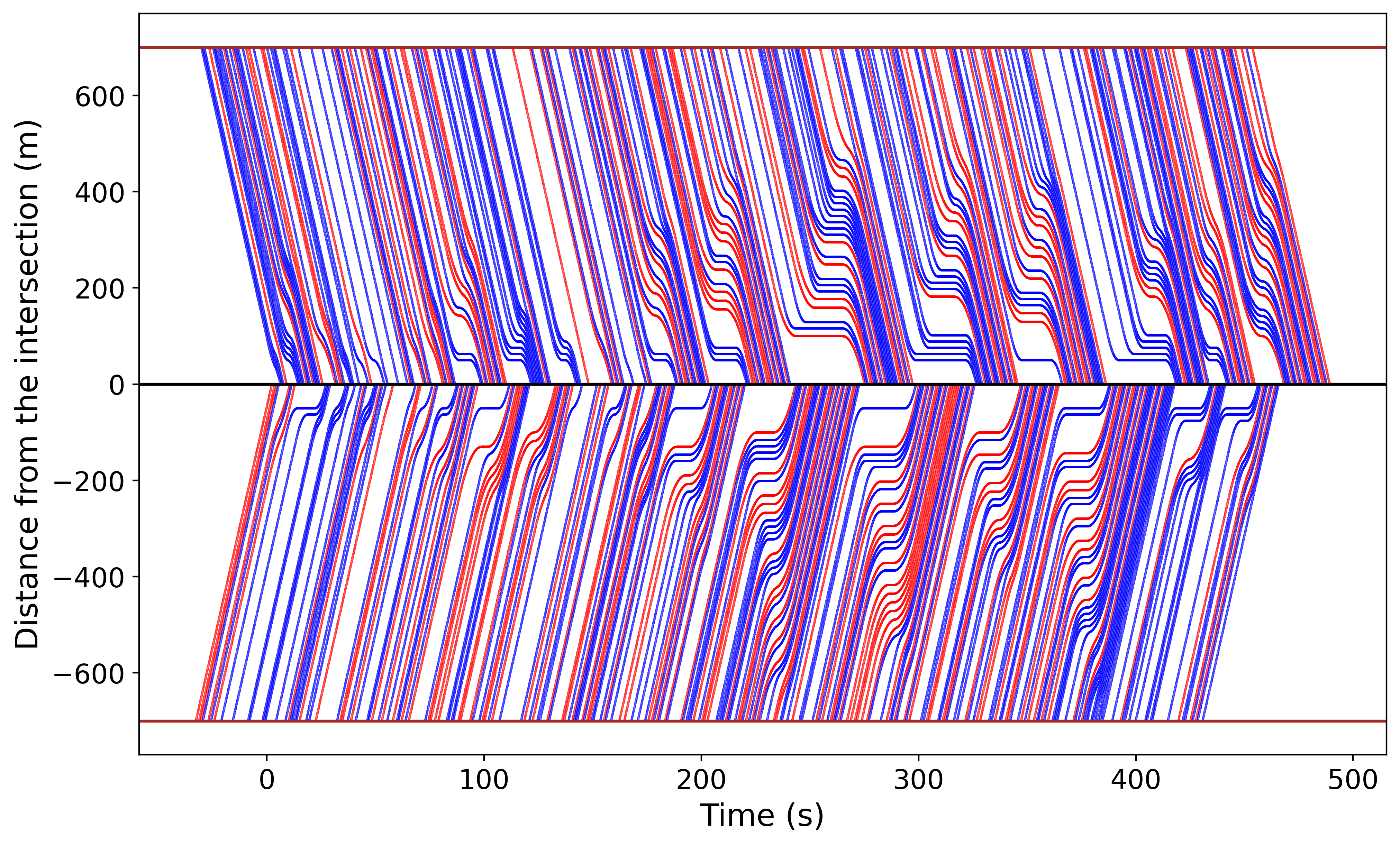}
        \caption{High traffic load situation}
    \end{subfigure}
    \caption{Platooning with cars and trucks}
    \label{fig:animation_cars_trucks}
\end{figure}
\begin{figure}[!t]
    \centering
    \begin{subfigure}{0.47\linewidth}
        \centering
        \includegraphics[width = \linewidth]{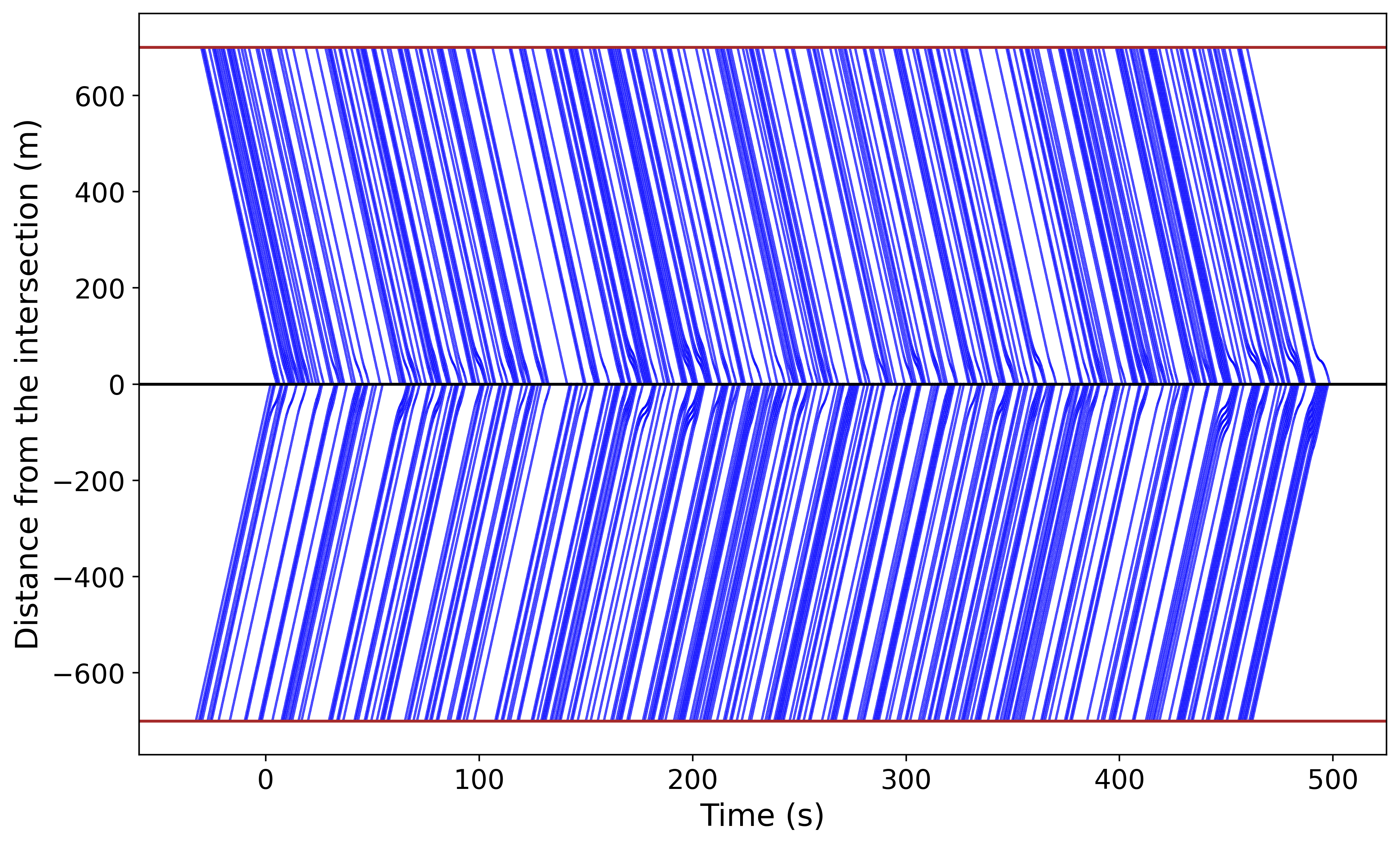}
        \caption{Medium traffic load situation}
    \end{subfigure}
    \begin{subfigure}{0.47\linewidth}
        \centering
        \includegraphics[width=\linewidth]{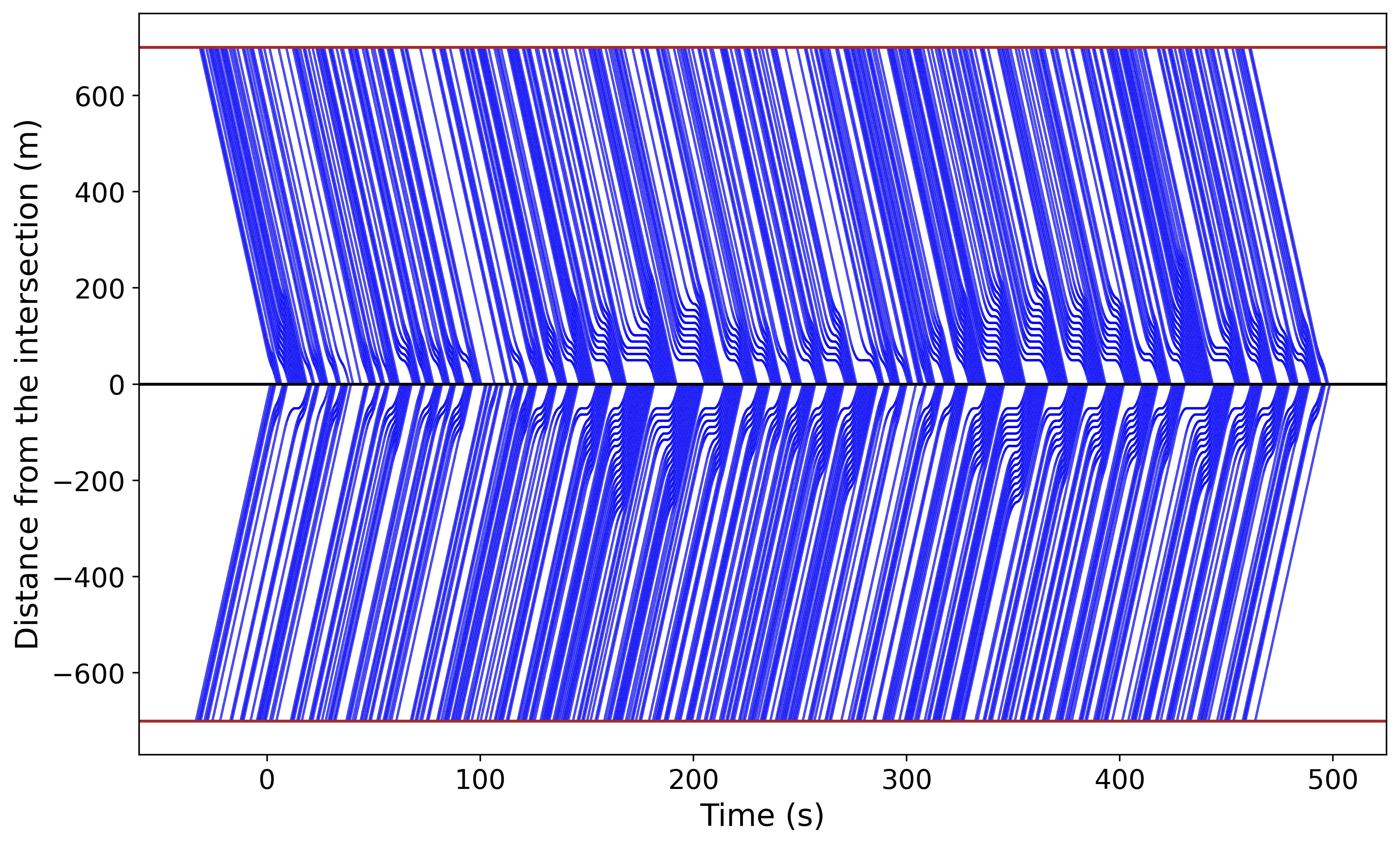}
        \caption{High traffic load situation}
    \end{subfigure}
    \caption{Platooning with only cars}
    \label{fig:animation_cars}
\end{figure}

As we can see, in the medium traffic load cases, the delay experienced vehicles is less, leading to smaller platoons. Also, the proportion of vehicles coming to a full stop in the control region is smaller compared to the high traffic load situation. Whereas in the situations with high traffic load, the platoons get larger in size, and the proportion of vehicles being delayed enough to warrant a full stop also increases. The figures moreover aid to appreciate the difference in the platooning and the trajectories with the addition of traffic heterogeneity. The load that can be sustained by the system changes dramatically, and longer control regions are also necessary with the addition of trucks to the vehicle pool.

\section{Speed-profiling algorithm: Closed-form expressions}

\subsection{Closest preceding truck does not decelerate to a full stop}\label{app:SPA_closed_form}

Similar to Section~\ref{sec:SPA_closed_form}, we wish to obtain the closed-form expressions for the trajectory of a car when its closest preceding truck does not decelerate to a full stop. Suppose that the car under consideration and its closest preceding truck occupy the $i^{th}$ and $j^{th}$ positions in a platoon, respectively, and let $u_2 > 0$ be the lowest speed attained by this truck in the control region. At the time of determining the trajectory of the current car, $u_2$ has already been calculated, and is equal to (also see \cite{Timmerman2021PlatoonIntersections}):
\begin{equation*}
    u_2 = v_{\textit{max}} - \sqrt{a_{\textit{max},tr}\left[v_{\textit{max}}(t_{f,j} - t_{0,j}) - x_0 \right]},
\end{equation*}
where $t_{f,j}$ and $t_{0,j}$ are the crossing and entering times of the closest preceding truck. Based on our observations of optimal trajectories obtained from the joint optimisation problem, we can classify the trajectory of the car into three categories, which are treated on a case-by-case basis below (also see Figure \ref{fig:SPA_Traj_truck_no_stop}).
\begin{figure}[!t]
    \centering
    \includegraphics[width = 0.8\linewidth]{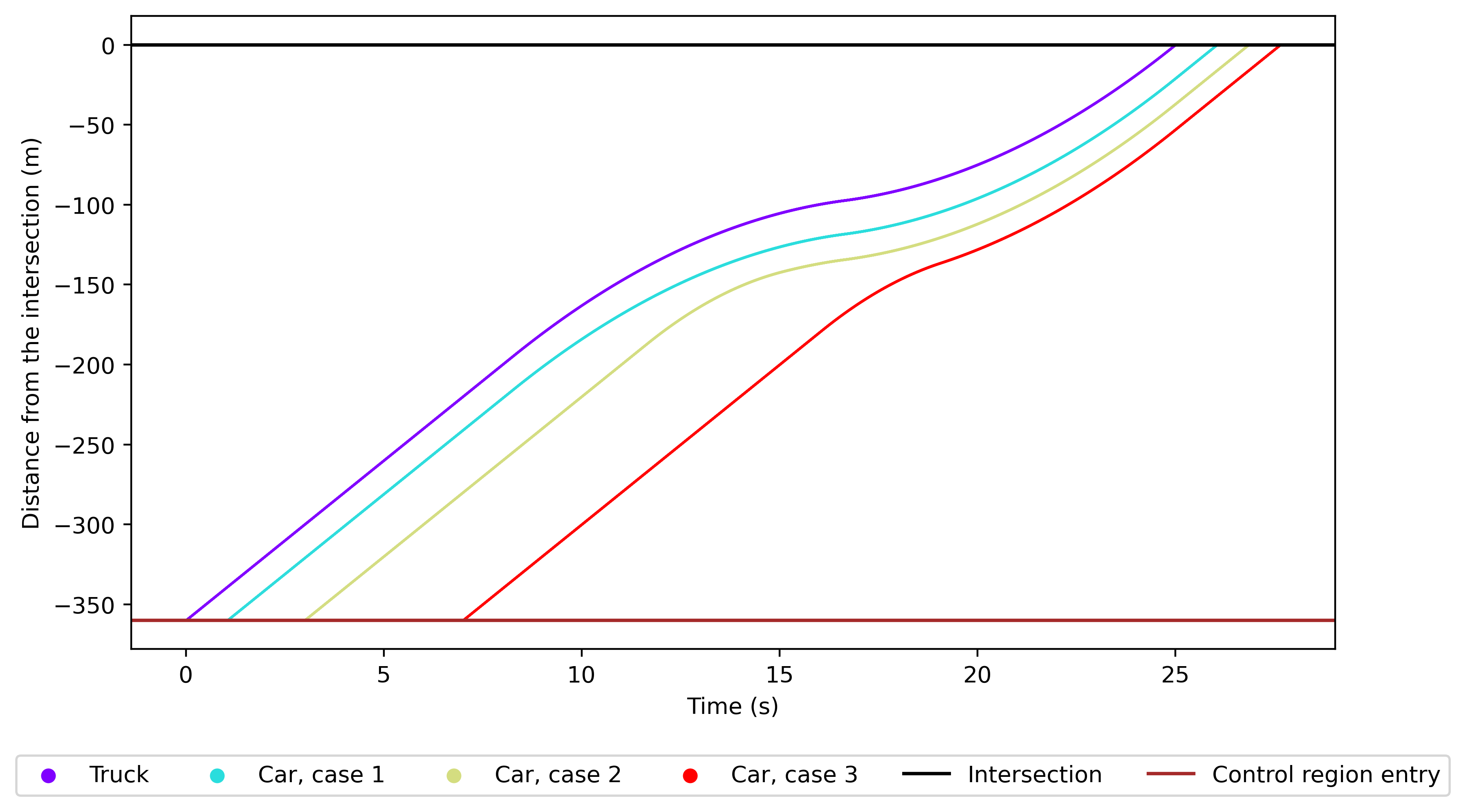}
    \caption{Trajectories of truck and car in various cases, given that the truck does not stop in the control region.}
    \label{fig:SPA_Traj_truck_no_stop}
\end{figure}

\textbf{1) Car follows truck exactly}
\begin{figure}
    \centering
    \includegraphics[width = 0.8\linewidth]{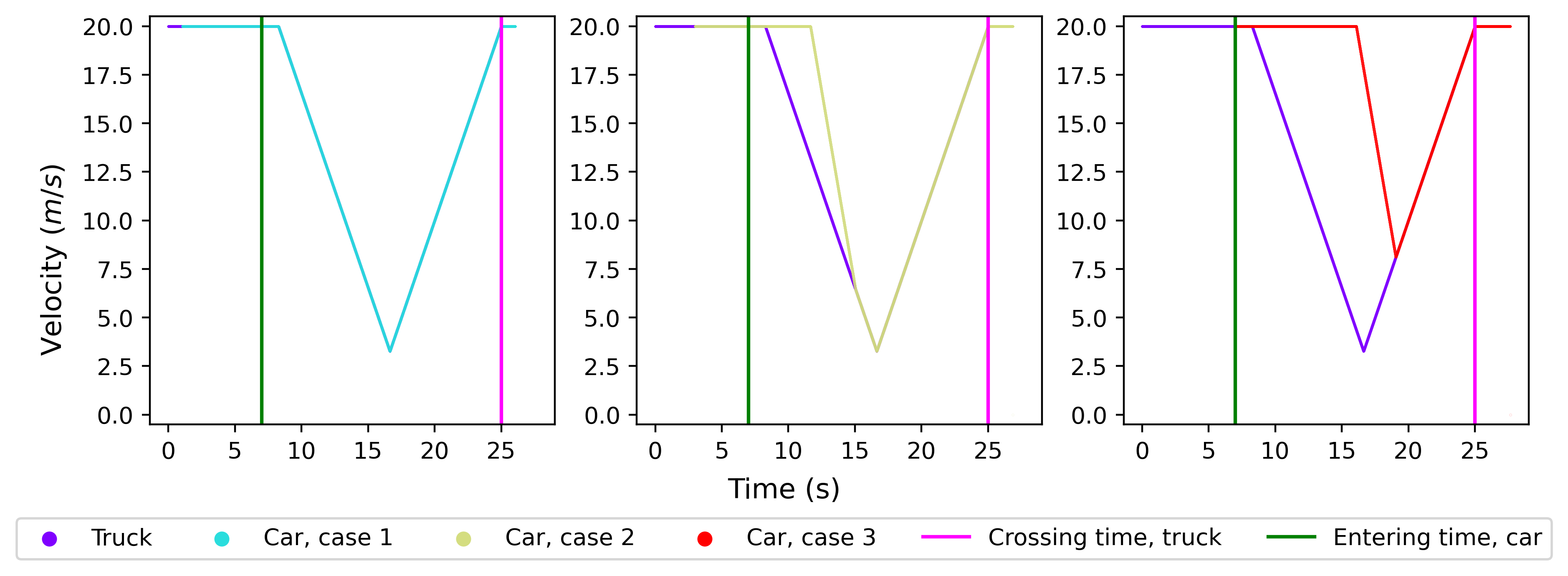}
    \caption{Velocity profile of car and truck in various cases, given that the truck does not stop in the control region.}
    \label{fig:SPA_Velo_truck_no_stop}
\end{figure}
When $t_{f,i} - t_{0,i} = t_{f,j} - t_{0,j}$, the velocity and acceleration functions of the car coincide with those of its closest preceding truck:\\
\begin{minipage}[!t]{0.45\linewidth}
    \centering
    \begin{equation*}
    \ddot{x}_i(s) = 
    \begin{dcases}
        \ddot{x}_j(s) \quad &s \in [t_{0,i}, t_{f,j}],\\
        0 \quad &s \in [t_{f,j}, t_{f,i}],
    \end{dcases}
\end{equation*}
\end{minipage}
\begin{minipage}[!t]{0.45\linewidth}
    \centering
    \begin{equation*}
        \dot{x}_i(s) =
        \begin{dcases}
            \dot{x}_j(s) \quad &s \in [t_{0,i}, t_{f,j}],\\
            v_{\textit{max}} \quad &s \in [t_{f,j}, t_{f,i}].
        \end{dcases}
    \end{equation*}
\end{minipage}

The velocity of the car is also represented on the leftmost plot in Figure \ref{fig:SPA_Velo_truck_no_stop}. The trajectory of the car can be obtained by integrating the velocity function $\dot{x}_i$ of the car, with the appropriate limits:
\begin{equation}\label{eq:spa_cs_truck_nostop_1}
    x_i(s) = 
    \begin{dcases}
        x_j(s) - v_{\textit{max}} \sum_{k = j}^{i-1} \tau_{\psi(k), \psi(k+1)} & s \in [t_{0,i}, t_{f,j}],\\
        -v_{\textit{max}} \sum_{k = j}^{i-1} \tau_{\psi(k), \psi(k+1)} + v_{\textit{max}}(t - t_{f,j}) & s \in [t_{f,j}, t_{f,i}].
    \end{dcases}
\end{equation}
Finally, we recall the definition of $T^*_i$ from Equation~\eqref{eq:t_star}. Here, $T^*_i = t_{0,i}$, as in Case 1 of Section~\ref{sec:SPA_closed_form}.

\textbf{2) Car switches decelerations}
As in Section~\ref{sec:SPA_closed_form}, the acceleration dynamics of the car moves away from that of the truck (and towards the optimal acceleration profile of a car), if the delay incurred by the car is less than that of the truck. In this situation (the exact conditions will be detailed at the end), the car switches decelerations, leading to the acceleration profile below. As before, the velocity and position function have been obtained by integrating the acceleration profile once and twice, respectively, and by substituting the proper limits:\\
\begin{minipage}[t!]{0.45\linewidth}
\begin{equation*}
    \ddot{x}_i(s) = 
    \begin{dcases}
    0 & s \in [t_{0,i}, t_{\textit{dec},i}),\\
    -a_{\textit{max},c} & s \in [t_{\textit{dec},i}, t_{\textit{sw},i})\\
    -a_{\textit{max},tr} & s \in [t_{\textit{sw},i}, t_{\textit{acc},i}),\\
    a_{\textit{max},tr} & s \in [t_{\textit{acc},i}, t_{f,1}),\\
    0 & s \in [t_{f,1}, t_{f,i}],
    \end{dcases}
\end{equation*}
\end{minipage}
\begin{minipage}[!t]{0.45\linewidth}
\vspace*{-\baselineskip}
\begin{equation*}
    \dot{x}_i(s) = 
    \begin{dcases}
    v_{\textit{max}} & s \in [t_{0,i}, t_{\textit{dec},i}),\\
    v_{\textit{max}} - a_{\textit{max},c}(t - t_{\textit{dec},i}) & s \in [t_{\textit{dec},i}, t_{\textit{sw},i})\\
    u_1 - a_{\textit{max},tr}(t - t_{\textit{sw},i}) & s \in [t_{\textit{sw},i}, t_{\textit{acc},i}),\\
    u_2 + a_{\textit{max},tr} & s \in [t_{\textit{acc},i}, t_{f,1}),\\
    v_{\textit{max}} & s \in [t_{f,1}, t_{f,i}],
    \end{dcases}
\end{equation*}
\end{minipage}

\begin{equation*}
    x_i(s) = 
    \begin{dcases}
    -x_0 + v_{\textit{max}} (t - t_{0,i}) & s \in [t_{0,i}, t_{\textit{dec},i}],\\
    -x_0 + v_{\textit{max}} (t - t_{0,i}) - \frac{1}{2} a_{\textit{max},c}(t - t_{\textit{dec},i})^2 & s \in [t_{\textit{dec},i}, t_{\textit{sw},i}],\\
    -x_0 + v_{\textit{max}}(t_{\textit{dec},i} - t_{0,i}) + \frac{v_{\textit{max}}^2 - u_1^2}{2a_{\textit{max},c}} + u_1(t - t_{\textit{sw},i}) - \frac{1}{2} a_{\textit{max},tr}(t - t_{\textit{sw},i})^2 & s \in [t_{\textit{sw},i}, t_{\textit{stop},i}],\\
    -x_0 + v_{\textit{max}}(t_{\textit{dec},i} - t_{0,i}) + \frac{v_{\textit{max}}^2 - u_1^2}{2a_{\textit{max},c}} + \frac{u_1^2 - u_2^2}{2a_{\textit{max},tr}} & s \in [t_{\textit{stop},i}, t_{\textit{acc},i}],\\
    -x_0 + v_{\textit{max}}(t_{\textit{dec},i} - t_{0,i}) + \frac{v_{\textit{max}}^2 - u_1^2}{2a_{\textit{max},c}} + \frac{u_1^2 - u_2^2}{2a_{\textit{max},tr}} + u_2(t - t_{\textit{acc},i})  & s \in [t_{\textit{acc},i}, t_{f,1}],\\
    \quad + \frac{1}{2} a_{\textit{max},tr}(t - t_{\textit{acc},i})^2 & \\
    -x_0 + v_{\textit{max}}(t_{\textit{dec},i} - t_{0,i}) + \frac{v_{\textit{max}}^2 - u_1^2}{2a_{\textit{max},c}} + \frac{u_1^2 + v_{\textit{max}}^2 - 2u_2^2}{2a_{\textit{max},tr}} + v_{\textit{max}}(t - t_{f,1}) & s \in [t_{f,1}, t_{f,i}],
    \end{dcases}
\end{equation*}
where $u_1$ is the speed of the car at time $t_{\textit{sw},i}$, i.e., while switching decelerations. Note that $T^*_i = t_{\textit{sw},i}$ in this case. Coupling the relation $t_{\textit{dec},i} = t_{f,1} - \frac{v_{\textit{max}} + u_1 - 2u_2}{a_{\textit{max},tr}} - \frac{v_{\textit{max}} - u_1}{a_{\textit{max},c}}$ with the positional terminal condition for the car, we obtain a quadratic equation in $u_1$, whose positive solution is:
\begin{equation*}
    \begin{aligned}
    u_1 &= v_{\textit{max}} - \sqrt{ \frac{2 a_{\textit{max},tr} a_{\textit{max},c} \left( x_0 - v_{\textit{max}} (t_{f,i} - t_{0,i}) + \frac{ (v_{\textit{max}} - u_2)^2}{a_{\textit{max},tr}} \right)}{a_{\textit{max},c} - a_{\textit{max},tr}} },\\
    &= v_{\textit{max}} - \sqrt{ \frac{2 a_{\textit{max},tr} a_{\textit{max},c} v_{\textit{max}}\Big[(t_{f,j} - t_{0,j}) - (t_{f,i} - t_{0,i})\Big]}{a_{\textit{max},c} - a_{\textit{max},tr}} }, \quad 0 < u_2 < u_1 < v_{\textit{max}}.
    \end{aligned}
\end{equation*}
At time $t_{\textit{sw},i}$, the car has caught up with the closest preceding truck and all of the vehicles in between, just as before in Section~\ref{sec:SPA_closed_form}. From then on, the distance between each consecutive pair of these vehicles is exactly the required minimum safety distance, until the vehicles exit the control region. This scenario is valid when $v_{\textit{max}} > u_1 > u_2$, leading to:
\begin{equation*}
    \boxed{ (t_{f,j} - t_{0,j})\frac{a_{\textit{max},c} + a_{\textit{max},tr}}{2a_{\textit{max},c}} < t_{f,i} - t_{0,i} < t_{f,j} - t_{0,j}.}
\end{equation*}

\textbf{3) Car catches up with truck during acceleration phase}
The final situation occurs when the car catches up with its closest preceding truck and the cars in between them, if any, during the acceleration phase. The acceleration function of the car can be written as:\\
\begin{minipage}[t!]{0.45\linewidth}
\begin{equation*}
    \ddot{x}_i(s) = 
    \begin{dcases}
    0 & s \in [t_{0,i}, t_{\textit{dec},i}),\\
    -a_{\textit{max},c} & s \in [t_{\textit{dec},i}, t_{\textit{acc},i}),\\
    a_{\textit{max},tr} & s \in [t_{\textit{acc},i}, t_{f,1}),\\
    0 & s \in [t_{f,1}, t_{f,i}],
    \end{dcases}
\end{equation*}
\end{minipage}
\begin{minipage}[!t]{0.45\linewidth}
\vspace*{-\baselineskip}
\begin{equation*}
    \text{and, } \dot{x}_i(s) = 
    \begin{dcases}
    v_{\textit{max}} & s \in [t_{0,i}, t_{\textit{dec},i}],\\
    v_{\textit{max}} - a_{\textit{max},c}(t - t_{\textit{dec},i}) & s \in [t_{\textit{dec},i}, t_{\textit{acc},i}],\\
    u_1 + a_{\textit{max},tr}(t - t_{\textit{acc},i}) & s \in [t_{\textit{acc},i}, t_{f,1}],\\
    v_{\textit{max}} & s \in [t_{f,1}, t_{f,i}],
    \end{dcases}
\end{equation*}
\end{minipage}
\begin{equation*}
    x_i(s) = 
    \begin{dcases}
    -x_0 + v_{\textit{max}} (t - t_{0,i}) & s \in [t_{0,i}, t_{\textit{dec},i}],\\
    -x_0 + v_{\textit{max}} (t - t_{0,i}) - \frac{1}{2} a_{\textit{max},c}(t - t_{\textit{dec},i})^2 & s \in [t_{\textit{dec},i}, t_{\textit{acc},i}],\\
    -x_0 + v_{\textit{max}} (t_{\textit{dec},i} - t_{0,i}) + \frac{v_{\textit{max}}^2 - u_1^2}{2a_{\textit{max},c}} + u_1(t - t_{\textit{acc},i}) + \frac{1}{2} a_{\textit{max},tr}(t - t_{\textit{acc},i})^2 & s \in [t_{\textit{acc},i}, t_{f,1}],\\
    -x_0 + v_{\textit{max}} (t_{\textit{dec},i} - t_{0,i}) + \frac{v_{\textit{max}}^2 - u_1^2}{2}\left(\frac{1}{a_{\textit{max},tr}} + \frac{1}{a_{\textit{max},c}} \right) + v_{\textit{max}}(t - t_{f,1}) & s \in [t_{f,1}, t_{f,i}].
    \end{dcases}
\end{equation*}
where $u_1$ is the lowest speed attained by the car while in the control region. To be specific, at time $t_{\textit{acc},i}$ the car has caught up with all its preceding vehicles up to the closest preceding truck, thus leading to the relation $T^*_i = t_{\textit{acc},i}$. By substituting $t_{\textit{dec},i} = t_{f,1} - (v_{\textit{max}} - u_1) \left( \frac{1}{a_{\textit{max},c}} + \frac{1}{a_{\textit{max},tr}}\right)$, and utilising $x_i(t_{f,i}) = 0$ (the initial positional condition for the car), we obtain a quadratic equation in $u_1$, which yields:
\begin{equation*}
    u_1 = v_{\textit{max}} - \sqrt{ \frac{2 a_{\textit{max},c} a_{\textit{max},tr} (v_{\textit{max}} (t_{f,i} - t_{0,i}) - x_0)}{a_{\textit{max},c} + a_{\textit{max},tr}} }.
\end{equation*}
This then describes the entire trajectory of the car. This situation occurs when the condition below is satisfied:
\begin{equation*}
    \boxed{ \frac{x_0}{v_{\textit{max}}} \leq t_{f,i} - t_{0,i} \leq (t_{f,j} - t_{0,j})\frac{a_{\textit{max},c} + a_{\textit{max},tr}}{2a_{\textit{max},c}}.}
\end{equation*}
This condition was obtained by requiring $ v_{\textit{max}} \geq u_1 \geq u_2$.

Again, we notice that the quantity $\Delta_i \coloneqq t_{f,i} - t_{0,i}$ for the car holding the $i^{th}$ position in the platoon decides the trajectory shape. We sketch the case distinction for the car trajectory for the situation when the closest preceding truck does not come to a stop in Figure \ref{fig:delay_cpbt_no_stop}. As before, the quantity $\Delta_i$ can only lie in the interval $[\frac{x_0}{v_{\textit{max}}}, \Delta_j]$.

\begin{figure}[t!]
    \centering
    \includegraphics[width = 0.9\linewidth]{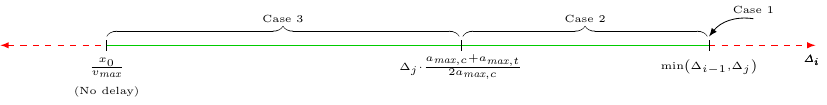}
    \caption{Case distinction for a car preceded by a truck in its platoon, given that the truck does not come to a stop in the control region.}
    \label{fig:delay_cpbt_no_stop}
\end{figure}

\subsection{Optimality of trajectories}\label{app:spa_optimality}
In Section~\ref{sec:SPA_closed_form} and Appendix~\ref{app:SPA_closed_form}, we obtained closed-form expressions for vehicle trajectories based on the solutions of the corresponding joint optimisation problems. However, it remains to be seen whether these expressions are in fact associated with optimal trajectories. Thus, in this section, we will show that closed-form expressions we have obtained are safe (according to \thref{def:safety}) and optimal (in the sense of Equations~\eqref{eq:opt_problem_general_form}-\eqref{eq:optimization_criterion_capacity}).

We prove that the safety constraint is satisfied for car trajectories when preceded in a platoon for a truck, since the other cases (truck preceded by car/truck, car preceded only by cars) have already been dealt with in (or can be easily extended from) \cite{Timmerman2021PlatoonIntersections}.

\begin{proposition}[Safety]\label{prop:safety_spa_cs}
    For every feasible instance of a car preceded by a truck in a platoon, the closed-form expressions for the car trajectory described in Section~\ref{sec:SPA_closed_form} or in Appendix~\ref{app:SPA_closed_form} above (as the case may be), satisfy the safety constraint stated in \thref{def:safety}.
\end{proposition}
\begin{proof}
    Let us pick a car, such that its trajectory has a closed-form expression as in either Section~\ref{sec:SPA_closed_form} or Appendix~\ref{app:SPA_closed_form}. We need to show that the trajectory is safe with respect to its preceding vehicle. That is, if the car occupies position $i$ in its platoon:
    \begin{equation}\label{eq:safety_lemma_spa}
        x_{i-1}(s) - x_i(s) \geq v_{\textit{max}} \tau_{\psi(i-1),1}, \quad \forall\ s \in [t_{0,i},t_{f,i-1}],
    \end{equation}
    where $\psi$ (see Equation~\eqref{eq:type_function}) denotes the type of a vehicle. Note that the position function for vehicles takes values in the negative real line when they are in the control region, thus the absolute values simplify to lead to the expression above.  Recall that we are dealing with a car that is preceded in its platoon by a truck. Further, we use elements from the trajectory of the closest preceding truck in order to construct the trajectory of the car under consideration. Hence, the vehicle that immediately precedes the car can either be:
    \begin{enumerate}
        \item the closest preceding truck, or
        \item a car.
    \end{enumerate}
    
    1) We first look at the first possibility, that is, when the immediately preceding vehicle is a truck (at the $(i-1)^{th}$ position). Recalling the definition of $T^*_i$ from Equation~\eqref{eq:t_star}, we can easily infer that the trajectories of the car and truck are safe with respect to each other from time $T^*_i$ onward, since the velocity of the car exactly matches that of the truck for the interval $[T^*_i, t_{f,i-1}]$, and the distance between them is exactly $v_{\textit{max}}\tau_{1,0}$ at time $T^*_i$. Thus we only need to show that the trajectories of the car and truck are safe in the interval $[t_{0,i}, T^*_i)$.
    
    We can write the evolution of the position of the truck (and similarly for the car) as:
    \begin{equation*}
        x_{i-1}(T^*_i) = x_{i-1}(s) + \int_s^{T^*_i} \dot{x}_{i-1}(w) dw, \qquad x_{i}(T^*_i) = x_{i}(s) + \int_s^{T^*_i} \dot{x}_{i}(w) dw,
    \end{equation*}
    for $s$ in the interval $[t_{0,i}, T^*_i]$. Subtracting the second equation from the first leads to:
    \begin{equation*}
        x_{i-1}(T^*_i) -  x_{i}(T^*_i) = x_{i-1}(s) - x_{i}(s) + \int_s^{T^*_i} \left(\dot{x}_{i-1}(w) - \dot{x}_{i}(w)\right) dw.
    \end{equation*}
    Now, we know that $x_{i-1}(T^*_i) -  x_{i}(T^*_i)$ is exactly the required minimum safety distance $v_{\textit{max}} \tau_{1,0}$, since the truck is immediately preceding the car. Further, we observe that for each time instant $s$ in the interval $[t_{0,i}, T^*_i]$, the relation
    \begin{equation*}
        \dot{x}_{i}(s) - \dot{x}_{i-1}(s) \geq 0,
    \end{equation*}
    holds, no matter the case distinction used to construct the trajectory of the car. On rearranging terms, we see:
    \begin{equation*}
            x_{i-1}(s) - x_{i}(s) = v_{\textit{max}} \tau_{1,0} + \int_s^{T^*_i} \left(\dot{x}_{i}(w) - \dot{x}_{i-1}(w)\right) dw \geq  v_{\textit{max}} \tau_{1,0}  \quad \text{for } s \in [t_{0,i}, T^*_i],
    \end{equation*}
    thus satisfying Equation~\eqref{eq:safety_lemma_spa}.

   2) For the case when the vehicle immediately preceding the car under consideration is also a car, we can prove safety in a similar manner, with the help of one key observation. Because delays in a platoon are non-increasing from the head to the tail of the platoon, if the immediately preceding car has a trajectory falling in Case $k$ of Section~\ref{sec:SPA_closed_form} ($k = 1, 2, 3, 4$ for Section~\ref{sec:SPA_closed_form} and $k = 1, 2, 3$ for Appendix~\ref{app:SPA_closed_form}), then the trajectory of the following car can only fall within Case $k_1$, where $k_1 \geq k$. This means that here too, we have the relation:
   \begin{equation*}
        \dot{x}_{i}(s) - \dot{x}_{i-1}(s) \geq 0 \quad \text{for } s \in [t_{0,i}, T^*_i],
    \end{equation*}
    which means that the arguments used to prove safety in the first case can be directly used here as well.
\end{proof}

\paragraph{A stronger optimality criterion}\leavevmode

The vehicle trajectories generated through the joint optimisation problem in Section~\ref{sec:SPA_optimisation} were optimised for the objective function given in Equation~\eqref{eq:optimization_criterion_capacity}.
That is, the optimal trajectories collectively have the smallest area under the position curve in the interval $[t_{0,i}, t_{f,i}]$. Informally, the vehicles with these optimal trajectories will be closest to the intersection for the greatest amount of time in the interval $[t_{0,i}, t_{f,i}]$, compared to any other set of trajectories that are feasible (i.e.\ from the class $\Pi_c$ or $\Pi_{\textit{tr}}$, as the case may be) and safe. Having shown that the trajectories with the closed-form expressions are safe w.r.t.\ to their preceding vehicle, we now turn our attention to showing that the trajectories whose expressions we have obtained are optimal as well. We will show this in two steps. First, we will show that the trajectories for which we have obtained closed-form expressions are optimal if optimised in a sequential manner (described in more detail later) and then we will show that the trajectories that are optimal according to the sequential optimisation procedure are also optimal for the \emph{joint} optimisation problem. 

Accordingly, we first show that the trajectories whose closed-form expressions we have obtained are optimal for a multi-step, sequential optimisation procedure, as done in \cite{Miculescu2020Polling-Systems-BasedSignals} and \cite{Timmerman2021PlatoonIntersections}. In fact, it turns out that these trajectories are optimal in a stronger sense. Indeed, we claim that \emph{any} other feasible and safe trajectory can be no closer to the intersection than the trajectories described by the expressions in Section~\ref{sec:SPA_closed_form} and in Appendix~\ref{app:SPA_closed_form}, at every time instant in the interval $[t_{0,i}, t_{f,i}]$, given the trajectory of the preceding vehicle.

Our closed-form expressions align very well with the sequential process of obtaining optimal trajectories, especially while obtaining the trajectory of a car preceded by a truck in the platoon -- the expressions for the car trajectory involve elements from the trajectory of this closest preceding truck. Apart from such cars however, the rest of the trajectories can be obtained independently of the trajectories of other vehicles, and only information of the delay of preceding vehicles is usually necessary.

Since this criterion is different than the one considered before by \textcite{Miculescu2020Polling-Systems-BasedSignals}, and \textcite{Timmerman2021PlatoonIntersections}, we still need to show that the trajectories are optimal also for a truck or for the case when a car is preceded in the platoon by only cars, in addition to the trajectories from Section~\ref{sec:SPA_closed_form} and Appendix~\ref{app:SPA_closed_form}. Recall that these are the same trajectories with the closed-form expressions described by \cite{Timmerman2021PlatoonIntersections}. For the sake of completeness, we write down the expressions for the trajectory of a vehicle in either of these cases, by following the same procedure as before. First, if the vehicle makes a full stop in the control region (see \cite{Timmerman2021PlatoonIntersections}), the acceleration, velocity and position functions can be expressed as:\\
\begin{minipage}{0.43\linewidth}
\centering
\begin{equation*}
    \ddot{x}_i(s) = 
    \begin{dcases}
        0 & s \in [t_{0,i}, t_{\textit{dec},i}),\\
        -a_{\textit{max},\dagger} & s \in [t_{\textit{dec},i}, t_{\textit{stop},i}),\\
        0 & s \in [t_{\textit{stop},i}, t_{\textit{acc},i}),\\
        a_{\textit{max},\dagger} & s \in [t_{\textit{acc},i}, t_{f,1}),\\
        0 & s \in [t_{f,1}, t_{f,i}],
    \end{dcases}
\end{equation*}
\end{minipage}
\begin{minipage}{0.49\linewidth}
    \centering
    \begin{equation*}
        \dot{x}_i(s) = 
    \begin{dcases}
        v_{\textit{max}} & s \in [t_{0,i}, t_{\textit{dec},i}),\\
        v_{\textit{max}} -a_{\textit{max},\dagger}(t - t_{\textit{dec},i}) & s \in [t_{\textit{dec},i}, t_{\textit{stop},i}),\\
        0 & s \in [t_{\textit{stop},i}, t_{\textit{acc},i}),\\
        a_{\textit{max},\dagger}(t - t_{\textit{acc},i}) & s \in [t_{\textit{acc},i}, t_{f,1}),\\
        v_{\textit{max}} & s \in [t_{f,1}, t_{f,i}],
    \end{dcases}
    \end{equation*}
\end{minipage}
\begin{equation}\label{eq:homogeneous_stop}
    x_i(s) = 
    \begin{dcases}
        -x_0 + v_{\textit{max}}(t - t_{0,i}) & s \in [t_{0,i}, t_{\textit{dec},i}),\\
        -x_0 + v_{\textit{max}}(t - t_{0,i}) -\frac{1}{2}a_{\textit{max},\dagger}(t - t_{\textit{dec},i})^2 & s \in [t_{\textit{dec},i}, t_{\textit{stop},i}),\\
        -x_0 + v_{\textit{max}}(t_{\textit{stop},i} - t_{0,i}) -\frac{v_{\textit{max}}^2}{2a_{\textit{max},\dagger}} & s \in [t_{\textit{stop},i}, t_{\textit{acc},i}),\\
        -x_0 + v_{\textit{max}}(t_{\textit{stop},i} - t_{0,i}) -\frac{v_{\textit{max}}^2}{2a_{\textit{max},\dagger}} + \frac{1}{2}a_{\textit{max},\dagger}(t - t_{\textit{acc},i})^2 & s \in [t_{\textit{acc},i}, t_{f,1}),\\
        -x_0 + v_{\textit{max}}(t_{\textit{stop},i} - t_{0,i}) + v_{\textit{max}}(t - t_{f,1}) & s \in [t_{f,1}, t_{f,i}],
    \end{dcases}
\end{equation}
where $\dagger = c$ if $\psi(i) = 0$ and $\dagger = t$ otherwise. The positional terminal condition leads to an expression for $t_{\textit{stop},i}$, which can be used to compute $t_{\textit{dec},i}$:
\begin{equation*}
    t_{\textit{stop},i} = \frac{x_0}{v_{\textit{max}}} - (t_{f,i} - t_{f,1}) + t_{0,i}.
\end{equation*}
This trajectory results if the following conditions are satisfied:
\begin{equation*}
    t_{\textit{stop},i} \leq t_{\textit{acc},i} \quad \text{and} \quad t_{\textit{dec},i} \geq t_{0,i},
\end{equation*}
leading to,
\begin{equation}\label{eq:homogeneous_stop_conditions}
    \quad \boxed{\frac{x_0}{v_{\textit{max}}} + \frac{v_{\textit{max}}}{a_{\textit{max},\dagger}} \leq t_{f,i} - t_{0,i}} \quad \text{and} \quad \boxed{ t_{f,i} - t_{f,1} \leq \frac{x_0}{v_{\textit{max}}} - \frac{v_{\textit{max}}}{a_{\textit{max},\dagger}}.}
\end{equation}
Also, as we observed that the delay in the polling model is non-increasing in a platoon as we go from the head to the tail, we have $t_{f,i-1}-t_{0,i-1}$ as an upper bound for $t_{f,i} - t_{0,i}$.

Similarly, if the vehicle does not make a stop, but instead decelerates to some minimum speed $u$ and then accelerates to the maximum permissible speed $v_{\textit{max}}$ while in the control region, we have:\\
\begin{minipage}{0.43\linewidth}
\centering
\begin{equation*}
    \ddot{x}_i(s) = 
    \begin{dcases}
        0 & s \in [t_{0,i}, t_{\textit{dec},i}),\\
        -a_{\textit{max},\dagger} & s \in [t_{\textit{dec},i}, t_{\textit{acc},i}),\\
        a_{\textit{max},\dagger} & s \in [t_{\textit{acc},i}, t_{f,1}),\\
        0 & s \in [t_{f,1}, t_{f,i}],
    \end{dcases}
\end{equation*}
\end{minipage}
\begin{minipage}{0.49\linewidth}
    \centering
    \begin{equation*}
        \dot{x}_i(s) = 
    \begin{dcases}
        v_{\textit{max}} & s \in [t_{0,i}, t_{\textit{dec},i}),\\
        v_{\textit{max}} -a_{\textit{max},\dagger}(t - t_{\textit{dec},i}) & s \in [t_{\textit{dec},i}, t_{\textit{acc},i}),\\
        u + a_{\textit{max},\dagger}(t - t_{\textit{acc},i}) & s \in [t_{\textit{acc},i}, t_{f,1}),\\
        v_{\textit{max}} & s \in [t_{f,1}, t_{f,i}],
    \end{dcases}
    \end{equation*}
\end{minipage}

\begin{equation}\label{eq:homogeneous_no_stop}
    x_i(s) = 
    \begin{dcases}
        -x_0 + v_{\textit{max}}(t - t_{0,i}) & s \in [t_{0,i}, t_{\textit{dec},i}),\\
        -x_0 + v_{\textit{max}}(t - t_{0,i}) -\frac{1}{2}a_{\textit{max},\dagger}(t - t_{\textit{dec},i})^2 & s \in [t_{\textit{dec},i}, t_{\textit{acc},i}),\\
        -x_0 + v_{\textit{max}}(t_{\textit{acc},i} - t_{0,i}) -\frac{(v_{\textit{max}} - u)^2}{2a_{\textit{max},\dagger}} + u(t - t_{\textit{acc},i}) + \frac{1}{2}a_{\textit{max},\dagger}(t - t_{\textit{acc},i})^2 & s \in [t_{\textit{acc},i}, t_{f,1}),\\
        -x_0 + v_{\textit{max}}(t_{\textit{acc},i} - t_{0,i}) + u(t_{f,1} - t_{\textit{acc},i}) + v_{\textit{max}}(t - t_{f,1}) & s \in [t_{f,1}, t_{f,i}],
    \end{dcases}
\end{equation}
where $u$ is the lowest speed attained by vehicle $i$ as indicated before. By using the supplementary relation $t_{\textit{acc},i} = t_{f,1} - \frac{v_{\textit{max}} - u}{a_{\textit{max},tr}}$ and the positional terminal condition, we can obtain a quadratic equation in $u$, which leads to:
\begin{equation}\label{eq:lowest_speed_homogeneous}
    u = v_{\textit{max}} - \sqrt{ a_{\textit{max},\dagger}[v_{\textit{max}}(t_{f,i} - t_{0,i}) - x_0] }.
\end{equation}
Finally, this case occurs by restricting $u$ to be non-negative, as well as (upper) bounded by $v_{\textit{max}}$:
\begin{equation*}
    \boxed{\frac{x_0}{v_{\textit{max}}} \leq t_{f,i} - t_{0,i} < \frac{x_0}{v_{\textit{max}}} + \frac{v_{\textit{max}}}{a_{\textit{max},\dagger}}.}
\end{equation*}
Note that, when $t_{f,i} - t_{0,i}$ attains its lower bound above, the instance reduces to the vehicle driving at speed $v_{\textit{max}}$ throughout the control region. Recalling that we set $\Delta_i = t_{f,i} - t_{0,i}$, we can observe in Figure \ref{fig:delay_truck_or_car} how $\Delta_i$ affects the trajectory of the $i^{th}$ vehicle in a platoon, if it is a truck or if it is a car preceded only by cars.

We remarked in Section \ref{sec:SPA_closed_form} that the trajectories of trucks or of cars preceded by only cars in their platoon can be computed independently of other trajectories. This can be verified by observing the closed-form trajectory expressions above -- the crossing time $t_{f,1}$ is the only parameter/ quantity that does not depend on vehicle $i$. It might seem surprising that no other knowledge about any preceding vehicles is necessary, once the crossing time $t_{f,i}$ has been determined, but this can be explained as follows -- the difference $t_{f,i} - t_{f,1}$ carries implicit information about the vehicles in between the $i^{th}$ vehicle in the platoon and the head. This is because $t_{f,i}$ is calculated by all the minimum time separations required between each pair of vehicles from the head of the platoon up to the $i^{th}$ vehicle. If we know $t_{f,i}$ and $t_{f,1}$ then we know the location of the $i^{th}$ vehicle at time $t_{f,1}$ (at a distance of $v_{\textit{max}}(t_{f,i} - t_{f,1})$ from the intersection). This in turn determines the entire trajectory -- whether the vehicle will come to a stop or not and when to start decelerating, especially because these vehicles only change speed with maximum acceleration/deceleration.

\begin{figure}[!t]
    \centering
    \includegraphics[width = 0.9\linewidth]{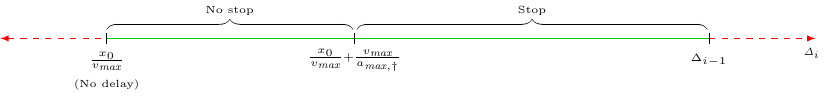}
    \caption{Case distinction for a truck or a car preceded by only cars in its platoon, where $\dagger = c$ if the associated vehicle is a car and $\dagger = t$ otherwise.}
    \label{fig:delay_truck_or_car}
\end{figure}

We now prove that the vehicle trajectories described above are optimal according to the new, stronger criterion. Consider the optimisation problems formulated by \cite{Miculescu2020Polling-Systems-BasedSignals}, which generated optimal trajectories in a sequential manner, given the optimal trajectory $x^*_{i-1}$ of the preceding vehicle:
\begin{equation}
\begin{aligned}\label{eq:miculescu_opt_prob}
\argmin_{ x_i:[t_{0,i},t_{f,i}] \to [-x_0, 0]} & \int_{t_{0,i}}^{t_{f,i}}|x_i(s)|\ ds \\
\textrm{subject to: } & x_i \in 
\begin{dcases}
\Pi_c, & \text{ if } \psi(i) = 0\\
\Pi_{\textit{tr}}, & \text{ if } \psi(i) = 1
\end{dcases} \\
& x^*_{i-1}(s) - x_i(s) \geq v_{\textit{max}} \tau_{\psi(i-1),\psi(i)} \quad \textrm{ for } s \in [t_{0,i} , t_{f,i-1}].
\end{aligned}
\end{equation}

In \thref{prop:optimal_homogeneous} and \thref{cor:optimal_inhomogeneous}, we will prove that the proposed trajectories are optimal when optimised in a sequential manner, but for a stronger optimality criterion. Then, in \thref{prop:jointly_optimal}, we show that our proposed trajectories are also an optimal solution to the optimisation problem formulated in Section \ref{sec:SPA_optimisation}.

\begin{proposition}[Optimality part 1]\label{prop:optimal_homogeneous}
    Suppose the $i^{th}$ vehicle in a platoon has a trajectory described by the position functions in either Equation~\eqref{eq:homogeneous_stop}, or Equation~\eqref{eq:homogeneous_no_stop}, henceforth denoted by $x^*_i$. Then, $x^*_i$ keeps the associated vehicle closest to the intersection compared to any other compatible trajectory, at every instant of its journey through the control region, given that the (optimal) trajectory $x^*_{i-1}$ of the preceding vehicle has been fixed. That is, the trajectory $x^*_i$ satisfies, for every trajectory $x^\pi_i \in \Pi_\dagger$ (where $\dagger = c$ if $\psi(i) = 0$ and $\textit{tr}$ otherwise), we have:    \begin{equation}\label{eq:optimality_strong_criterion_cot}
    \begin{aligned}
        x^*_{i}(s) \geq&\ x^\pi_i(s) \quad \textrm{ for all } s \in [t_{0,i} , t_{f,i-1}],\\
        \text{such that } x^*_{i-1} - y \geq&\ v_{\textit{max}} \tau_{\psi(i-1),\psi(i)}, \quad y \in \{x^*_i, x^\pi_i\}
    \end{aligned}
    \end{equation}
\end{proposition}
\begin{proof}
    The proof works by partitioning the interval $[t_{0,i},t_{f,i}]$ into four sub-intervals:
    \begin{enumerate}
        \item $[t_{0,i}, t_{\textit{dec},i})$,
        \item $[t_{\textit{dec},i}, t_{\textit{acc},i})$,
        \item $[t_{\textit{acc},i}, t_{f,1})$, and
        \item $[t_{f,1}, t_{f,i}]$,
    \end{enumerate}
    and by choosing an arbitrary \emph{compatible} trajectory $x^\pi_i$, (belonging to the same class $\Pi_c$ or $\Pi_{\textit{tr}}$ as vehicle $i$), and showing that the trajectory $x^*_i$ satisfies the condition in Equation~\eqref{eq:optimality_strong_criterion_cot} in each sub-interval, but in the order -- 1, 4, 3, 2. Further, we do not consider vehicles for which the only compatible trajectory is to drive at speed $v_{\textit{max}}$ towards the intersection throughout the control region, i.e.\ vehicles at position $i$ in a platoon such that $t_{f,i} - t_{0,i} = \frac{x_0}{v_{\textit{max}}}$.

    \textbf{1. In the sub-interval $[t_{0,i}, t_{\textit{dec},i})$:} During this sub-interval, the vehicle drives towards the intersection at speed $v_{\textit{max}}$, regardless of whether it comes to a stop later or not. Since the maximum permissible speed for both kinds of vehicles is $v_{\textit{max}}$, there is no other trajectory that can perform better in terms of getting closer to the intersection in this sub-interval. Thus, the condition in Equation~\eqref{eq:optimality_strong_criterion_cot} trivially holds for this sub-interval.

    \textbf{4. In the sub-interval $[t_{f,1}, t_{f,i}]$:} ($i>1$) At every instant while in the control region, the $i^{th}$ vehicle must be at a distance of at least $v_{\textit{max}} \tau_{\psi(i-1), \psi(i)}$ from the immediately preceding vehicle, and thus at a distance of $\sum_{k=0}^{i-1} v_{\textit{max}} \tau_{\psi(k), \psi(k+1)}$ from the head of the platoon. Thus, at time $t_{f,1}$, when the head of the platoon is exiting the control region, the $i^{th}$ vehicle must be at a distance of at least $\sum_{k=0}^{i-1} v_{\textit{max}} \tau_{\psi(k), \psi(k+1)}$ from the intersection. This is the closest position vehicle $i$ can attain at time $t_{f,1}$ while satisfying the safety constraints. In all of our optimal trajectories, we see that each vehicle attains this closest position at the crossing time of the head of its platoon (see Equations~\eqref{eq:homogeneous_stop} and \eqref{eq:homogeneous_no_stop}).
    
    Furthermore, at the instant $t_{f,i}$, vehicle $i$ must be at the intersection, and the difference between $t_{f,1}$ and $t_{f,i}$ is exactly $\sum_{k=0}^{i-1} \tau_{\psi(k), \psi(k+1)}$. Thus, if vehicle $i$ attains its closest possible position at time $t_{f,1}$, the only possible trajectory for this sub-interval involves driving at constant speed $v_{\textit{max}}$. No other trajectory $x^\pi$ can be closer to the intersection at any instant in $[t_{f,1}, t_{f,i}]$ without either violating the safety constraint or exceeding the maximum permissible speed.

    \textbf{3. In the sub-interval $[t_{\textit{acc},i}, t_{f,1})$:} We have already established that every vehicle in a platoon needs to be driving at speed $v_{\textit{max}}$ from time $t_{f,1}$ up to its own crossing time. Before time $t_{f,1}$ however, if a vehicle possesses the maximum permissible speed $v_{\textit{max}}$, then it would have to be further away from the intersection compared to the case where the velocity was less than $v_{\textit{max}}$. Hence, according to the optimality criterion, each vehicle in a platoon needs to attain speed $v_{\textit{max}}$ at $t_{f,1}$, while maintaining a slower speed prior to $t_{f,1}$. This can be achieved if vehicles undergo acceleration. For an arbitrary time point $s$ just before $t_{f,1}$, the velocity of vehicle $i$ corresponding to an arbitrary trajectory $x_i$ can be expressed by using the continuity of $\dot{x}_i$:
\begin{equation}\label{eq:lower_bound_velo_homogeneous}
\begin{aligned}
   v_{\textit{max}} = \dot{x}_i(t_{f,1})  &=  \dot{x}_i(s) + \int_s^{t_{f,1}} \ddot{x}_i(w)\ dw \\
   &\leq  \dot{x}_i(s) + a_{\textit{max},\dagger}(t_{f,1} - s), \quad s \leq t_{f,1},\ x_i \in \Pi_\dagger,
\end{aligned}
\end{equation}
where $\dagger = c$ if $\psi(i) = 0$ (car), and $\dagger = t$ if $\psi(i) = 1$ (truck). Thus there is a lower bound on the velocity of vehicle $i$ in the period just before $t_{f,1}$. Further, this lower bound is meaningful (i.e.\ non-negative) only during the time interval $[t_{f,1} - \frac{v_{\textit{max}}}{a_{\textit{max},\dagger}} , t_{f,1}] $.

Similarly, as every trajectory $x_i$ in $\Pi_c$ or $\Pi_{\textit{tr}}$ is continuous, we can write, using the bound in Equation~\eqref{eq:lower_bound_velo_homogeneous}, for $x_i \in \Pi_\dagger $ and $s$ such that $ t_{f,1} - \frac{v_{\textit{max}}}{a_{\textit{max},\dagger}} \leq s \leq t_{f,1}$:
\begin{align}
-v_{\textit{max}}\sum_{k=0}^{i-1} \tau_{\psi(k), \psi(k+1)}  =  x_i(t_{f,1}) =&\ x_i(s) + \int_s^{t_{f,1}} \dot{x}_i(w)\ dw\nonumber \\
  \geq &\ x_i(s) + \int_s^{t_{f,1}} \left[v_{\textit{max}} - a_{\textit{max},\dagger}(t_{f,1} - w)\right]\ dw,\no \\
 \implies  x_i(s) \leq&\ -v_{\textit{max}}\sum_{k=0}^{i-1} \tau_{\psi(k), \psi(k+1)} -v_{\textit{max}}(t_{f,1} - s) + \frac{1}{2}a_{\textit{max},\dagger}(t_{f,1} - s)^2 . \no
\end{align}

The upper bound for $x_i(\cdot)$ is actually attained by the trajectories described in Equations~\eqref{eq:homogeneous_stop},~\eqref{eq:homogeneous_no_stop} for the entirety of the time interval $[t_{\textit{acc},i}, t_{f,1}]$. This can be seen by first expressing the velocity of vehicle $i$ as follows: 
\begin{align*}
    v_{\textit{max}} = \dot{x}_i(t_{f,1}) =&\ \dot{x}_i(s) + a_{\textit{max},\dagger}(t_{f,1} - t), \quad s \in [t_{\textit{acc},i}, t_{f,1}],\\
    \implies \dot{x}_i(s) =&\ v_{\textit{max}} - a_{\textit{max},\dagger}(t_{f,1} - t),
\end{align*}
and then repeating the same procedure to obtain an expression for the position $x_i$.

Thus, our optimal trajectory ensures that each vehicle achieves the closest position to the intersection at every time instant in the acceleration phase.

\textbf{2. In the sub-interval $[t_{\textit{dec},i}, t_{\textit{acc},i})$:} In the final part of this proof, we consider the sub-interval $[t_{\textit{dec},i}, t_{\textit{acc},i})$. Now depending on the delay experienced by vehicle $i$, this sub-interval may be composed of just a deceleration phase (see the required condition in Equation~\eqref{eq:homogeneous_no_stop}), or a deceleration phase followed by a phase of rest (see Equation~\eqref{eq:homogeneous_stop}). 
\begin{figure}[!t]
    \centering
    \begin{tikzpicture}
        \node[anchor=south west,inner sep=0] (image) at (0,0) {\includegraphics[width=0.8\textwidth]{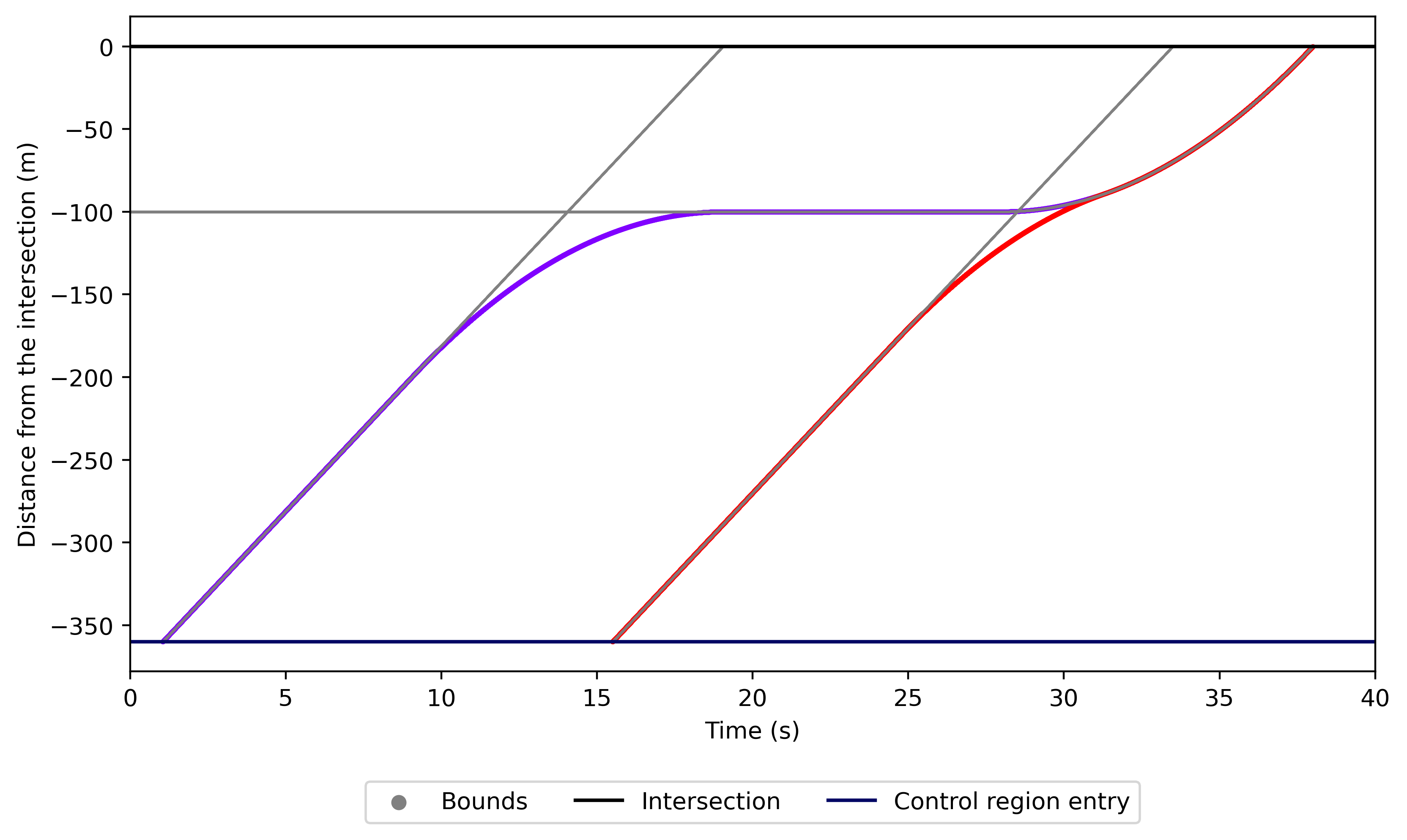}};
        \begin{scope}[x={(image.south east)},y={(image.north west)}]
            \fill (0.403, 0.747) circle[radius=1.8pt];
            \fill (0.723, 0.747) circle[radius=1.8pt];      
            \node at (0.6, 0.84) {Velocity not continuous};
            \draw[->] (0.48, 0.835)--(0.42,0.76);
            \draw[->] (0.71, 0.815)--(0.72,0.76);
        \end{scope}
    \end{tikzpicture}
    \caption{Optimal trajectory construction.}
    \label{fig:optimal_traj_construction}
\end{figure}
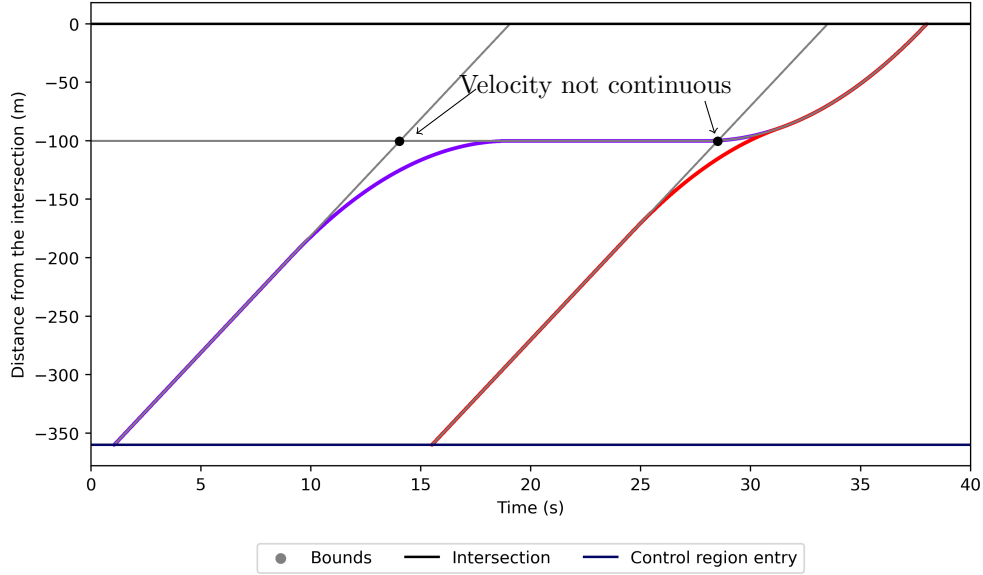

The construction of an optimal trajectory is demonstrated in Figure \ref{fig:optimal_traj_construction}. We have already determined the optimal trajectory around the start of the control region, as well as near the intersection (or the end of the control region). However, depending on the difference between the entering and crossing times of a vehicle, the middle part of the optimal trajectory changes. To see this, we have kept $t_{f,i}$ fixed for vehicle $i$ in Figure \ref{fig:optimal_traj_construction}, and we vary the entering time $t_{0,i}$ -- once to have a trajectory with a full stop in the control region, and once to have a trajectory with no stop. The shape of the optimal trajectory in the middle depends on where the optimal trajectory parts before the deceleration phase and from the acceleration phase onward intersect. At this point of intersection (marked in Figure \ref{fig:optimal_traj_construction}) however, the velocity function is not continuous. To be more precise, at such intersection points, the velocity jumps from the maximum permissible speed $v_{\textit{max}}$ to 0 if the trajectory was defined to be exactly as per the bounds obtained (in grey). Thus to maintain compatibility (i.e.\ belong to the classes $\Pi_c$ or $\Pi_{\textit{tr}}$), the trajectory must be `smoothed' sufficiently in the neighbourhood of this intersection point, which is where the deceleration (and an optional rest phase, depending on the delay of the vehicle) plays a role. In intuitive terms, it is clear that the optimal trajectory needs to start the deceleration phase as late as possible in order to stay as close to the intersection as possible. However, it is still not clear as to when is the latest time at which deceleration can start for the trajectory to be optimal. We claim that:
\begin{enumerate}
    \item if $t_{f,i} - t_{0,i} \geq \frac{x_0}{v_{\textit{max}}} + \frac{v_{\textit{max}}}{a_{\textit{max},\dagger}}$, choosing the start of deceleration process to be at the instant $\frac{x_0}{v_{\textit{max}}}-(t_{f,i}-t_{f,1}) + t_{0,i} - \frac{v_{\textit{max}}}{a_{\textit{max},\dagger}}$ along with decelerating at $-a_{\textit{max},\dagger}$ is the optimal decision, and
    \item if $\frac{x_0}{v_{\textit{max}}} \leq t_{f,i} - t_{0,i} < \frac{x_0}{v_{\textit{max}}} + \frac{v_{\textit{max}}}{a_{\textit{max},\dagger}}$, then starting the deceleration phase at time $t_{f,1} - 2\frac{v_{\textit{max}} - u}{a_{\textit{max},\dagger}}$ where $u$ is as in Equation~\eqref{eq:lowest_speed_homogeneous} along with decelerating at $-a_{\textit{max},\dagger}$ is optimal.
\end{enumerate}
Note that the claimed optimal start-of-deceleration time points exactly match those in the trajectories described before, in Equations~\eqref{eq:homogeneous_stop} and \eqref{eq:homogeneous_no_stop}.

Suppose that there exists a trajectory $x^\pi_i$ for vehicle $i$, that manages to get closer to the intersection at some time $t^\pi$ lying between $t_{\textit{dec},i}$ and $t_{\textit{acc},i}$. That is:
\allowdisplaybreaks
\begin{align*}
    x^\pi_i(t^\pi) >&\ x^*_i(t^\pi)\\
    \implies x^\pi_i(t_{\textit{dec},i}) + \int_{t_{\textit{dec},i}}^{t^\pi} \dot{x}^\pi_i(w) dw >&\ x^*_i(t_{\textit{dec},i}) + \int_{t_{\textit{dec},i}}^{t^\pi} \dot{x}^*_i(w) dw\\
    \implies \int_{t_{\textit{dec},i}}^{t^\pi} \dot{x}^\pi_i(w) - \dot{x}^*_i(w) dw >&\ x^*_i(t_{\textit{dec},i}) - x^\pi_i(t_{\textit{dec},i}) \geq 0,\\
    \implies \dot{x}^\pi_i(s) - \dot{x}^*_i(s) >&\ 0 \text{ for } s \in S \subseteq (t_{\textit{dec},i}, t_{\textit{acc},i}].
\end{align*}
That is, if vehicle $i$ attains a position closer to the intersection during the deceleration phase of $x^*_i$, then that implies that the velocity of vehicle $i$ under trajectory $x^\pi_i$ was higher than $x^*_i$ for some time therein. This could be the result of either decelerating later than optimal, or slower (on average). Consequently, vehicle $i$ attains a position under trajectory $x^\pi_i$ that is closer to the intersection than with trajectory $x^*_i$ at time $t_{\textit{acc},i}$, which means that:
\begin{enumerate}
    \item if vehicle $i$ is the head of a platoon, then it cannot satisfy its terminal condition for velocity, i.e.\ it cannot attain speed $v_{\textit{max}}$ when exiting the intersection, or
    \item if vehicle $i$ is not the head of the platoon, then covering a greater distance than under $x^*_i$ during the deceleration phase means that the safety constraint is violated for vehicles $i$ and $i-1$ (at least) at time $t_{f,1}$.
\end{enumerate}
In other words, there is no \emph{compatible} trajectory $x^\pi_i$ that can get closer to the intersection during the deceleration phase than $x^*_i$ and still satisfy the safety constraint, which means that $x^*_i$ is optimal in this sub-interval as well, thus completing the proof.
\end{proof}

Thus, we proved that if vehicle $i$ in a platoon is a truck, or if it is preceded by vehicles of the same type as itself, then the trajectories with the expressions \eqref{eq:homogeneous_stop} and \eqref{eq:homogeneous_no_stop} are optimal
in keeping the vehicle close to the intersection at all times in the control region. In fact, the arguments we used to prove optimality of the trajectory $x^*_i$ for the deceleration phase can be generalised:
\begin{corollary}\label{cor:optimality_deceleration_section}
Suppose, for a trajectory $x:[t_0, t_f] \to \R$, there exist two time points $t_1$ and $t_2$ in its domain with $t_1 < t_2$ such that the associated vehicle has a greater velocity at $t_1$ than at $t_2$ ($\dot{x}(t_1) > \dot{x}(t_2)$), and during the interval $(t_1,\ t_2)$, the velocity of the vehicle does not match $\dot{x}(t_1)$ and $\dot{x}(t_2)$. Further suppose that it can be shown that the optimal trajectory that keeps the vehicle closest to the intersection (given the trajectory of the preceding vehicle) coincides with $x$ for the interval $[t_0,\ t_1] \cup [t_2,\ t_f]$. Then, $x$ is the optimal trajectory if the associated vehicle decelerates maximally during the interval $[t_1, t_2)$, provided that the safety constraint is satisfied.
\end{corollary}
\begin{proof}
The proof follows from the same arguments used to show that the trajectory $x^*_i$ is optimal during the interval $[t_{\textit{dec},i},\ t_{\textit{acc},i}]$ in the proof of \thref{prop:optimal_homogeneous}. 
\end{proof}

Having established the optimality of the trajectories of trucks, it is very easy to show that the car trajectories described in Section~\ref{sec:SPA_closed_form} and Appendix~\ref{app:SPA_closed_form} are optimal. The following corollary establishes their optimality:
\begin{corollary}[Optimality part 2]\label{cor:optimal_inhomogeneous}
    Suppose a car holds the $i^{th}$ position in a platoon, and further suppose that there is at least one truck in the first $i-1$ vehicles in the same platoon, resulting in the trajectory of car $x^*_i$ described in either Section~\ref{sec:SPA_closed_form} or Appendix~\ref{app:SPA_closed_form}, depending on whether the last truck in the first $i-1$ vehicles makes a stop in the control region or not. Then, the trajectory $x^*_i$ keeps the corresponding car closest to the intersection at every instant of its journey throughout the control region, while satisfying the safety constraint w.r.t.\ the (optimal) trajectory $x^*_{i-1}$ of vehicle $i-1$. Thus:
    \begin{equation}\label{eq:sequential_opt_problem_strong_cfbt}
    \begin{aligned}
        x^\pi_i (s) \leq x^*_i(s) &\text{ for all } s \in [t_{0,i}, t_{f,i}], \text{ for every } x^\pi_i \in \begin{dcases}
            \Pi_c, &\text{ if } \psi(i) = 0,\\
            \Pi_{\textit{tr}}, &\text{ otherwise},
        \end{dcases}\\
        x^*_{i-1} (s) - y(s) \geq v_{\textit{max}} \tau_{\psi(i-1), \psi(i)}, &\text{ for all } s \in [t_{0,i}, t_{f,i-1}],\ y \in \{ x^\pi_i, x^*_i\} .
    \end{aligned}
    \end{equation}
\end{corollary} 
\begin{proof}
    Since we are dealing with a car which is preceded by a truck, we denote the position of the truck in the platoon by the index $j$, where $j \in \{0, \cdots, i-1\}$. Now, notice that the trajectory of the car can be divided broadly into two parts:
    \begin{enumerate}
        \item the trajectory of the car until the time it catches up with its closest preceding truck, i.e.\ in the interval $[t_{0,i}, T^*_i)$ (see the definition in \eqref{eq:t_star}), and
        \item the trajectory of the car after it has caught up with its closest preceding truck, i.e.\ during the interval $[T^*_i, t_{f,i}]$. 
    \end{enumerate}
    With this partition of the trajectory of the car, we can consider each part separately to see if the trajectories we described are optimal.

    Firstly, we consider the second part of the car's trajectory, i.e.\ in the interval $[T^*_i, t_{f,i}]$, after it has caught up with the closest preceding truck. In this time interval, the car follows the exact trajectory of the truck, and the distance between them is always the minimum distance required to satisfy the safety constraint, $v_{\textit{max}}\sum_{k=j}^{i-1}\tau_{\psi(k), \psi(k+1)}$. In other words, there is no possible trajectory that can perform better in terms of being close to the intersection at any instant in this time interval without violating the safety constraint. Since we already proved in \thref{prop:optimal_homogeneous} that our proposed trajectory for a truck is optimal, it follows that the proposed trajectory $x^*_i$ for the car is also optimal in this time interval. Note that this completely proves optimality of $x^*_i$ for the case when the car follows the truck exactly throughout the control region, i.e.\ when $T^*_i = t_{0,i}$. Thus for the next part, $T^*_i \neq t_{0,i}$.

    We now turn to the other part of the trajectory of the car, covering the duration before it has caught up with the truck. Here, we observe that, regardless of the trajectory of the truck and the relative difference between the delay of car and the truck, the car always has exactly two phases in this part of the trajectory: travelling at (the maximum permissible) constant speed $v_{\textit{max}}$ and thereafter decelerating at $-a_{\textit{max},c}$.
    
    We have established before that there is no other compatible trajectory with better performance than $x^*_i$ during the phase involving travelling at speed $v_{\textit{max}}$, i.e.\ during $[t_{0,i}, t_{\textit{dec},i}]$, which means that the trajectory $x^*_i$ is optimal for the interval $[t_{0,i}, t_{\textit{dec},i}]$ as well. So far, we have established that the car trajectory is optimal in the time period $[t_{0,i}, t_{\textit{dec},i}] \cup [T^*_i, t_{f,i}]$. We can directly apply \thref{cor:optimality_deceleration_section}, with $t_1 = t_{\textit{dec},i}$ and $t_2 = T^*_i$, to obtain that $x^*_i$ is also optimal for the (sub-)interval $(t_{\textit{dec},i},T^*_i)$, thus completing the proof.
    \end{proof}

Finally, we now show how the trajectories that are optimal for the sequential optimisation procedure are also optimal for the joint optimisation problem. 
\begin{proposition}[Joint optimisation]\label{prop:jointly_optimal}
Consider a platoon of $m$ vehicles, labelled from $1$ to $m$. Suppose that the trajectory of each of those vehicles (denoted by $x^*_i$) is obtained either from Section \ref{sec:SPA_closed_form} or from Appendix \ref{app:SPA_closed_form}, by following the appropriate case distinction, and by taking into account the type of the current vehicle as well as the type of preceding vehicles in its platoon. Then, those vehicle trajectories form a solution to the optimisation problem \eqref{eq:opt_problem_general_form}, with $f$ given by Equation~\eqref{eq:optimization_criterion_capacity}.
\end{proposition}
\begin{proof}
Suppose that some \emph{compatible} trajectories $\hat{x}_i$, $i \in \{1, \cdots, m\}$, form an optimal solution to the joint optimisation problem \eqref{eq:opt_problem_general_form}-\eqref{eq:optimization_criterion_capacity}. Then, we have the following relation:
\begin{equation}\label{eq:joint_opt_proof_step_0}
   \sum_{i=1}^m \int_{t_{0,i}}^{t_{f,i}} |\hat{x}_i(w)|\ dw \leq \sum_{i=1}^m \int_{t_{0,i}}^{t_{f,i}}  |x^*_i(w)|\ dw ,
\end{equation}
since the trajectories $x^*_i$ individually keep the vehicles closest to the intersection, but it is possible that the combined distances to the intersection may be minimised by a different set of trajectories, namely by $\hat{x}_i$. 

Now, we, know that the trajectory $x^*_1$ keeps the first vehicle in a platoon closest to the intersection at every instant while in the control region. Thus, it must follow that:
\begin{equation}
\begin{aligned}\label{eq:joint_opt_proof_step_1}
    x^*_1(s) &\geq \hat{x}_1(s), \text{ for all } s \in [t_{0,1},t_{f,1}],\\
    \implies \int_{t_{0,1}}^{t_{f,1}} |x^*_1(w)|\ dw &\leq  \int_{t_{0,1}}^{t_{f,1}} |\hat{x}_1(w)|\ dw,
\end{aligned}
\end{equation}
recalling that all trajectories have the range $[-x_0, 0]$, in the negative half of the real line.

Similarly, comparing the constraints on $x_2$ in the joint optimisation problem to the individual optimisation problem leads to the conclusion that the set of constraints is identical except for the safety constraint - for the individual problem \eqref{eq:miculescu_opt_prob} (similar to the problem formulated by \cite{Miculescu2020Polling-Systems-BasedSignals}) and \eqref{eq:optimality_strong_criterion_cot} or \eqref{eq:sequential_opt_problem_strong_cfbt} (depending on the situation), the safety condition is w.r.t.\ the optimal trajectory $x^*_1$ of the preceding vehicle, whereas for the joint optimisation problem \eqref{eq:opt_problem_general_form}-\eqref{eq:optimization_criterion_capacity}, the safety constraint is for a general trajectory $x_1$. It is easy to show that the trajectory $\hat{x}_2$ is also safe w.r.t.\ to the trajectory $x^*_1$:
\begin{equation*}
\begin{aligned}
    v_{\textit{max}}\tau_{\psi(1),\psi(2)} &\leq \hat{x}_1(s) - \hat{x}_2(s) \text{ for all } s \in [t_{0,2}, t_{f,1}],\\
    \implies v_{\textit{max}}\tau_{\psi(1),\psi(2)} + \hat{x}_2(s) &\leq \hat{x}_1(s) \leq x^*_1(s) \text{ (from \eqref{eq:joint_opt_proof_step_1})} \\
    \implies  v_{\textit{max}}\tau_{\psi(1),\psi(2)} &\leq x^*_1(s) - \hat{x}_2(s) \text{ for all } s \in [t_{0,2}, t_{f,1}].
\end{aligned}
\end{equation*}
The above statement also implies that $\hat{x}_2$ is a \emph{feasible} solution to the individual optimisation problem of the second vehicle, and hence we can write, similarly to \eqref{eq:joint_opt_proof_step_1}:
\begin{equation*}
    \int_{t_{0,2}}^{t_{f,2}} |x^*_2(w)|\ dw \leq  \int_{t_{0,2}}^{t_{f,2}} |\hat{x}_2(w)|\ dw.
\end{equation*}

By repeating the argument used for the trajectory of the second vehicle for all the vehicles that follow, we can write, for vehicle $i$ in the platoon:
\begin{equation}\label{eq:joint_opt_proof_step_2}
    \int_{t_{0,i}}^{t_{f,i}} |x^*_i(w)|\ dw \leq  \int_{t_{0,i}}^{t_{f,i}} |\hat{x}_i(w)|\ dw.
\end{equation}
By combining the inequalities in \eqref{eq:joint_opt_proof_step_1}-\eqref{eq:joint_opt_proof_step_2}, we get:
\begin{equation*}
     \sum_{i=1}^m \int_{t_{0,i}}^{t_{f,i}} |x^*_i(w)|\ dw \leq \sum_{i=1}^m \int_{t_{0,i}}^{t_{f,i}}  |\hat{x}_i(w)|\ dw .
\end{equation*}
The above inequality, in conjunction with the relation \eqref{eq:joint_opt_proof_step_0}, implies:
\begin{equation}\label{eq:joint_opt_proof_1A_step_3}
    \sum_{i=1}^m \int_{t_{0,i}}^{t_{f,i}} |x^*_i(w)|\ dw = \sum_{i=1}^m \int_{t_{0,i}}^{t_{f,i}}  |\hat{x}_i(w)|\ dw .
\end{equation}
In other words, the trajectories $x^*_i$, $i \in \{1, \cdots, m\}$ also form an optimal solution to the joint optimisation problem \eqref{eq:opt_problem_general_form}-\eqref{eq:optimization_criterion_capacity}.
\end{proof}

\section{Time separations}\label{app:time_separations}

Safety plays an important role in our framework - it is ensured via service times in the polling model used to schedule crossing times and through constraints in the joint optimisation problem. A unique feature of our framework is that the time separation, which is used to determine the service times and the safety constraints, depends on the vehicles involved. This makes the framework more flexible and scalable to greater heterogeneity in traffic. Here we briefly conceptualise how to compute the safety distance and time separation in various situations.

\subsection{Time separation for vehicles in the same lane}
According to \thref{def:safety}, the minimum required distance between two consecutive vehicles on the same lane depends on the maximum allowed speed and the corresponding time separation. To compute this separation, it is necessary to consider the worst-case scenario. Consider two vehicles - vehicle A of type $i$ followed by vehicle B of type $j$, both at maximum allowed speed $v_{\textit{max}}$ in the same lane, and separated by a distance of $v_{\textit{max}} \tau^{\ell_i, \ell_i}_{i, j}$. Now suppose that vehicle A starts decelerating suddenly with delayed/no communication. In the worst-case scenario, vehicle A would decelerate maximally, to a complete stop. Then, to avoid a collision, vehicle B would also need to decelerate after a brief response time $t_{\textit{res}}$, so that it comes to a complete stop $d_i + \delta$ units of distance away from the front of vehicle A, where $d_i$ is the length of vehicle A and $\delta~(>0)$ is a tolerance parameter that can be tuned depending on the vehicles involved. Thus, the time separation would need to be:
\begin{equation}\label{eq:time_separation_same_lane}
    \tau^{\ell_i, \ell_i}_{i, j} = t_{\textit{res}} + \frac{d_i + \delta}{v_{\textit{max}}} + \max\left[0, \frac{v_{\textit{max}}}{2} \left( \frac{1}{a_i} - \frac{1}{a_{j}} \right) \right],
\end{equation}
where $a_i$ and $a_j$ are the maximum feasible decelerations for vehicles A and B, given by:
\begin{equation*}
    a_k = 
    \begin{dcases}
        -a_{\textit{max},c}, \quad \text{ if vehicle of type $k$ is a car,} \\
        -a_{\textit{max},tr}, \quad \text{ if vehicle of type $k$ is a truck}.
    \end{dcases}
\end{equation*}

From Equation~\eqref{eq:time_separation_same_lane}, we see that the time separation is always at least the response time $t_{\textit{res}}$ plus the time required to cover the distance $d_i + \delta$ at speed $v_{\textit{max}}$. Additionally, depending on the type of vehicles, we have an extra term depending on the decelerations. If the two vehicles involved are of the same type, then that term vanishes. Also, if the following vehicle decelerates quicker than the preceding vehicle (i.e., when a truck is followed by a car), then again the last term does not contribute to the time separation. In fact, in this framework, the last term only contributes for the situation when a car is followed by a truck - the truck decelerates slower than the car, and so during the deceleration process, the truck will be at higher speeds for longer than the car, thus covering more distance than the car does. To avoid the possibility of a collision, the time separation is larger in this case, as can be seen in Table~\ref{tab:time_separation_same_lane}.

For our framework, the response time $t_{\textit{res}}$ is set to 0.5 seconds, which is the commonly assumed estimate for autonomous vehicles \cite{Feng2019}, and with the length of a car to be $5$ m and the length of a truck to be $10$ m, $v_{\textit{max}} = 20$ m/s, $\delta = 1$ m, $-a_{\textit{max},c} = -4$ m/s$^2$ and $-a_{\textit{max},tr} = -2$ m/s$^2$, we get:
\begin{table}[h!]
\centering
\begin{tabular}{c c *{2}{p{5mm}}}
&                            & \multicolumn{2}{p{27mm}}{{\centering Preceding Vehicle}}            \\ \cline{3-4} 
& \multicolumn{1}{c|}{}      & \multicolumn{1}{p{10mm}|}{Car} & \multicolumn{1}{p{10mm}|}{Truck} \\ \cline{2-4} 
\multicolumn{1}{c|}{\multirow{2}{*}{\parbox{1.6cm}{\centering Following Vehicle}}} & \multicolumn{1}{c|}{Car}   & \multicolumn{1}{p{5mm}|}{0.8}    & \multicolumn{1}{p{5mm}|}{1.05}      \\ \cline{2-4} 
\multicolumn{1}{c|}{}  & \multicolumn{1}{c|}{Truck} & \multicolumn{1}{p{5mm}|}{3.3} & \multicolumn{1}{p{5mm}|}{1.05}  \\ \cline{2-4} 
\end{tabular}
\caption{Time separation in various cases, same lane (in seconds).}
\label{tab:time_separation_same_lane}
\end{table}

\subsection{Time separation between vehicles travelling in different lanes}
In order to calculate the time separation between vehicles in different lanes, we refer to \textcite{Tachet2016RevisitingSystems}. This separation can be used to compute how far the first vehicle from a new platoon should be from the intersection area when the last vehicle from the current platoon is just starting to cross. The formula from \cite{Tachet2016RevisitingSystems} for the separation $\tau^{\ell_i, \ell_j}_{i,j}$ between preceding vehicle $i$ and following vehicle $j$ is slightly adapted to accommodate the heterogeneity of traffic, and is presented here below:
\begin{equation}\label{eq:time_separation_diff_lane}
   \tau^{\ell_i, \ell_j}_{i, j} = t_{\textit{res}} - \frac{v_{\textit{max}}}{2a_j} + \frac{s + d_i}{v_{\textit{max}}}, \quad \ell_i \neq \ell_j,
\end{equation}
where $d_i$ is the length of vehicle $i$ and $a_j$ is the maximum feasible deceleration for vehicle $j$ (negative, as before). Clearance times are implicitly absorbed into this time separation, in the last term $(\frac{s}{v_{\textit{max}}})$, which depends on the width $s$ of the intersection. The time separation values in Table~\ref{tab:time_separation_diff_lane} are computed assuming the lengths for vehicles as before, and the width of an intersection to be $8$ m. 
\begin{table}[h!]
\centering
\begin{tabular}{c c *{2}{p{5mm}}}
&                            & \multicolumn{2}{p{27mm}}{{\centering Preceding Vehicle}}            \\ \cline{3-4} 
& \multicolumn{1}{c|}{}      & \multicolumn{1}{p{10mm}|}{Car} & \multicolumn{1}{p{10mm}|}{Truck} \\ \cline{2-4} 
\multicolumn{1}{c|}{\multirow{2}{*}{\parbox{1.6cm}{\centering Following Vehicle}}} & \multicolumn{1}{c|}{Car}   & \multicolumn{1}{p{5mm}|}{3.65}    & \multicolumn{1}{p{5mm}|}{3.9}      \\ \cline{2-4} 
\multicolumn{1}{c|}{}  & \multicolumn{1}{c|}{Truck} & \multicolumn{1}{p{5mm}|}{6.15} & \multicolumn{1}{p{5mm}|}{6.4}  \\ \cline{2-4} 
\end{tabular}
\caption{Time separation in various cases, different lanes (in seconds).}
\label{tab:time_separation_diff_lane}
\end{table}

Note that the right-hand-side of Equation~\eqref{eq:time_separation_diff_lane} does not depend on the lanes in which vehicles $i$ and $j$ are travelling, a choice made here for convenience and simplicity. The platoon-forming algorithm is already implemented to handle lane-dependent time separations, and its speed-profiling counterpart can easily be extended to do so. 

A final remark concerns the concept behind time separations. As discussed before, a time separation between two vehicles is computed so as to allow the following vehicle to come to a complete stop in case of an unsafe situation. A possible extension could be to include some sort of a swerving manoeuvre to avoid obstacles/ unsafe situations, as discussed in \cite{Urmson2006}.

\end{document}